\documentclass[12pt]{article}

\usepackage[T1]{fontenc}
\usepackage{amsmath}
\usepackage{amsthm}
\usepackage{amssymb}
\PassOptionsToPackage{normalem}{ulem}
\usepackage{ulem}
\usepackage{mathtools}
\usepackage{amsfonts}
\usepackage[varg]{txfonts}
\usepackage{graphicx}
\usepackage{caption}
\usepackage{subcaption}
\usepackage[top=1in, bottom=1in, left=1in, right=1in]{geometry}
\usepackage{setspace}
\usepackage[toc,page]{appendix}
\usepackage{bbm} 
\usepackage{enumitem}
\usepackage{xurl}
\usepackage{verbatim}
\usepackage{tikz}
\usepackage{chngcntr}
\usepackage{apptools}
\usetikzlibrary{shapes.multipart}
\usetikzlibrary{decorations.pathreplacing}
\usetikzlibrary{decorations.markings}
\usetikzlibrary{arrows}
\usetikzlibrary{calc,intersections,through,backgrounds,matrix,patterns}
\usetikzlibrary{patterns}
\usepackage{pgfplots}
\usepackage{subcaption}
\usepackage{adjustbox}
\usepackage{xcolor}

\usepgfplotslibrary{fillbetween,decorations.softclip}
\pgfdeclarelayer{bg}
\pgfsetlayers{bg,main}
\pgfplotsset{width=10cm,compat=1.9}
\pgfkeys{%
  /tikz/on layer/.code={
    \pgfonlayer{#1}\begingroup
    \aftergroup\endpgfonlayer
    \aftergroup\endgroup
  }
}
\pgfplotsset{
    highlight/.code args={#1:#2}{
        \fill [every highlight] ({axis cs:#1,0}|-{rel axis cs:0,0}) rectangle ({axis cs:#2,0}|-{rel axis cs:0,1});
    },
    /tikz/every highlight/.style={
        on layer=\pgfkeysvalueof{/pgfplots/highlight layer},
        blue!20
    },
    /tikz/highlight style/.style={
        /tikz/every highlight/.append style=#1
    },
    highlight layer/.initial=axis background
    standard/.style={
        axis x line=middle,
        axis y line=middle,
        enlarge x limits=0.15,
        enlarge y limits=0.15,
        every axis x label/.style={at={(current axis.right of origin)},anchor=north west},
        every axis y label/.style={at={(current axis.above origin)},anchor=north east}
    }
}
\pgfplotsset{every major tick/.append style={ultra thick}}

\tikzset{
    hatch distance/.store in=\hatchdistance,
    hatch distance=10pt,
    hatch thickness/.store in=\hatchthickness,
    hatch thickness=0.3pt
}
\AtAppendix{\counterwithin{equation}{section}}
\newtheorem{theorem}{Theorem}

\newtheorem{corollary}{Corollary}
\newtheorem{lemma}{Lemma}
\newtheorem{assumption}{Assumption}
\newtheorem{proposition}{Proposition}

\DeclareMathOperator*{\argmax}{arg\,max}
\DeclareMathOperator*{\argmin}{arg\,min}

\DeclareMathOperator*{\epislon}{\epsilon}

\DeclareMathOperator*{\supp}{\text{supp}}

\DeclareMathOperator*{\e}{\epsilon}

\DeclareMathOperator*{\righatrrow}{\rightarrow}

\usepackage[unicode=true,
bookmarks=true,bookmarksnumbered=false,bookmarksopen=true,bookmarksopenlevel=1,
breaklinks=false,pdfborder={0 0 1},backref=false,colorlinks=true,linkcolor=blue,citecolor=blue]{hyperref}

\title{
    Competitive Information Disclosure with Heterogeneous Consumer Search\thanks{
We are grateful to Dan Bernhardt, Dongkyu Chang, Andy Bongjune Choi, Theo Durandard, Jack Fanning, Seung-Hyun Hong, Myeongeon Koh, Youngwoo Koh, Jinwoo Kim, Kyungmin (Teddy) Kim, Mingeon Kim, Chanjoo Lee,  
Jorge Lemus, Keeyoung Rhee, Lena Song, Can Urgun  for their valuable comments and suggestions.
We have benefitted from comments by conference participants at the 
26th ACM Conference on Economics and Computation,
36th Stony Brook International Conference on Game Theory
and the seminar participants of 
Seoul National University, University of Illinois Urbana-Champaign, University of Tokyo
for various helpful and constructive comments. 
This research was supported by the BK21 FOUR (Fostering Outstanding Universities for Research) funded by the Ministry of Education(MOE, Korea) and National Research Foundation of Korea(NRF).
}}
\author{Dongjin Hwang\thanks{Department of Economics, Seoul National University; \texttt{djsteve@snu.ac.kr} (first author)}
\;\;\; Ilwoo Hwang\thanks{Department of Economics, Seoul National University; \texttt{ilwoo.hwang@snu.ac.kr} (corresponding author)}
\\ \small Seoul National University
}
\date{\today}

\usepackage[natbib=true,minnames=1,style=apa]{biblatex}  
\DeclareFieldFormat{citehyperref}{%
  \DeclareFieldAlias{bibhyperref}{noformat}
  \bibhyperref{#1}}

\DeclareFieldFormat{textcitehyperref}{%
  \DeclareFieldAlias{bibhyperref}{noformat}
  \bibhyperref{%
    #1%
    \ifbool{cbx:parens}
      {\bibcloseparen\global\boolfalse{cbx:parens}}
      {}}}

\savebibmacro{cite}
\savebibmacro{textcite}

\renewbibmacro*{cite}{%
  \printtext[citehyperref]{%
    \restorebibmacro{cite}%
    \usebibmacro{cite}}}

\renewbibmacro*{textcite}{%
  \ifboolexpr{
    ( not test {\iffieldundef{prenote}} and
      test {\ifnumequal{\value{citecount}}{1}} )
    or
    ( not test {\iffieldundef{postnote}} and
      test {\ifnumequal{\value{citecount}}{\value{citetotal}}} )
  }
    {\DeclareFieldAlias{textcitehyperref}{noformat}}
    {}%
  \printtext[textcitehyperref]{%
    \restorebibmacro{textcite}%
    \usebibmacro{textcite}}}

\begin{document}
\maketitle

\begin{abstract}
    We study a model of competitive information design in an oligopoly search market with heterogeneous consumer search costs. A unique class of equilibria\textemdash upper-censorship equilibria\textemdash emerges under intense competition. In equilibrium, firms balance competitive pressure with local monopoly power granted by search frictions. Notably, firms disclose only partial information even as the number of firms approaches infinity. The maximal informativeness of equilibrium decreases under first-order shifts in the search cost distribution, but varies non-monotonically under mean-preserving spreads. The model converges to the full-disclosure benchmark as search frictions vanish, and to the no-disclosure benchmark as search costs become homogeneous.
        Moreover, we show that the well-known discontinuities of equilibrium with respect to the search cost distribution in the price search literature carry over to information competition.

    \smallskip
    
\end{abstract}

\section{Introduction}
Consumers actively navigate markets to make informed purchase decisions. However, the information they encounter about products is rarely exogenous. Rather, it is strategically curated by firms through channels such as advertisements, product reviews, and free trials. By controlling the information content, firms shape consumers' perceptions and influence their decisions, ultimately aiming to maximize their chances of being chosen over competitors.

In many markets, consumers inherently differ in their ability or willingness to search due to factors such as technological proficiency, market experience, time constraints, or cognitive capacity.\footnote{
    The heterogeneity of consumer search costs is well-documented
    across various markets, including
    mutual funds markets \parencite{hortaccsu2004product},
    e-commerce \parencite{de2018consumer}
    local gasoline markets \parencite{nishida2018determinants},
    mortgage market \parencite{allen2014effect,allen2019search},
    or the trade-waste market \parencite{salz2022intermediation}.
}
These differences critically affect consumers' willingness to explore available options. Since firms cannot observe individual consumers' search costs, they must design their information disclosure strategies in response to the overall distribution of search costs while accounting for competitive pressures. These observations raise key questions: To what extent does competition encourage information disclosure in the presence of heterogeneous search frictions? How is equilibrium informativeness tied to the distribution of search costs, and how do firms' incentives change as this distribution varies?

To address these questions, we study a model of competitive information disclosure in an oligopoly search market with $n$ horizontally differentiated firms. Consumers differ in their search costs, represented by a continuous distribution $H$ over $[\underline{c},\overline{c}]$. They are initially uninformed of their match values for firms' products, which are independently drawn from a common distribution $F$. Firms strategically commit to signal structures that convey information about these match values.  Importantly, consumers must incur a search cost to randomly explore an unvisited firm, after which they observe \textit{both} the firm's signal structure and the realized signal. Search is with perfect recall.

As in traditional search models, consumers follow reservation value strategies: stopping their search when the posterior mean of the firm-specific match value exceeds their reservation value and continuing otherwise. If consumers have explored all firms, they purchase from the firm with the highest posterior mean. This behavior creates two opposing incentives for firms. On one hand, competition drives firms to disclose more information to increase the likelihood of being revisited after consumers explore all firms. On the other hand, search frictions provide firms local monopoly power, incentivizing them to pool information just above the reservation value to deter further search. In the absence of search frictions, competitive pressure dominates, leading firms to fully disclose information as the market becomes more competitive---resembling the classical Bertrand outcome \parencite{hwang2023}. By contrast, if search costs are identical for all consumers, the hold-up problem prevails, yielding no disclosure as the unique equilibrium \parencite{au2024attraction}, reminiscent of the Diamond paradox \parencite{diamond}.

Our first main result demonstrates that heterogeneous search costs yield equilibrium with \textit{partial information disclosure}. Specifically, firms use an ``upper-censorship'' strategy $U_a$: all match values above the censorship threshold $a$ are pooled into a single atom, while values below $a$ are fully revealed. Notably, these upper-censorship strategies persist as an equilibrium even in the limit of intense competition (i.e., $n\to\infty$).

This equilibrium structure arises from firms' attempts to balance competitive pressure with search frictions. Firms exploit their local monopoly power by concealing information at the top, while yielding to competitive forces at the bottom. In this equilibrium, consumers' reservation values are dispersed over an interval $[\underline{r},\overline{r}]$. Given this, firms strategically choose the censorship threshold so that the pooled signal precisely matches the highest reservation value, prompting immediate purchase by all consumer types. Typically, the censorship threshold $a$ falls below the lowest reservation value $\underline{r}$, prompting all consumer types to continue searching. As the market becomes essentially frictionless for signals below $\underline{r}$, competitive forces dominate, leading to full disclosure at the bottom.


We establish that a continuum of upper-censorship equilibria exists---always including equilibrium with no disclosure---and it constitutes the unique class of equilibria in the intense-competition limit. All equilibria are ranked by (Blackwell) informativeness, i.e., those with higher censorship thresholds are more informative. We focus on the equilibrium with the highest censorship threshold, which we call  the \textit{maximally informative equilibrium}.\footnote{While this is the maximally informative upper-censorship equilibrium, it becomes the maximally informative equilibrium overall in the limit of intense competition.} This equilibrium not only maximizes information disclosure, but is also the socially efficient and consumer-optimal equilibrium.

To characterize the maximally informative equilibrium, we introduce a simple measure based on the distribution of search costs $H$: maximal informativeness strictly increases with the minimum average density of $H$, defined as $\min_{c}(H(c)/c)$. Roughly speaking, this measure is high when search costs are evenly distributed. To understand the intuition, suppose that $H$ is uneven and concentrated at a specific point. Then firms can profitably deviate by targeting only high-cost consumers (those at or above the concentrated point), since losing low-cost consumers has minimal impact due to the distribution's imbalance.

Using this insight, we establish our second main result, which examines how changes in $H$ affect equilibrium informativeness. We analyze two dimensions: (1) shifts in the overall level of search costs (first-order stochastically dominating shifts) and (2) changes in search cost heterogeneity (second-order shifts). Not surprisingly, maximal equilibrium informativeness decreases with first-order shifts, as higher search costs increase firms' incentives to exploit local monopoly power.

More intriguingly, we find that equilibrium informativeness is non-monotonic with respect to mean-preserving spreads (MPS).  This result stems from the observation that informativeness increases with the ``evenness'' of the distribution. Clearly, the impact of an MPS shift on the evenness is non-monotonic. For instance, if $H$ is concentrated around midpoints (i.e. has a quasi-concave density), a mean-preserving spread makes the distribution more even. However, if $H$ is already polarized (i.e. has a quasi-convex density), a mean-preserving spread leads to a more uneven distribution.

Furthermore, we establish convergence results to the benchmark cases and discuss extensions. As $H$ approaches the frictionless case, the maximally informative equilibrium converges to full disclosure. Conversely, when $H$ approaches the homogeneous search friction case, no disclosure emerges as the unique equilibrium. Additionally, we show that our main results remain qualitatively robust for small markets ($n$ finite) under mild assumptions about the match value distribution $F$ and the search cost distribution $H$.

Finally, we relate our findings to the classical price search literature \parencite{stahl1989oligopolistic, stahl1996oligopolistic}. When firms compete in prices, equilibrium outcomes hinge on the \textit{local} density $h(0)$ of consumers with zero search costs. In contrast, informativeness in equilibrium information disclosure depends on the \textit{global} shape of the search cost distribution—specifically, its evenness. Also, we demonstrate that the discontinuities in equilibrium behavior observed in price search models, as the search cost distribution varies, also arise in the competitive information disclosure.

\subsection{Related Literature}

We contribute to the growing body of literature on competitive information design. Incorporating heterogeneous search costs, we bridge the two well-studied extremes---markets with no search  friction (\cite{ivanov2013information}; \cite{au2020competitive,hwang2023}) and homogeneous search friction \citep{au2023attraction,he2023competitive}. 
We offer a complete characterization of equilibrium in the limit of intense competition, which entails partial information disclosure balancing competitive forces and search frictions.\footnote{
    In many single-sender information design problems, upper censorship distributions arise as the optimal information structure (e.g., \textcite{kolotilin2022censorship} and \textcite{arieli2023optimal}). Ours is the first to establish such a rule as an equilibrium with multiple senders.} 
Moreover, we establish a continuity result, demonstrating that the maximally informative equilibrium converges to full information (resp. no information) as the cost distribution $H$ approaches the no-friction case (resp. the homogeneous friction case).\footnote{
    Partial convergence result to full information have also been established in \textcite{au2023attraction} as search cost vanishes. In the context of pricing, the seminal work of \textcite{stahl1989oligopolistic} shows similar convergence to the Diamond paradox and Bertrand pricing.}

 \citet{au2023attraction} consider a setting with homogeneous search costs where consumers can observe information structures before searching and direct their search accordingly. They show that firms commit to some level of information disclosure to improve their ranking in the search order. This `attraction motive' disappears when information structures are not observable prior to searching \citep{au2024attraction} or when search is random \citep{he2023competitive}. Our findings complement this literature by demonstrating that search cost heterogeneity alone can induce equilibrium information disclosure---even in the absence of attraction motives.

We draw a close parallel between our results and the classical price search literature \parencite{moraga2017prices,anderson2006advertising}, particularly \textcite{stahl1996oligopolistic}. 
Like our model, \textcite{stahl1996oligopolistic} assumes an atomless, continuous distribution of search costs and admits a continuum of symmetric pure-strategy equilibria---including one that exhibits the Diamond paradox. 
However, the forces sustaining the consumer-optimal equilibrium differ substantially. 
In price competition, equilibrium outcomes hinge on the local density of shoppers, $h(0)$, which determines the loss of market share from a small price increase. In our setting, in contrast, the maximal informativeness of equilibria depends on the global shape of the search cost distribution, namely the minimum average density. This reflects the key feasibility constraint in information design: deviations must reallocate signals in a Bayes-plausible way.

In a concurrent and independent paper, \textcite{boleslavsky2023information} analyzes a related problem as ours but under binary search costs with a postive mass of ``shoppers” $(H(0)>0$), in the spirit of \textcite{stahl1989oligopolistic}. We view their work as complementary to ours: although the settings appear similar, the results diverge sharply, offering a richer understanding of how consumer heterogeneity affects informational competition. They show that the presence of a mass of such shoppers sustains competition even at high signals, and under sufficiently intense competition, the equilibrium features no information disclosure below the reservation value. In contrast, our model admits continuum of equilibria in which firms fully disclose information below the reservation value and pool at the top. In Section~\ref{section:extensions}, we identify two features of $H$ that critically shape this divergence in equilibrium structure: (i) the presence of an atom of shoppers and (ii) the existence of a ``gap” in the support of $H$ around zero. Each feature induces sharp discontinuity in disclosure behavior: upper-censorship equilibria (no atom and no gap; Theorem~\ref{thm:upper}), no disclosure at the bottom (atom and a gap; \textcite{boleslavsky2023information}), and full disclosure (atom and no gap; Proposition~\ref{proposition:atom_full_info}). This discontinuity pararells that of the price search literature, where the same features determine whether equilibrium takes the form of pure strategies, mixed strategies, or Bertrand pricing, respectively.

Our comparative statics result relates to several recent papers that establish \textit{prior-independent} comparative statics in information design \parencite{kolotilin2022censorship,gitmez2022informational,curello2022comparative}, extending this literature to settings where interim payoffs arise endogenously from strategic interactions.\footnote{
    Existing studies have primarily focused on identifying conditions under which, for arbitrary prior distributions, specific shifts in exogenously given interim payoff functions lead to a more (or not less) informative optimal signal structure.}
While maximal informativeness of our equilibrium depends on the prior distribution, its response to shifts in the search cost distribution $H$ remains prior-independent. We also introduce a simple graphical method to analyze these shifts.

Finally, we situate our work within the broader literature on information design \parencite{kamenica2011bayesian,dworczak2019simple}
in various search markets. Recent papers have explored related themes, such as the joint design of information and pricing by a platform \parencite{dogan2022consumer} or competing firms \parencite{au2024attraction,boleslavsky2023information}, 
characterizing all implementable search behaviors in Pandora's problem under some information structure \parencite{sato2025feasible},
optimal disclosure for search goods \parencite{lyu2023information,choi2019optimal}, information disclosure regarding common values under coordination \parencite{board2018competitive}, and when information structure itself constitutes the product being sold \parencite{mekonnen2023persuaded}.

The remainder of the paper is organized as follows. Section~\ref{section:model} introduces the model. Section~\ref{section:preliminary} provides preliminary analysis and establishes two benchmark cases. Section~\ref{section:equilibrium} characterizes the set of equilibria. Section~\ref{section:comparative} provides comparative static analyses of the maximally informative equilibrium under various shifts in $H$, along with  convergence results to the benchmark cases. Section~\ref{section:extensions} discusses extension and relate to the price search literature. Section~\ref{section:conclusion} concludes. The Appendix contains omitted proofs and detailed analysis of the model extensions.

\section{Model}
\label{section:model}

We consider a search market with $n$ horizontally differentiated firms and a continuum of risk-neutral consumers with unit demand. Firms, indexed by $i$, are ex-ante identical. Each consumer's match value $v_i$ for firm $i$ is independently and identically distributed across consumers and firms according to a distribution $F$ on [0,1]. We assume that $F$ is twice continuously differentiable, admits a positive density $f$, and that $f$ and  $f'$ are bounded.
Each firm earns a unit payoff if a consumer purchases its product and zero otherwise.


Consumers differ in their search cost $c$, which is independently and identically distributed according to a distribution $H$ on $[\underline{c}, \overline{c}]$,  where $\overline{c}>\underline{c}\geq0$. By the exact law of large numbers \parencite{sun2006exact}, the distribution function $H(c)$ also represents the exact fraction of consumers whose search cost is no greater than $c$.\footnote{The exact law of large numbers ensures that aggregate consumer behavior corresponds to the distribution $H$, despite the stochastic nature of individual behavior.} We refer to a consumer with search cost $c$ as a \textit{$c$-consumer} hereafter.  As is standard in the literature, we focus on the case where $\overline{c} < \mu:=E_F[v_i]$, which ensures that all consumers search at least once. We assume $H$ is twice continuously differentiable and admits a strictly positive density $h$, with at most finitely many local extrema.\footnote{
    Formally, there exists a finite partition $0=x_0\leq \dots \leq x_K=1$  such that $h$ is monotone in each $[x_{k-1}, x_k]$ for all $i=1,\dots, K$. Our results extend to the case where the continuity assumption is relaxed, allowing $H$ to have at most finitely many interior jumps, while being twice continuously differentiable elsewhere.
    }

Each firm $i$ strategically provides information about the match value $v_i$. In doing so, they have full flexibility in choosing their information structure $\pi_i: [0,1]\rightarrow \Delta(S_i)$.  Consumers do not observe the match value $v_i$ directly but instead receive an informative signal $s_i\in S_i$ about $v_i$.  Hence, the risk-neutral consumer's purchasing decision depends solely on the posterior mean $x_i = \mathbb{E}[v_i|s_i]$. 

Each information structure $(\pi_i, S_i)$ induces a distribution of posterior means $G_i$, which is a mean-preserving contraction (MPC) of the prior distribution $F$.\footnote{ 
    $G$ is a MPC of $F$ if and only if  $\int_0^1tdF(t)=\int_0^1tdG(t)$ and $\int_0^xG(t)dt \leq \int_0^xF(t)dt$ for all $x\in[0,1]$.} It is well known that a distribution $G_i$ of posterior means can be induced by some information structure if and only if $G_i$ is an MPC of $F$.  This result allows us to work directly with the distribution $G_i$, rather than considering the full set of feasible information structures. Formally, let $\text{MPC}(F)$ denote the set of all mean-preserving contractions of $F$; we say that $G_i$ is \emph{feasible} if $G_i \in \text{MPC}(F)$. We refer to $G_i$ as an \textit{information structure}, and to the induced posterior mean $x_i = \mathbb{E}[v_i|s_i]$ as a \textit{signal}. Note that full disclosure of the match value corresponds to $G = F$,  while no disclosure corresponds to $G=\delta_{\mu}$, a degenerate distribution at the prior mean $\mu$.

The game proceeds as follows.  Before the search process begins, each firm $i$ simultaneously commits to an information structure $G_i \in \text{MPC}(F)$, and each consumer's search cost $c$ is privately realized. The chosen information structures $(G_i)_{i=1}^n$ are not observable  to consumers at the outset of the search process. At each stage of the search process, a consumer decides whether to continue searching or stop.\footnote{We assume consumers with $c=0$ use the weakly dominant strategy of visiting all firms.} If they choose to search, they incur the search cost and randomly (with equal probability) visit one of the remaining unvisited firms. Upon visiting firm $i$, the consumer observes \textit{both} the information structure $G_i$ and their privately realized signal $x_i$. If they choose to stop, they may purchase from any of the previously visited firms with perfect recall.\footnote{Our model assumes that (1) information structures are unobservable prior to visiting a firm, and (2) search is random. All our results continue to hold when search is directed, as long as information remains unobservable until visit. Alternatively, if information structures are publicly observable before search but search remains random, we conjecture that our results still hold whenever an equilibrium exists. As noted in related literature, observable information structures combined with directed search introduce an 'attraction motive' \parencite{au2023attraction}, which renders equilibrium analysis prohibitively complex.}

\bigskip

Given the ex-ante symmetry of the environment, we focus on symmetric equilibria.\footnote{Characterizing asymmetric equilibria in the competitive information disclosure game is challenging, even in the absence of search costs. \textcite{du2024competitive} show that even with no search friction, no asymmetric equilibrium exists when there are two firms.} Let $(\widetilde{G_1}, \dots, \widetilde{G_n})$ be the consumers' conjecture about the firms' strategies prior to visiting them. A profile of feasible distributions $(G^*, \dots, G^*)$ is a \textit{(symmetric) equilibrium} if: (i) each firm's strategy $G_i=G^*$ is optimal given the behavior of the other firms and consumers; (ii) the consumers' search rule is optimal given their conjectures $(\widetilde{G_1}, \dots, \widetilde{G_n})$; and (iii) the consumers' conjectures are correct on-path, i.e., $\widetilde{G}_i=G^*$ for all $i$.  We refer to $G^*$ as the \textit{equilibrium distribution}.

We denote $D_i(x_i; G)$ as the \emph{interim demand} of firm $i$, defined as the probability that a consumer purchases from firm $i$ when all competing firms choose $G$, the consumer holds the conjecture $\widetilde{G_i} = G$ for all $i$, and observes a signal $x_i$ for firm $i$.  Given this, firm $i$'s expected payoff from an unobservable deviation to some alternative information strategy $G_i \neq G$ is
$$\mathbb{E}_{G_i}[D_i(x_i; G)] = \int_0^1 D_i(x_i; G) \, dG_i(x_i).$$ 
It follows that $G^*$ is an equilibrium distribution if and only if 
$$G^* \in \argmax_{G_i \in \text{MPC}(F)} \int_0^1 D_i(x_i; G^*) \, dG_i(x_i) \;\;\; \forall i = 1, \dots, n.$$
For notational simplicity, we occasionally the subscript $i$ in $D_i(x_i;G)$ when it does not lead to ambiguity.

\section{Preliminary Analysis}
\label{section:preliminary}
In this section, we first characterize each consumers' optimal search rule and derive the interim demand $D(x_i;G)$ for each firm. We then briefly study the benchmark results.

\subsection{Consumers}
Given any symmetric conjectures, each consumer's optimal search rule can be characterized by their \textit{reservation value} \parencite{mccall1970economics,weitzman1978optimal,wolinsky1986true}. For any $c>0$, the reservation value $r=r(c;\widetilde{G})$ under conjecture $(\widetilde{G}, \dots, \widetilde{G})$ is uniquely defined by the following indifference condition:
\begin{equation}
    c=\int_0^1 \max(x-r,0)d\widetilde{G}(x)=\int_{r}^1(1-\widetilde{G}(t))dt.
    \label{eq:reservation_value}
\end{equation}

At this reservation value, the $c$-consumer is indifferent between stopping and continuing to search. The left-hand side of the first equality is the cost of one additional search, while the right-hand side represents the expected incremental surplus from searching once more, given an `outside option' of value $r$. The second equality follows from integration by parts. In equilibrium, conjectures must be correct $(\widetilde{G}=G^*)$. Thus, we simplify notation by omitting explicit dependence on $G^*$, and write $r(c)=r(c;G^*)$ when the context is clear.

The optimal search rule follows a simple cutoff rule: the consumer stops and purchases immediately from firm $i$ if $x_i\geq r(c)$. If $x_i<r(c)$ for all firms $i$, the consumer purchases from the firm with the highest $x_i$, breaking ties uniformly in case of multiplicity.

For convenience, we define the inverse of the reservation value function $c_G(\cdot)=r^{-1}(\cdot;G)$:
\begin{equation}
    c_G(x):=\int_x^1 (1-G(t))dt.
    \label{eq:c_G}
\end{equation}
The function $c_G(x)$ represents the search cost of the consumer whose reservation value is $x$ (under conjecture $G$). Equivalently, it captures the expected incremental surplus from an additional search when the current best match is $x$.\footnote{
    For this reason, \textcite{dogan2022consumer} refers to $c_G$ as the \textit{incremental-benefit} function.
}
Since reservation values are strictly decreasing in search costs, a consumer's reservation value exceeds $x$ (i.e., $r(c)>x$) if and only if their search cost is below $c_G(x)$. Thus, $H(c_G(x))$ represents the fraction of consumers who continue searching upon receiving a signal $x$.

Heterogeneity of search costs naturally generates heterogeneity in consumers' reservation values, distributed continuously across an interval $[\underline{r},\overline{r}]$, where we denote $\underline{r} := r(\overline{c})$ and $\overline{r} := r(\underline{c})$ as the lowest and the highest reservation value of the consumers, respectively.\footnote{
    If $\underline{c}=0$, we define $\overline{r} := \max(\text{supp}(G))$. This is a notational simplification, as $0$-consumers do not follow the reservation value strategy but instead follow the weakly dominant strategy of visiting all firms.}
Observe that for any consumer with $c>0$, their reservation value satisfies $r(c)<\max(\supp(G))$, since no additional surplus can be obtained from searching after receiving the highest possible signal $x=\max(\supp(G))$.

\subsection{Firms}
To obtain the interim demand of a firm $D(x_i;G)$, we first derive the type-specific interim demand $D^c(x_i;G)$ for each consumer type $c\in[\underline{c},\overline{c}]$. If $x_i \geq r(c)$, the consumer purchases immediately upon visiting firm $i$. A visit to firm $i$ is made if and only if all previously visited firms have provided signals strictly below the reservation value $r(c)$. Given the random search order, the probability firm $i$ is visited in $k$-th order is $\frac{1}{n} G(r(c) -)^{k-1}$; summing over $k=1,\dots, n$ yields\footnote{
    The left limit $G(x-):=\lim_{t\rightarrow x-}G(t)$ always exists since $G$ is monotone increasing.
}
$$D^c(x_i;G) = \sum_{k=1}^{n}\frac{1}{n}G(r(c)-)^{k-1}=\frac{1-G(r(c)-)^n}{n(1-G(r(c)-))}.$$

If $x_i < r(c)$, the consumer continues searching and eventually returns to firm $i$ only if it provided the highest signal among all visited firms. A standard argument shows that equilibrium distribution $G$ cannot have interior atoms: in case of a tie, a firm can profitably deviate by slightly spreading the mass to secure a discrete increase in its payoff.\footnote{
    For example, one can deviate by splitting the mass at $a$ into $\{a - n\epsilon, a + \epsilon\}$ for some $\epsilon > 0$ (Lemma \ref{lemma:continuity_atom}).}
Accordingly, we restrict attention to distributions without interior atoms. In this case, $D^c(x_i;G)=G(x_i)^{n-1}$ for $x_i<r(c)$.

Hence, the type-specific interim demand $D^c(x_i;G)$ is summarized as:
\begin{equation}
    D^c(x_i;G)=
    \begin{dcases}
        G(x_i)^{n-1} & \text{if } x_i<r(c),
        \\ \frac{1-G(r(c)-)^n}{n(1-G(r(c)-))} & \text{if } x_i\geq r(c).
    \end{dcases}
    \label{eq:D^c}
\end{equation}
Note that $D^c(x_i; G)$ exhibits a discrete jump of size $J_G(r(c))$ at $x = r(c)$, where 
\begin{equation*}
J_G(x):=\frac{1-G(x-)^n}{n(1-G(x-))}-G(x)^{n-1}>0. \label{eq:J_G}
\end{equation*}




Integrating $D^c(x_i;G)$ over all search costs yields the interim demand $D(x_i;G)=\int_{\underline{c}}^{\overline{c}}D^c(x_i;G)dH(c)$:
\begin{equation}
    D(x_i;G)=
    \begin{dcases}
        G(x_i)^{n-1}& x<\underline{r}\\ 
        \underbrace{G(x_i)^{n-1}H(c_G(x))\vphantom{\int}}_{\text{Consumers who continue search}} + \int_{c_G(x)}^{\overline{c}}\underbrace{\frac{1-G(r(c))^n}{n(1-G(r(c)))}dH(c)}_{\text{Consumers who stop search}} & x\geq \underline{r} 
    \end{dcases}.
    \label{D(x_i;G)}
\end{equation}
For signals below $\underline{r}$, the market is essentially frictionless, as all consumers visit all firms. For signals above $\underline{r}$, firms acquire local monopoly power, since a positive fraction $1-H(c_G(x))$ of consumers stop immediately. This stopping fraction increases in $x$, eventually reaching one as $x$ approaches $\max(\supp(G))$.\footnote{To be precise, $c=0$ consumers continue their search but are of measure zero.}


\subsection{Two Benchmark Cases}
\label{section:benchmark}

To analyze the effect of heterogeneous search costs on competitive information disclosure, we first highlight two opposing forces: competition and search frictions. We introduce two benchmark cases: (1) the \textbf{no search friction} case and (2) the \textbf{homogeneous search friction} case.  Our model nests both extremes: as $\overline{c} \rightarrow 0$, our model converges to the frictionless case; as $\underline{c} \rightarrow \overline{c}$, we approach the case with homogeneous search costs.
Convergence results of our equilibrium to these benchmarks are presented in Section~\ref{section:comparative}.

\begin{proposition}[\textcite{hwang2023}]
    Assume there is no search friction (i.e., $H(c)=\delta_{0}$).
    In the unique symmetric Nash equilibrium, each firm discloses full information $(G=F)$ when $n$ is sufficiently large.
    \label{proposition_nofriction}
\end{proposition}

In Bertrand competition, firms undercut prices down to marginal cost. Analogously, in information competition without search frictions, firms seek to "overcut" rivals by providing higher signals to win out competition. With no search friction, demand simplifies to $D(x; G) = G(x)^{n-1}$. 
As competition intensifies ($n$ increases), each firm's probability of being the highest signal provider decreases, 
strengthening the incentive to overcut competitors.
To this end, firms truthfully reveal the match value at the top---any garbling of information necessarily results in a lower posterior mean. This process cascades to the bottom, unraveling to full disclosure of match values, i.e., $G=F$.

\begin{proposition}[Informational Diamond Paradox \parencite{au2024attraction}]
    Assume all consumers share a common search cost $c>0$ (i.e., $H(c)=\delta_{c}$).
    There exists an essentially unique perfect Bayesian equilibrium, in which no firm discloses any information $(G=\delta_{\mu})$, and 
    all consumers stop at their first visit.\footnote{$G = \delta_{\mu}$ is \textit{essentially unique} 
    in the sense that any equilibrium distribution is outcome-equivalent to $\delta_{\mu}$: 
    the same reservation value is induced, no useful information is learned, and all consumers stop at the first visited firm.}
    \label{proposition_homogeneous}
\end{proposition}

The standard Diamond paradox \parencite{diamond} exemplifies a hold-up problem: 
a consumer who anticipates a price $p_0$ will,
after incurring the search cost $c$, accept any price below $p_0+c$.
This inelasticity grants firms local monopoly power,
enabling them to incrementally raise price until 
it reaches the monopoly price. 
A similar unraveling occurs in information competition. 

If a consumer expects to receive some `useful' information---specifically, learning that the match quality is low with a strictly positive probability $G(r(c)-)>0$ and continues searching---firms can always profitably deviate. Since the purchase probability jumps by $J_G(r(c))>0$ at the reservation value $r(c)$, a firm can secure a discrete payoff increase by pooling information just above $r(c)$ in a mean-preserving way. This deviation continues to unravel until all information is pooled above the reservation value, i.e. $G(r(c)-)=0$. Consumers are held up by information pooled above their stopping thresholds, receiving no `useful' information and stopping their search at the first visited firm---the `Informational Diamond Paradox'.

This intuition extends to the case where 
consumers have different, but positive, search costs. When $\underline{c}>0$, if a consumer with search cost $\underline{c}$ stops upon receiving signal $x$, then any consumer with a higher search cost $c>\underline{c}$ will also stop at $x$. Therefore, heterogeneity becomes irrelevant as firms effectively treat all consumers as if they had the lowest search cost, $\underline{c}>0$.

\begin{proposition}[Persistence of Informational Diamond Paradox]
    Assume $\underline{c}>0$. 
    Then, no information $G(x)=\delta_{\mu}(x)$ is the essentially unique equilibrium distribution.
    \label{proposition:positive_c_h}
\end{proposition}

In light of Proposition \ref{proposition:positive_c_h}, we restrict attention to the case where $\underline{c}=0$ hereafter.

\section{Equilibrium}
\label{section:equilibrium}

In this section, we examine the equilibrium in markets with heterogeneous search costs.
We demonstrate that a simple class of equilibrium always exists
and that it is the unique class of equilibria when the number of firms is sufficiently large.

For each $a\geq 0$, define an \textit{upper censorship distribution} (UCD) $U_a$ as follows:
\begin{equation}
    U_{a}(x)=
    \begin{dcases}
        F(x) & x\;\in[0,a) \\
        F(a) & x\;\in[a,k_a) \\ 
        1 &  x\in[k_a,1]
    \end{dcases}, \; \; \;
    \; \; k_a = \mathbb{E}[v|\;v\geq a]\geq\mu.
    \label{eq:upper_censorship}
\end{equation}
We call an equilibrium where all firms play $U_a$ an \textit{upper censorship equilibrium} (UCE).
$U_a$ has a simple structure supported on $[0,a]\cup \{k_a\}$:
all match values $v$ below the censorship threshold $a$ are truthfully revealed, while all values $v>a$ are pooled to an atom at $k_a$,
via a signal that only reveals that the value exceeds $a$.
Note that $U_a$ is feasible by construction, with $U_0 = \delta_{\mu}$ (no disclosure) and $U_1 = F$ (full disclosure). Importantly, all UCDs can be ordered in terms of Blackwell informativeness by their censorship thresholds:
$U_b$ Blackwell dominates $U_a$ if and only if $b\geq a$.

The following result characterizes the set of upper censorship equilibria of the game, and shows that this is the only class of symmetric equilibrium when the number of firms is sufficiently large.

\begin{theorem}$\;$
    Let $E_n$ be the set of all (symmetric) equilibrium distributions when there are $n$ firms.
    \begin{enumerate}[label=(\alph*)]
        \item For any $n$, there exists an upper censorship equilibrium. Moreover, if $U_a$ is an equilibrium distribution, then so is $U_b$ for any $b<a$.
        \item 
            In the limit of intense competition $(n\rightarrow\infty)$, all symmetric equilibria are upper censorship. Specifically, there exists $a^M\in[0,1)$ such that $\lim_{n\rightarrow\infty}E_n=\{U_a:a\leq a^M\}$. Furthermore, $\{U_a:a\leq a^M\}\subseteq E_n$ for large enough $n$.\footnote{
        Under Assumption~\ref{assumption:diag} (page~\pageref{assumption:diag}), a stronger conclusion holds:
    $E_n=\{U_a:a\leq a^M\}$ for large enough $n>N$.
    }
    \end{enumerate}
    \label{thm:upper}
\end{theorem}

    The first part of Theorem \ref{thm:upper} shows that for any number of firms, a continuum of upper censorship equilibria exist---including one that exhibits the informational Diamond paradox ($U_0=\delta_{\mu})$. 
    The second part of Theorem~\ref{thm:upper} fully characterizes equilibrium behavior in the limit of intense competition $(n\rightarrow\infty)$: all symmetric equilibria must be the upper-censorship form, while any alternative symmetric equilibria that may arise at lower levels of competition eventually vanish. The entire set of upper censorship equilibria is indexed by a single threshold $a^M$, which we refer to as the \textit{maximally informative threshold}, since $U_{a^M}$ is the most informative equilibrium in the limit of intense competition.


Two key implications follow from Theorem \ref{thm:upper}. 
First, the maximally informative threshold $a^M$ is strictly less than 1, implying that full disclosure does \textit{not} occur even in the limit of intense competition---firms retain an incentive to pool information, as the local monopoly power granted by search frictions persists. Second, whenever $a^M>0$, informative equilibria emerge despite the presence of search frictions, random search, and the unobservability of deviations. These findings stand in stark contrast with the two benchmark cases in Section \ref{section:benchmark}, highlighting the role of consumer heterogeneity in information provision.

For ease of exposition, we focus for the remainder of the paper on the case in which the censorship threshold $a$ lies below the lowest reservation value, that is, $a\leq \underline{r} = r(\overline{c}; U_a)$. While the case of $a > \underline{r}$ follows a similar intuition, it introduces additional analytical complexities; Appendix \ref{appendix:b} provides a comprehensive analysis for the case of $a > \underline{r}$.

For any $a\leq \underline{r}$, upper censorship distribution $U_a$ consists of two segments: the \textit{search} signals $x\in[0,a]$ and the \textit{purchase} signal $x=k_a=\overline{r}$.\footnote{
    Hereafter, to emphasize that the search signal
    is precisely located at the highest reservation value $\overline{r}$,
    we will use $x=\overline{r}$ instead of $x=k_a$,
    unless the dependence on the censorship threshold $a$ needs to be explicitly highlighted.
}
Anticipating $U_a$, the consumers' reservation values are continuously distributed over $[\underline{r}, \overline{r}]=[\underline{r}, k_a]$. Hence, search signals encourage all consumers to continue searching, regardless of their search costs, while the purchase signal $x=\overline{r}$ induces all consumers to cease their search and purchase immediately.

Under an upper censorship equilibrium with $U_a$, each firm's interim demand is given by
\begin{equation*}
    D(x;U_a)=
    \begin{dcases}
        F(x)^{n-1} & \text{ } x\in [0,a)
        \\ F(a)^{n-1} & \text{ } x\in [a,\underline{r})
        \\ F(a)^{n-1} + J_F(a)\underbrace{\Bigg(1-H(c_{U_a}(x))\Bigg)}_{
        \text{fraction of ``stoppers''}} & \text{ } x\in [\underline{r},\overline{r})
        \\ F(a)^{n-1}+J_F(a) & \text{ } x\in [\overline{r},1]
    \end{dcases}, \; \; \;\; 
    J_F(a):=\left(\frac{1-F(a)^n}{n(1-F(a))}-F(a)^{n-1}\right).
\end{equation*}

The right panel of Figure \ref{fig:upper_censorship_D} depicts $D(x;U_a)$ when consumers' search costs are distributed according to $H$ on the left panel. For signals below the lowest reservation value \( \underline{r} \), firms essentially face a frictionless market, leading to the interim demand of  \( U_a(x)^{n-1}=F(x)^{n-1} \).

Notably, while competing firms provide no information in the region \( (a, \overline{r}) \), the demand remains flat over \( (a, \underline{r}) \) but increases strictly over \( (\underline{r}, \overline{r}) \).

This increase arises because, firm's deviation to a \textit{partial-purchase signal} \( x \in (\underline{r}, \overline{r}) \) induces high-cost consumers with $c\geq c_{U_a}(x)$ to halt their search immediately, whereas they would have continued searching under $x=a$. The gain from capturing these consumers is quantified by the term $J_F(a)(1-H(c_{U_a}(x)))$. Note that a UCE essentially represents an all-or-nothing gamble: each signal either prompts all consumers to continue their search or all to stop immediately. In contrast, providing partial-purchase signals serves a safer gamble by differentiating purchase probability across consumer types.




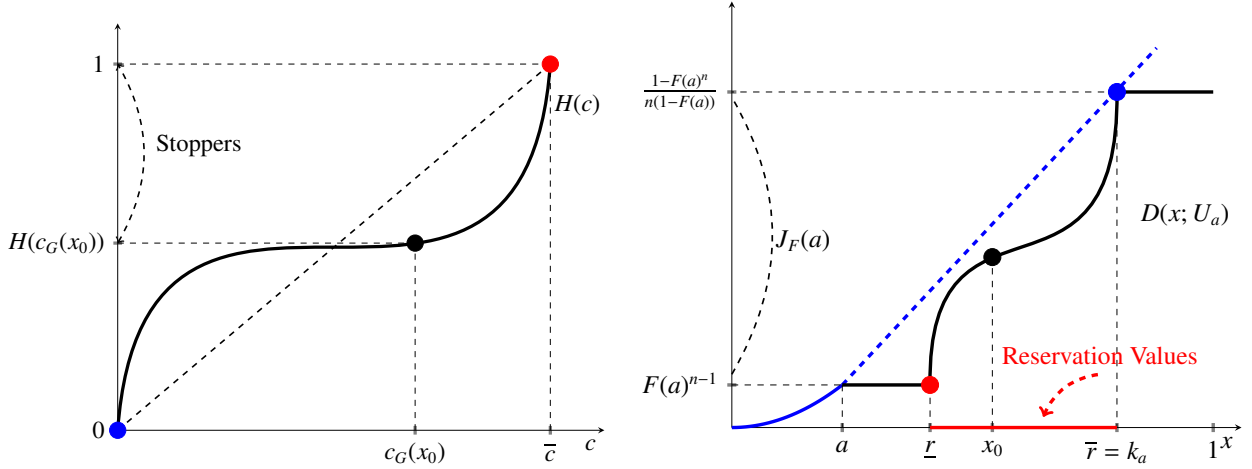
\begin{figure}[ht]
    \begin{subfigure}[b]{.49\textwidth}
        \begin{adjustbox}{width=\linewidth}
            \begin{tikzpicture}
                \begin{axis}[
                    xlabel style = {at={(axis description cs:1,0)},anchor=north east},
                    ylabel style = {at={(axis description cs:0,1)},anchor=south east},
                    axis lines = left,
                    xmin = 0,
                    xmax = 1.8,
                    ymin = 0,
                    ymax = 1,
                    clip = false,
                    legend pos = north west,
                    xtick = {1.1, 1.6},
                    xticklabels = {$c_G(x_0)$, $\overline{c}$},
                    ytick = {0, 0.46, 0.9},
                    yticklabels = {$0$, $H(c_G(x_0))$, $1$},
                ]            
                \draw[ultra thick]    
                (axis cs: 0.0,0.0) .. controls (axis cs: 0.1,0.9) and (axis cs: 1.5,0) .. (axis cs: 1.6,0.9) ;
                \node[right] at (axis cs: 1.8, 0) {$c$};                
                \node[left] at (axis cs: 0, 1) {$H(c)$};
                \draw[dashed, thin] (axis cs: 1.1,0) -- (axis cs: 1.1,0.46);
                \draw[dashed, thin] (axis cs: 0,0.46) -- (axis cs: 1.1,0.46);
                \draw[dashed, thick , <->] (axis cs: 0,0.46) to [bend right] (axis cs: 0,0.9) ;
                \node at (axis cs: 0.3, 0.7) {``stoppers''};
                \node[circle, draw=black, fill=black, inner sep=0pt, minimum size = 8pt] at (axis cs: 1.1, 0.46) {};
                \addplot[
                    domain = 0:1.6,
                    color=black,
                    dashed,
                    style = thick]
                    {0.9*x/1.6};
                \draw[dashed, thin] (axis cs: 0,0.9) -- (axis cs: 1.6,0.9);
                \draw[dashed, thin] (axis cs: 1.6,0) -- (axis cs: 1.6,0.9);
                \node[circle, draw=red, fill=red, inner sep=0pt, minimum size = 8pt] at (axis cs: 1.6, 0.9) {};
                \node[circle, draw=blue, fill=blue, inner sep=0pt, minimum size = 8pt] at (axis cs: 0, 0) {};
                \end{axis}
            \end{tikzpicture}
        \end{adjustbox}
    \end{subfigure}
    \hfill
    \begin{subfigure}[b]{.49\textwidth}
        \begin{adjustbox}{width=\linewidth}
            \begin{tikzpicture}
                \begin{axis}[
                    xlabel style = {at={(axis description cs:1,0)},anchor=north east},
                    ylabel style = {at={(axis description cs:0,1)},anchor=south east},
                    axis lines = left,
                    xmin = 0,
                    xmax = 1.8,
                    ymin = 0,
                    ymax = 0.8,
                    clip = false,
                    legend pos = north west,
                    xtick = {0.39, 0.7, 0.92, 1.36, 1.7},
                    xticklabels = {$a$, $\underline{r}$, $x_0$, $\overline{r}=k_a$, 1},
                    ytick = {0.08, 0.63},
                    yticklabels = {$F(a)^{n-1}$, $\frac{1-F(a)^n}{n(1-F(a))}$},
                ]            
                \addplot[
                    domain = 0:0.39,
                    color=blue,
                    style = ultra thick]
                    {x^2/1.9025};
                \addplot[
                    domain = 0.39:1.5,
                    color=blue,
                    dashed,
                    style = ultra thick]
                    {0.57*(x-0.39)+0.08};
                \draw[ultra thick]    (axis cs: 0.39,0.08) -- (axis cs: 0.70,0.08) ;
                \draw[ultra thick]    (axis cs: 0.70,0.08) .. controls (axis cs: 0.69,0.49) and (axis cs: 1.37,0.18) .. (axis cs: 1.36,0.63) ;
                \draw[ultra thick]    (axis cs: 1.7,0.63) -- (axis cs: 1.36,0.63) ;
                \node[right] at (axis cs: 1.8, 0) {$x$};
                \node[left] at (axis cs: 0, 0.8) {$D(x;U_a)$};
                \draw[dashed, thin] (axis cs: 1.36,0.63) -- (axis cs: 1.36,0);
                \draw[dashed, thin] (axis cs: 1.36,0.63) -- (axis cs: 0,0.63);
                \draw[dashed, thin] (axis cs: 0.39,0.08) -- (axis cs: 0.39,0);
                \draw[dashed, thin] (axis cs: 0.70,0.08) -- (axis cs: 0.70,0);
                \draw[dashed, thin] (axis cs: 0,0.08) -- (axis cs: 0.39,0.08);
                \draw[dashed, thin] (axis cs: 0.92,0) -- (axis cs: 0.92,0.32);
                \node[circle, draw=black, fill=black, inner sep=0pt, minimum size = 8pt] at (axis cs: 0.92, 0.32) {};
                \node[circle, draw=blue, fill=blue, inner sep=0pt, minimum size = 8pt] at (axis cs: 1.36, 0.63) {};
                \node[circle, draw=red, fill=red, inner sep=0pt, minimum size = 8pt] at (axis cs: 0.7, 0.08) {};
                \draw[dashed, thick] (axis cs: 0,0.1) to [bend right] (axis cs: 0,0.62);
                \draw[red, ultra thick] (axis cs: 0.7,0) to (axis cs: 1.36, 0);
                \draw [<-,red, dashed, ultra thick] (axis cs: 1.1,0.02) to [bend left] (axis cs: 1.3,0.1)
                node[above] {  reservation values };
                \node at (axis cs: 0.25,0.35) {$J_F(a)$};
                \end{axis}
            \end{tikzpicture}
    \end{adjustbox}
    \end{subfigure}
    \caption{
        The cost distribution $H(c)$ (left panel) 
        and the payoff $D(x;U_a)$ (right panel).
        The blue/black/red dots in the right panel corresponds to the reservation values of the consumers with different search costs (blue/black/red dots) in the left panel.
        }
    \label{fig:upper_censorship_D}
\end{figure}

Over the interval $(\underline{r},\overline{r})$ of partial-purchase signals,
the demand $D(x_i;U_a)$ is a reflected and scaled version of the cost distribution $H$.
Mathematically, this arises from the linearity of $c_{U_a}(x)=(F(a)-1)(x-k_a)$ in this interval. Intuitively, because no information is provided under $U_a$ in this interval, the probability of receiving a higher signal remains constant.

We next present a graphical argument for equilibrium verification: For sufficiently large $n$, $U_a$ constitutes an equilibrium  if the line connecting \( (a, D(a; U_a)) \) and \( (\overline{r}, D(\overline{r}; U_a)) \) (blue dotted line in Figure~\ref{fig:upper_censorship_D}) remains above \( D(x; U_a) \) for all \( x \in [\underline{r}, \overline{r}] \). To understand this, define the auxiliary function 
\( \phi_a:[0,1] \rightarrow \mathbb{R} \) as 
\begin{equation}  
\phi_a(x) :=  
\begin{cases}  
    D(x, U_a), & x < a, \\  
    \dfrac{D(\overline{r}, U_a) - D(a, U_a)}{\overline{r} - a}(x - a) + F(a)^{n-1}, & x \geq a.  
\end{cases}  
\label{eq:upper_DM}  
\end{equation}  
For  \( x\in\supp(U_a) = [0, a] \cup \{\overline{r}\} \), the function \( \phi_a(\cdot) \) matches the interim demand \( D(\cdot; U_a) \), while over the interval \( [a, 1] \), it is the secant line of $D(\cdot;U_a)$ connecting points $a$ and $\overline{r}$. 

The function $\phi_a$ provides a straightforward equilibrium verification: $U_a$ is an equilibrium if and only if (1) $\phi_a$ is convex, and (2) $\phi_a(x)\geq D(x;U_a)$ for all $x$.\footnote{
    The function $\phi_a$ is the price function in \textcite{dworczak2019simple}.
    They show that $G\in \text{MPC}(F)$ is an equilibrium if and only if there exists some $\phi:[0,1]\rightarrow \mathbb{R}$ such that
    (a) $\phi(x)$ is convex, and $\phi(x)\geq D(x,G)$ for all $x$, (b) $\text{supp}(G)\subseteq \{x| \phi(x)=D(x,G)\}$, and
    (c) $\int_0^1 \phi(x)dF(x)=\int_0^1\phi(x)dG(x)$. Our function $\phi_a$ satisfies (b) and (c) by construction. See Lemma \ref{lemma:iff}.
}  Convexity ensures that redistributing mass within $\supp(U_a)$ is unprofitable. The second condition guarantees any deviation outside of $\supp(U_a$)---a partial-purchase signal---is unprofitable. When the number of firms $n$ is sufficiently large, the convexity of $\phi_a$ naturally emerges.\footnote{
    See Lemma \ref{lemma:phi_a_convex_F}
    and Lemma \ref{lemma:phi_a_convex_kink}
    in Appendix \ref{appendix:b}.} 
Over $[0,a]$, $\phi_a=F^{n-1}$ reflects the distribution of the consumer's best alternative. Increased competition makes consumers more likely to encounter better alternatives, implying the convexity of $F^{n-1}$. Thus, firms provide more information at the bottom to win out competition. Moreover, intense competition decreases the likelihood of consumers returning after leaving a firm, increasing the incentive to pool at signal $\overline{r}$, which manifests as an upward kink of $\phi_a$ at $x=a$.\footnote{
    This intuition is analogous to that of \textcite{stahl1989oligopolistic}, where as $n\rightarrow \infty$, the increased incentive to secure consumers results as the price converging to the monopoly price.
} As $n$ increases, the local incentive to outcompete rivals at the bottom, and the global incentive to retain searchers at the top become more pronounced.

The function \( \phi_a \) carries a deeper economic interpretation---it represents the \textit{virtual interim demand} faced by a firm choosing the information structure \( U_a \).\footnote{  
    In analogy to the virtual value in \textcite{myerson1981optimal}, where it represents the Lagrange multiplier of the agent's local incentive compatibility constraint, \textcite{kim2023choosing} adopts this terminology in the context of information design. Here, the curvature \( \phi_a'' \) serves as the Lagrange multiplier associated with the MPC constraint \( G \in MPC(F) \).  
}  
The gap between the virtual interim demand \( \phi_a(x) \) and 
the interim demand \( D(x; U_a) \) quantifies 
the virtual interim gain of following \( U_a \),
rather than deviating to a different distribution \( G \neq U_a \) 
with the signal \( x \) in its support.  

To illustrate,
fix some consumer type \( c \in (0, \overline{c}) \), and consider a deviation from \( U_a \)  
to a partial-purchase signal \( x = r(c)\in(\underline r, \overline r) \),  collapsing some mass from  
\( a \) and \( \overline{r} \).  
Such deviation entails both a cost and a gain.  
Shifting mass downward from the purchase signal $\overline{r}$ to $r(c)$ forgoes the consumers with cost $c'\in[0,c]$, who would have stopped immediately upon receiving signal $\overline{r}$. However, shifting mass upward from the search signal $a$ to $r(c)$ captures the high-cost consumers with $c'\in[c,\overline{c}]$, who would have otherwise continued searching upon receiving signal $a$. Since deviations must occur in a mean-preserving way, inducing an atom of size $dm$ at $x=r(c)$ requires moving exactly $\frac{c}{c_F(a)}dm$ mass from signal $a$ and the complementary \( \left(1 - \frac{c}{c_F(a)}\right)dm \) mass from \( \overline{r} \).\footnote{To see this, let $\beta dm$ be the mass shifted from $a$ to $r(c)$ and $(1-\beta)dm$ from $\overline{r}$ to $r(c)$. Bayes plausibility implies $\beta dm \times (r(c)-a) + (1-\beta)dm \times (r(c)-\overline{r})=0$. Dividing by $dm$ and rearranging gives $\beta (\overline{r}-a)=\overline{r}-r(c)$, so: $$\beta=\frac{\overline{r}-r(c)}{\overline{r}-a}=\frac{(1-F(a))(\overline{r}-r(c))}{(1-F(a))(\overline{r}-a)}=\frac{c_{U_a}(r(c))}{c_{U_a}(a)}=\frac{c}{c_F(a)}.$$ The third equality follows from $c_{U_a}(x)=(1-F(a))(\overline{r}-x)$ over $x\in[a,\overline{r}]$ (see definition \eqref{eq:c_G}); the final equality uses $c_{U_a}(a)=c_F(a)$ since $U_a=F$ over $[0,a]$.
} Therefore, the net gain from the deviation is \begin{align}  
    \text{Net Gain}  
    &=
    J_F(a) \left(  
        \frac{c}{c_F(a)}(1 - H(c))
        - \left(1 - \frac{c}{c_F(a)}\right)H(c)
    \right)dm \label{eq:net_gain}\\  
    & = -\Bigg(\phi_a(r(c)) - D(r(c); U_a)\Bigg)dm.\label{eq:diff}  
\end{align}  

The first and second terms in \eqref{eq:net_gain} represent the gain 
from capturing the high-cost consumers and the loss from forgoing the low-cost consumers, respectively.  
By \eqref{eq:diff}, any deviation to a partial-purchase signal is unprofitable if and only if
$\phi_a(x)\geq D(x_i;U_a)$ for all $x$. 

Equation \eqref{eq:net_gain} clarifies why multiple UCEs exist. As the censorship threshold $a$ increases, the weight $\frac{c}{c_F(a)}$ associated with the gain grows, making deviations increasingly attractive. In other words, firms become more tempted to exploit the local monopoly power when they are already disclosing relatively more information. Consequently, the constraint $\phi_a(x) \geq D(x; U_a)$ tightens uniformly as $a$ increases. If $U_a$ is an equilibrium, lowering the threshold to any $b<a$ relaxes the constraint, implying $U_b$ is an equilibrium.

    It is worthwhile to note that the the sign of the large bracketed term in \eqref{eq:net_gain} is independent of the number of firms $n$.    This sheds light on why upper censorship equilibria remain stable in the limit of intense competition. While the profitability of partial-purchase deviations does not vary with $n$, the function $\phi_a$ becomes convex for sufficiently large $n$. In other words, the incentive to deviate remains unchanged, while the appeal of adhering to $U_a$ strengthens in the intense competition limit, reinforcing the equilibrium.

\paragraph{Maximally Informative Equilibrium}

Given the above analysis, we can characterize the maximally informative threshold $a^M$---the highest censorship threshold $a$ at which no profitable deviations exist. Rearranging \eqref{eq:net_gain}, we derive a simple condition: Given an upper-censorship profile $U_a$, a deviation to \textit{any} partial-purchase signal is unprofitable if and only if
\begin{equation}
\frac{1}{c_F(a)}\leq \underline{h}^{avg}:= \min_{c\in[0,\overline{c}]}\frac{H(c)}{c}. 
\label{eq:unprofitable}
\end{equation}
We call $\underline{h}^{avg}$ the \emph{minimum average density} of $H$. Note that as the censorship threshold $a$ rises, the left-hand side of \eqref{eq:unprofitable} strictly increases, so the constraint becomes more binding. Thus, at the maximally informative threshold $a^M$, \eqref{eq:unprofitable} binds with equality, and any further increase in $a$ would trigger profitable deviations.

The term $\underline{h}^{avg}$ serves a sufficient statistic for our analysis and has a straightforward economic interpretation: it measures  the \textit{evenness} of the distribution $H$.
As $H$ approaches the uniform distribution---the most even distribution---$\underline{h}^{avg}$ increases.
A higher slope $\underline{h}^{avg}$ indicates a distribution closer to uniformity (the diagonal), representing a more evenly spread distribution of search costs. Consequently, the uniform distribution, being the most even, attains the highest $\underline{h}^{avg}=\frac{1}{\overline{c}}$. Notably, the maximal informativeness of equilibria is determined by the \emph{global} structure of search cost distribution.


To provide a clear characterization of the maximally informative threshold, we impose an assumption on the search cost distribution $H$. While relaxing this assumption introduces some technical complexity, a key qualitative feature of our analysis remains unchanged. We discuss the general case at the end of this section and present the full analysis in Appendix~\ref{appendix:b}.

\begin{assumption}    \label{assumption:diag}
    The minimum of $H(c)/c$ is attained  at some $c^M\in[0,\overline{c})$, i.e., $\underline{h}^{avg}=h(c^M)$.\footnote{
        At $c=0$, we define $\frac{H(0)}{0}:=h(0)$.
        A necessary and sufficient condition for the assumption
        is that there exists  some $c$ such that  $\frac{H(c)}{c}\leq \frac{1}{\overline{c}}$.
        Graphically, this is equivalent to $H$ lying weakly below the uniform distribution over some interval.
        This condition holds for any distribution $H$ with a symmetric density, distributions that are MPS or MPC of the uniform distribution, or whenever either 
        $h(0)\leq \frac{1}{\overline{c}}$ or $h(\overline{c})\geq \frac{1}{\overline{c}}$.
        }
\end{assumption}

In Appendix \ref{appendix:b}, we show that Assumption \ref{assumption:diag} implies $a\leq \underline r$, so the preceding analysis applies in this case. 
Assuming the solution that binds \eqref{eq:unprofitable} with equality exists, $a^M$ is characterized as:
\begin{equation}  
\frac{1}{c_F(a^M)} = \underline{h}^{avg} \quad \Longleftrightarrow\quad a^M=c_F^{-1}\left(\frac{1}{\underline{h}^{avg}}\right).
\label{eq:optimal}  
\end{equation}  
Equation \eqref{eq:optimal} implies that, in the maximally informative equilibrium, firms are indifferent to deviating to the optimal partial-purchase signal. In other words, the signal $x=r(c^M, U_{a^M})$ maximizes the net gain \eqref{eq:net_gain} from deviation but results in zero gain under the profile $U_{a^M}$. Graphically, this corresponds to the virtual interim demand $\phi_{a^M}(\cdot)$ being \textit{tangent} to the interim demand (Figure~\ref{fig:maximally}).

\begin{figure}[ht]
    \begin{subfigure}[b]{.49\textwidth}
        \begin{adjustbox}{width=.99\linewidth}
            \begin{tikzpicture}
                \begin{axis}[
                    xlabel style = {at={(axis description cs:1,0)},anchor=north east},
                    ylabel style = {at={(axis description cs:0,1)},anchor=south east},
                    axis lines = left,
                    xmin = 0,
                    xmax = 1.8,
                    ymin = 0,
                    ymax = 1,
                    clip = false,
                    legend pos = north west,
                    xtick = {1.36, 1.6},
                    xticklabels = {$\textcolor{red}{c^M}$, $\overline{c}$},
                    ytick = {0, 0.9},
                    yticklabels = {$0$, $1$},
                ]            

                \draw[ultra thick]    
                (axis cs: 0.0,0.0) .. controls (axis cs: 0.1,0.9) and (axis cs: 1.5,0) .. (axis cs: 1.6,0.9) ;
                \node[right] at (axis cs: 1.8, 0) {$c$};                
                \node[left] at (axis cs: 0, 1) {$H(c)$};
                \node at (axis cs: 1, 0.2) {\textcolor{blue}{Slope=$\underline{h}^{avg}=h(c^M)$}};
                \addplot[
                    domain = 0:1.6,
                    color=blue,
                    style = ultra thick]
                    {0.38*x};
                \addplot[
                    domain = 0:1.6,
                    color=black,
                    dashed,
                    style = thick]
                    {0.9*x/1.6};
                \draw[dashed, thin] (axis cs: 0,0.9) -- (axis cs: 1.6,0.9);
                \draw[dashed, thin] (axis cs: 1.6,0) -- (axis cs: 1.6,0.9);
                \node[circle, draw=red, fill=red, inner sep=0pt, minimum size = 5pt] at (axis cs: 1.36, 0.52) {};
                \node[circle, draw=blue, fill=blue, inner sep=0pt, minimum size = 5pt] at (axis cs: 0, 0) {};
                \draw[dashed, red, thin] (axis cs:1.36,0.52) -- (axis cs: 1.36,0);
                \end{axis}
            \end{tikzpicture}
        \end{adjustbox}
        \caption{ Cost distribution $H(c)$ and the point of tangency from the origin ($c^M$)}
        \label{fig:maximal_H}
    \end{subfigure}
    \hfill
    \begin{subfigure}[b]{.49\textwidth}
        \begin{adjustbox}{width=.99\linewidth}
            \begin{tikzpicture}
                \begin{axis}[
                    xlabel style = {at={(axis description cs:1,0)},anchor=north east},
                    ylabel style = {at={(axis description cs:0,1)},anchor=south east},
                    axis lines = left,
                    xmin = 0,
                    xmax = 1.8,
                    ymin = 0,
                    ymax = 1,
                    clip = false,
                    legend pos = north west,
                    xtick = {0.65, 0.96 ,1.53, 1.7},
                    xticklabels = {$a^M$, $\underline r$,$k_a$, 1},
                    ytick = {0.17, 0.73},
                    yticklabels = {$F(a^M)^{n-1}$, $\frac{1-F(a^M)^n}{n(1-F(a^M))}$},
                ]            
                \addplot[
                    domain = 0:0.65,
                    color=blue,
                    style = ultra thick]
                    {x^2/2.5};
                \addplot[
                    domain = 0.65:1.8,
                    color=blue,
                    style = ultra thick]
                        {(x-0.65)/1.57 +0.17}
                        node [pos = 0.9, above left] {$\phi_{a^M}(x)$};
                \draw[ultra thick]    (axis cs: 0.65,0.17) -- (axis cs: 0.96,0.17) ;
                \draw[ultra thick]    (axis cs: 0.96,0.17) .. controls (axis cs: 0.96,0.71) and (axis cs: 1.53,0.26) .. (axis cs: 1.53,0.73) ;
                \draw[ultra thick]    (axis cs: 1.53,0.73) -- (axis cs: 1.8,0.73) ;
                \node[circle, draw=red, fill=red, inner sep=0pt, minimum size = 5pt] at (axis cs: 1.07, 0.44) {};
                \node[circle, draw=blue, fill=blue, inner sep=0pt, minimum size = 5pt] at (axis cs: 1.53, 0.73) {};
                \node[right] at (axis cs: 1.8, 0) {$x$};
                \node[left] at (axis cs: 0, 1) {$D(x;U_a)$};
                \draw[dashed, thin] (axis cs: 0.65,0.17) -- (axis cs: 0.65,0);
                \draw[dashed, thin] (axis cs: 0.96,0.17)  -- (axis cs: 0.96,0);
                \draw[dashed, thin] (axis cs: 0.65,0.17)  -- (axis cs: 0,0.17);
                \draw[dashed, thin] (axis cs: 1.53,0.73) -- (axis cs: 1.53,0);
                \draw[dashed, thin] (axis cs: 1.53,0.73) -- (axis cs: 0,0.73);

                \draw[dashed, red, thick] (axis cs: 1.07,0.44) -- (axis cs: 1.07 ,0);

                \draw[<-, red, line width=0.5mm](axis cs:1.09, 0.02) -- ( axis cs: 1.18,0.08 )
                node[above right]  {$\textcolor{red}{r(c^M, U_{a^M})}$};
                \end{axis}
            \end{tikzpicture}
    \end{adjustbox}
    \caption{Firm's interim demand $D(x,U_{a^M})$ and the corresponding virtual interim demand $\phi_{a^M}(x)$}
    \label{fig:maximal_demand}
    \end{subfigure}
    \caption{The maximally informative upper censorship equilibrium $U_{a^M}$}
    \label{fig:maximally}
\end{figure}
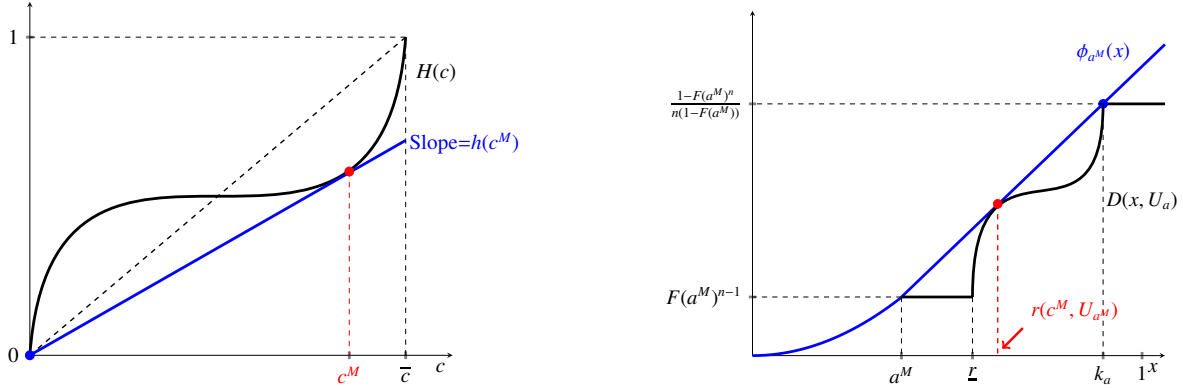

Since $c_F(a)\leq \mu$ for all $a$, a solution \( a^M \) to \eqref{eq:optimal} exists if and only if \( \underline{h}^{avg} \geq \frac{1}{\mu} \). If this condition fails (\( \underline{h}^{avg} < \frac{1}{\mu} \)), inequality \eqref{eq:unprofitable} cannot hold for any threshold \( a > 0 \), and thus the only equilibrium is no information: $U_0=\delta_{\mu}$.
The following proposition characterizes the set of all UCEs under Assumption \ref{assumption:diag}, summarizing the preceding discussion.\footnote{
    The general characterization without the assumption is characterized in Theorem \ref{thm:appendix_maximal} in Appendix \ref{appendix:b}.
}

\begin{proposition}
    Assume $H$ satisfies Assumption \ref{assumption:diag}. Then,
        \begin{enumerate}[label=(\alph*)]
        \item If $\underline{h}^{avg}\in (0,\frac{1}{\mu}]$, then 
        $$a^M=0;$$
        \item If $\underline{h}^{avg}\in (\frac{1}{\mu},\frac{1}{\overline{c}}]$, then 
        $$a^M=c_F^{-1}\left(\frac{1}{\underline{h}^{avg}}\right).$$
    \end{enumerate}
    \label{proposition:charcteriztion}
\end{proposition}

The following result is a straightforward implication of Proposition~\ref{proposition:charcteriztion}.
\begin{corollary}
    The maximally informative threshold $a^M$ is increasing in $\underline{h}^{avg}$.
    Among the class of distributions satisfying Assumption \ref{assumption:diag},
    $a^M$ is highest in the uniform distribution.
    \label{corollary:a^M}
\end{corollary}

To understand why the ``evenness'' $\underline{h}^{avg}$ of a search cost distribution determines the maximum equilibrium informativeness, consider a distribution $H$ that is uneven. In this case, the interim demand under an upper censorship distribution also becomes uneven, as demonstrated in Figures~\ref{fig:upper_censorship_D}~and~\ref{fig:maximally}. Thus, firms have higher incentives to exploit the unevenness by offering a partial-purchase signal that captures a large mass of high-cost consumers while minimizing the loss of forgoing low-cost consumers. As $\underline{h}^{avg}$ increases, such deviations become less profitable.

%
%

Importantly, the monotonic relationship between the maximal informativeness and $\underline{h}^{avg}$ is independent of the prior match value distribution $F$. When comparing two cost distributions $H_1$ and $H_2$, the relative informativeness of equilibrium is determined solely by their intrinsic properties and remains unaffected by $F$. Similarly, the optimal segmentation $c^M$ of low- and high-cost consumers depends solely on the distribution $H$.  Furthremore, $\underline{h}^{avg}$ is easily identifiable through a graphical approach (see Figure~\ref{fig:maximal_H}): it corresponds exactly to the slope of the lowest tangent line passing through the origin and supporting the distribution $H$. This intuitive graphical interpretation provides a simple and effective method to determine and compare maximal informativeness across different cost distributions. A detailed comparative static analysis regarding the role of cost distributions is provided in Section~\ref{section:comparative}.

\paragraph{Uniqueness}
We briefly outline why upper censorship distributions uniquely emerge as equilibria under intense competition. First, 
any disclosed information at lower signals  $(x<\underline{r})$ must be fully revealed.
The intuition mirrors the frictionless benchmark (Proposition \ref{proposition_nofriction}): as the market is essentially frictionless for low signals, any disclosed information unravels toward full disclosure as competition intensifies.

Second, we argue that in any equilibrium, the firms must provide no information for signals close to $\overline r$ (this holds true for any number of $n$). To see this, consider the marginal increase in interim demand at $x\geq \underline{r}$:
\begin{equation}
D'(x,G)=\underbrace{J_G(x)\times\overbrace{(1-G)h(c_G)}^{=(1-H(c_G(x)))'}\vphantom{{\frac{1}{2}}}}_{\text{local monopoly gain}}+\underbrace{\vphantom{{\frac{1}{2}}}\Big( G^{n-1}\Big)'H(c_G)}_{\text{competition gain}}.
\label{eq:marginal.interim.demand}
\end{equation}
    The gain of providing a slightly higher signal can be decomposed to two parts: the \emph{local monopoly gain} and the \emph{competition gain}. The local monopoly gain captures ``marginal stoppers'', i.e., the consumers with reservation value $x$ who would stop under a slightly higher signal. Their density is exactly $(1-G(x))h(c_G(x))=(1-H(c_G(x)))'$ and the purchase probability jumps by $J_G(x)$. The competition gain comes from consumers who continue searching (with fraction  $H(c_G(x))$). A slightly higher signal marginally increases the likelihood the firm would win out competition.

Now suppose to the contrary that there exists an equilibrium in which $G$ is continuous at $\overline{r}$. In this case, both terms in \eqref{eq:marginal.interim.demand} converges to zero as $x\rightarrow \overline{r}$: the competition gain disappears as most consumers stop searching ($H(c_G(x))\rightarrow0$), and the local monopoly gain vanishes because firms can almost guarantee purchase without inducing stopping ($G^{n-1}\rightarrow 1$ implies $J_G(x)\rightarrow 0$). Therefore, the interim demand function must be strictly concave around $x=\overline{r}$, inducing firms to deviate by pooling information.\footnote{
    If there are a mass of `shoppers', i.e. $H(0)>0$, the competition gain do not vanish for the highest signals, yielding starkingly different results. For more details, see Section~\ref{section:extensions}.}

Finally, we argue that any signals in the middle values of $x$ (partial-purchase signals) must either be fully revealed or fully concealed. As competition grows, the local monopoly gain exponentially dominates the competition gain: capturing consumers on their first visit becomes critical as those who leave rarely return.\footnote{This follows from 
$\lim_{n \rightarrow \infty} \frac{J(x)}{G(x)^{n-1}} = \infty$ for any $x$ such that $G(x) < 1$.} In this limit, the local curvature of $D$ is governed by the slope of the stopper density $(1-G)h(c_G)$. If $h'(c_G(x))\leq 0$, the density of marginal stoppers rises with $x$, raising the local monopoly gain and rendering the payoff strictly convex---prompting firms to provide more information. Conversely, if $h'(c_G(x))>0$, further pooling becomes optimal. 
Therefore, firms strictly prefer to provide strictly more or strictly less information locally.



\paragraph{General search cost distributions} \label{pg:general.H}
We now clarify the role of Assumption \ref{assumption:diag} by describing the results under a general search cost distribution $H$. Simply put, Assumption~\ref{assumption:diag} allows us to focus only on the global property---``evenness''---of $H$. Without it, we need an additional equilibrium condition regarding the local shape of $H$ in the higher cost range.

Recall that $U_a$  is an equilibrium if and only if (1) $\phi_a$ is convex and (2) $\phi_a(x)\geq D(x;U_a)$ for all $x$. Assumption \ref{assumption:diag} implies that  equilibrium thresholds $a$ are always below the lowest reservation value $\underline{r}$, ensuring that $\phi_a$ is convex for large~$n$. However, when the assumption is relaxed, the possibility of $a\geq \underline{r}$ must be considered. In such cases, consumers with $c\geq c_F(a)$ may stop upon receiving signals $x\in[\underline{r},a]$, affecting the shape of the interim demand function.

We show in Appendix~\ref{appendix:b} that $U_a$ constitutes an equilibrium if and only if the following two equilibrium conditions hold:
\begin{equation}
\frac{H(c_F(a))}{c_F(a)}\leq \min_{c\in[0,c_F(a)]}\frac{H(c)}{c},\quad\text{ and }  \quad   h'(c)\leq 0 \text{ for all } c\in[c_F(a),\overline c]. 
\label{eq:condition}
\end{equation}
The first condition mirrors the earlier condition \eqref{eq:unprofitable}, replacing $\overline c$ with $c_F(a)$. It focuses on global structure of $H$ and ensures $\phi_a(x)\geq D(x;U_a)$. Additionally, the second condition requires that $H$ is concave for $c\geq c_F(a)$. It ensures convexity of $\phi_a$ over $x\in[\underline{r},a]$ in the limit of intense competition as the density of marginal stoppers $h(c_F(x))$ weakly increases over the relevant range.


The graphical method of characterizing $a^M$ naturally extends to cost distributions failing Assumption \ref{assumption:diag}, including those with interior atoms.\footnote{Boundary atoms are separately discussed in Section~\ref{section:extensions} and Online Appendix \ref{section:online_appendix_c}.} In these cases, the equilibrium condition replaces the minimum with an infimum and requires continuity of $H$ over the interval $[c_F(a),\overline{c}]$.\footnote{
    Even if $H$ is discontinuous at the point $c^M=\text{arg}\inf_c\frac{H(c)}{c}$, equilibrium identification remains straightforward: it amounts to finding the largest linear function below $H$ passing through the origin.
}

\section{Comparative Statics}
\label{section:comparative}

In this section, we explore how the maximally informative equilibrium responds to shifts in the cost distribution $H$.
Specifically, we consider compositional shifts that increase search costs (first-order shifts), and that increase heterogeneity (second-order shifts). For each cost distribution $H_k$, denote $\underline{h}_k^{avg}$ as the minimum average density.

We begin by briefly discussing how increased informativeness (higher censorship threshold $a$) affects welfare. Greater informativeness increases consumers' likelihood of being matched with their best alternatives, but it also encourages more active search. Increased probability of searching, $F(a)$, result in longer expected search lengths, $\frac{1-F(a)^n}{1-F(a)}$ and higher cumulative search costs, $\frac{1-F(a)^n}{1-F(a)}c$. Nonetheless, the benefits of improved matching always outweigh these additional costs. Consequently, consumers always prefer more information, and the maximally informative equilibrium is both consumer-optimal and socially optimal.\footnote{Firm profits remain constant at $\frac{1}{n}$ across all equilibria.}

\begin{proposition}
    Let $CS(a;H)$ be the expected consumer surplus under $U_a$ with search cost distribution $H$. 
    \begin{enumerate}[label=(\alph*)]
        \item $CS(a;H)$ is strictly positive and strictly increasing in $a$.
        \item Consider two distributions $H_1\neq H_2$ with the same mean.
        If $a^M_1\leq a^M_2$ and $a^M_k\leq r(\overline{c};U_{a^M_k})$ for $k=1,2$,  then $CS(a^M_1;H_1)\leq CS(a^M_1;H_2)$.
    \end{enumerate}
    \label{proposition:surplus}
\end{proposition}

Part (a) is intuitive: since consumer faces a decision problem, Blackwell improvement of information always increases her payoff.
Part (b) compares welfare across different search cost distributions, demonstrating that under mild conditions, welfare rankings depend solely on comparing censorship thresholds. Thus, by  Corollary~\ref{corollary:a^M}, comparing $\underline{h}_k^{avg}$
suffices to rank welfare across distributions that satisfy Assumption~\ref{assumption:diag}. Practically, this comparison can be easily conducted graphically.

\subsection{Effect of Higher Search Costs}

We begin by analyzing how first-order shifts in the search cost distribution $H$ affect the maximally informative threshold $a^M$. We consider two classes of such shifts. First, we examine \emph{$\alpha$-scale stretches}, where the shape of the density is preserved but the support is proportionally expanded. Second, we consider shifts that preserve the support but alter the distribution in the sense of first-order stochastic dominance (FOSD), skewing it toward higher search costs.

Formally, $H_2$ is an $\alpha$-scale stretch of $H_1$ if $H_1\left(c\right)=H_2(\alpha c)$ for all $c\in[0,\overline{c}]$. This transformation uniformly scales all consumers' search costs upward, stretching the support from $[0,\overline{c}]$ to $[0,\alpha \overline{c}]$ while preserving the density shape. Consequently, key sufficient statistics scale proportionally: $c^M_2=\alpha c^M_1$ and
    $\underline{h}_2^{avg}=\frac{1}{\alpha}\underline{h}_1^{avg}<\underline{h}_1^{avg}$.
From Corollary~\ref{corollary:a^M}, it immediately follows:

\begin{proposition}
    Let $1<\alpha<\frac{\mu}{\overline{c}}$ and
    $H_2$ be an $\alpha$-scale stretch of $H_1$.
    Then, $a^M_2\leq a^M_1$.
    Furthermore, if $a^M_1>0$, then $a^M_2<a^M_1$.
    \label{proposition:comp_stat_shift}
\end{proposition}

This result indicates that uniformly scaling up consumer search costs leads to weakly less information disclosure. If there was positive disclosure initially, informativeness strictly decreases after scaling. Aligned with the standard intuition, increased search friction weakens competition, granting firms greater local monopoly power and exacerbating hold-up incentives, thus discouraging information disclosure.

Another interpretation focuses on partial-purchase signals becoming more profitable under $H_2$. Specifically, the net gain from deviating to any $x=r(c)$ (equation \eqref{eq:net_gain}) is greater under $H_2$ than $H_1$:

\begin{equation}
    J_F(a)\left(\frac{c}{c_F(a)}-H_2(c)\right)>(\geq)\; J_F(a)\left(\frac{c}{c_F(a)}-H_1(c)\right) \quad \forall c\in[0,\overline{c}]
    \label{eq:comparative_stretch}
\end{equation}
Under the stretched distribution $H_2$, fewer low-cost consumers are forgone, and each partial-purchase signal captures a larger range of high-cost consumers. Thus, deviations become uniformly more attractive, reducing the equilibrium's maximal informativeness.

A further observation arises directly from expression \eqref{eq:comparative_stretch}. Deviating to partial-purchase signals becomes weakly more profitable under $H_2$ compared to $H_1$ precisely when $H_2$ first-order stochastically dominates $H_1$. 

\begin{proposition}
    Assume $H_1$ and $H_2$ satisfies Assumption \ref{assumption:diag},
    and $H_2$ first order stochastically dominates $H_1$.
    Then,
    $$a^M_2\leq a^M_1.$$
    \label{proposition:comp_stat_fosd}
    Moreover, if $H_2(c)<H_1(c)$ for all $c\in(0,\overline{c})$ and $a^M_1>0$, then $a^M_2<a^M_1$.
\end{proposition}
Proposition \ref{proposition:comp_stat_fosd} can also be directly derived from Corollary \ref{corollary:a^M}: $H_2(c)/c\leq H_1(c)/c$ for all $c$
implies $\underline{h}_2^{avg}\leq \underline{h}_1^{avg}$.

Both notions of increased search costs (the $\alpha$-scale stretch or FOSD) lead firms to disclose less information in equilibrium. In both cases, deviations to partial-purchase signals become more profitable, since under $H_2$, fewer low-cost consumers are lost and more high-cost consumers can be captured. An $\alpha$-scale stretch disperses consumers more widely, while FOSD shifts the distribution toward higher costs, enhancing profitability of deviations.

Next, we establish a convergence results toward the benchmarks studied in Section \ref{section:benchmark}. Consider an initial distribution $H_0$ and define a sequence $\{H_k\}$ obtained by repeatedly applying an $\frac{1}{\alpha}$-scale stretch (shrinking) $k$ times. As $k$ increases, the support shrinks toward the degenerate distribution $\delta_0$, corresponding to the frictionless scenario. Because 
$\underline{h}_k^{avg}=\alpha^k \underline{h}_0^{avg}$
diverges, the maximally informative threshold
$a^M_k=c_F^{-1}(\frac{1}{\underline{h}_k^{avg}})$ converges to $c_F^{-1}(0)=1$. Hence, the maximally informative equilibrium converges in distribution to full information. This generalizes to arbitrary sequence of distributions with shrinking support.

\begin{corollary}[Convergence to the no-friction case]
    Consider a sequence of cost distributions 
    $\{H_k\}_{k=1}^{\infty}$ converging to $\delta_0$ in distribution, with $\supp(H_k)=[0,\overline{c}_k]$ and $\overline{c}_k\searrow0$.
    Then $U_{a_k^M}\xrightarrow{d}F$, or equivalently, $a^M_k\rightarrow 1$.\footnote{
        This convergence is only in distribution since $a^M_k<1$ for all $k$.
    }
    \label{corollary:convergence}
\end{corollary}

Similarly, we establish a convergence result toward the homogeneous friction scenario. As the lower bound $\underline{c}$
approaches the upper bound $\overline{c}$, the distribution approaches a degenerate point mass at $\overline{c}$, a scenario already analyzed in Proposition~\ref{proposition:positive_c_h}. Even when the support remains fixed as $[0,\overline{c}]$, convergence to homogeneous friction can still be established: equilibrium informativeness ultimately collapses, resulting in no information disclosure.

\begin{corollary}[Convergence to the homogeneous friction case]    \label{corollary:convergence_homo}
    Consider a sequence of cost distributions $\{H_k\}_{k=1}^{\infty}$ with 
    $\text{supp}(H_k)=[0,\overline{c}]$
    converging to $\delta_{\overline{c}}$ in distribution.
    There exists some $K>0$ such that $U_{a^M_k}=\delta_{\mu}$ for all $k>K$, or equivalently, 
    $a^M_k=0$ for all $k>K$.\footnote{
        $H_{k+1}\succeq_{FOSD}H_{k}$ for all $k$ is an example that satisfies the assumption. Moreover, while Proposition \ref{proposition_homogeneous} states $\delta_{\mu}$ is the essentially unique equilibrium in the output-equivalence sense, Corollary \ref{corollary:convergence_homo} states it is the \textit{unique} equilibrium.}
\end{corollary}

\subsection{Effect of Greater Cost Heterogeneity}
A more nuanced issue  concerns the effect of second-order shifts in the search costs on equilibrium information provision. Greater heterogeneity can manifest either as a more evenly dispersed distribution of search costs or, alternatively, as a more polarized distribution with costs concentrated at the extremes. 

Our main finding is that the equilibrium informativeness is increasing in the \textit{evenness} of the distribution---specifically, how closely it resembles the uniform distribution---rather than with dispersion alone.
In fact, informativeness responds non-monotonically with respect to MPS shifts, precisely because their impact on evenness is non-monotonic. 

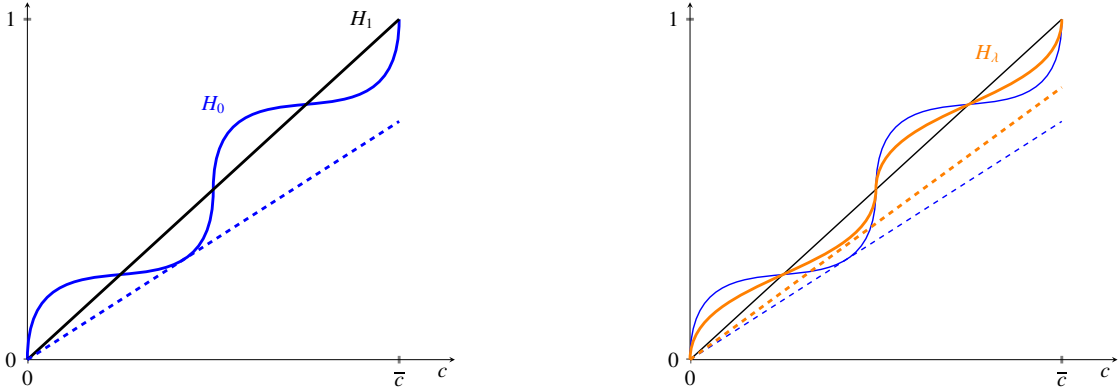
\begin{figure}[ht]
    \begin{subfigure}[b]{.46\textwidth}
        \centering
        \begin{adjustbox}{width=.8\linewidth}
        \begin{tikzpicture}[x=0.75pt,y=0.75pt,yscale=1,xscale=1]
            \begin{axis}[
                xlabel style = {at={(axis description cs:1,0)},anchor=north east},
                ylabel style = {at={(axis description cs:0,1)},anchor=south east},
                axis lines = left,
                xmin = 0,
                xmax = 2.3,
                ymin = 0,
                ymax = 2.1,
                clip = false,
                legend pos = north west,
                xtick = {0, 2.0},
                xticklabels = {$0$,$\overline c$},
                ytick = {0, 2.0},
                yticklabels = {$0$, $1$},
                xlabel = {$c$},
            ]      
            \node at (axis cs: 1.0, 1.5) {$\textcolor{blue}{H_0}$};
            \node at (axis cs: 1.8, 2.0) {$H_1$};
            \draw[ultra thick, blue]  (axis cs: 0,0) .. controls (axis cs: 0,0.9) and (axis cs: 1.0,0.1) .. (axis cs: 1.0,1.0) ;
            \draw[ultra thick, blue]  (axis cs: 1,1) .. controls (axis cs: 1,1.9) and (axis cs: 2.0,1.1) .. (axis cs: 2.0,2.0) ;
            \draw[ultra thick, black] (axis cs: 0,0)-- (axis cs: 2.0,2.0);
            \draw[ultra thick, dashed, blue]    (axis cs: 0.0,0.00) -- (axis cs: 2.0,1.4) ;
            \end{axis}
            \end{tikzpicture}
        \end{adjustbox}
    \end{subfigure}
    \hfill
    \begin{subfigure}[b]{.46\textwidth}
        \centering
        \begin{adjustbox}{width=.8\linewidth}
            \begin{tikzpicture}[x=0.75pt,y=0.75pt,yscale=1,xscale=1]
            \begin{axis}[
                xlabel style = {at={(axis description cs:1,0)},anchor=north east},
                ylabel style = {at={(axis description cs:0,1)},anchor=south east},
                axis lines = left,
                xmin = 0,
                xmax = 2.3,
                ymin = 0,
                ymax = 2.1,
                clip = false,
                legend pos = north west,
                xtick = {0, 2.0},
                xticklabels = {$0$,$\overline c$},
                ytick = {0, 2.0},
                yticklabels = {$0$, $1$},
                xlabel = {$c$},
            ]      
            \draw[thick, blue]  (axis cs: 0,0) .. controls (axis cs: 0,0.9) and (axis cs: 1.0,0.1) .. (axis cs: 1.0,1.0) ;
            \draw[thick, blue]  (axis cs: 1,1) .. controls (axis cs: 1,1.9) and (axis cs: 2.0,1.1) .. (axis cs: 2.0,2.0) ;
            \draw[thick, black] (axis cs: 0,0)-- (axis cs: 2.0,2.0);
            \draw[thick, dashed, blue]    (axis cs: 0.0,0.00) -- (axis cs: 2.0,1.4) ;
        \node at (axis cs: 1.6, 1.8) {$\textcolor{orange}{H_{\lambda}}$};
        \draw[ultra thick, orange]  (axis cs: 0,0) .. controls (axis cs: 0,0.5) and (axis cs: 1.0,0.5) .. (axis cs: 1.0,1.0) ;
        \draw[ultra thick, orange]  (axis cs: 1,1) .. controls (axis cs: 1,1.5) and (axis cs: 2.0,1.5) .. (axis cs: 2.0,2.0) ;
        \draw[ultra thick, dashed, orange]    (axis cs: 0.0,0.00) -- (axis cs: 2.0,1.6) ;
        \end{axis}
    \end{tikzpicture}
        \end{adjustbox}
    \end{subfigure}
    \label{fig:evenness}
    \caption{$H_0\neq H_1=U[0,\overline{c}]$ and the interpolation $H_{\lambda}=\lambda U[0,\overline{c}]+(1-\lambda)H_0$.}
\end{figure}

Consider an arbitrary cost distribution $H_0$.
For each $\lambda\in[0,1]$, define a new distribution $H_{\lambda}$ as a convex combination of $H_0$ and the uniform distribution, given by: $H_{\lambda}:=\lambda U[0,\overline{c}]+(1-\lambda)H_0$. This family of distributions continuously interpolates between $H_0$ and the uniform distribution $U[0,\overline{c}]$, becoming increasingly even as $\lambda$ rises. Importantly, this representation provides a natural interpretation of our measure of evenness, 
    $\underline{h}_{\lambda}^{avg}=\lambda \frac{1}{\overline{c}}+(1-\lambda)\underline{h}_0^{avg}$
which strictly increases in $\lambda$ if $H_0$ satisfies Assumption~\ref{assumption:diag} (since $\underline{h}_0^{avg}<\frac{1}{\overline{c}}$).
The following proposition directly follows from Corollary~\ref{corollary:a^M}. 

\begin{proposition}
    Consider $H_0\neq U[0,\overline{c}]$, and 
    define $H_{\lambda}:=\lambda U[0,\overline{c}]+(1-\lambda)H_0$
    for each $\lambda\in[0,1]$.
    \begin{enumerate}[label=(\alph*)]
        \item $a^M_{\lambda}$ is weakly increasing in $\lambda\in[0,1)$.
        \item If $H_0$ satisfies Assumption \ref{assumption:diag},
        then $a^M_{\lambda}$ is strictly increasing in $\lambda\in[0,1]$.
    \end{enumerate}
    \label{proposition:even}
\end{proposition}
The intuition is clear: firms exploit the unevenness of $H$ by providing partial-purchase signals that minimize the loss from forgoing low-cost consumers. Increased evenness hinders the profitability of such deviations, resulting in higher equilibrium informativeness.

This finding also suggests why equilibrium informativeness is generally non-monotone with respect to MPS shifts. The minimum average density $\underline{h}^{avg}$
is highly sensitive to the global structure of $H$. Since MPS shifts can either spread mass locally or globally, their impact on $\underline{h}^{avg}$ is unclear a priori. The following proposition illustrates this ambiguity by examining commonly used class of distributions, and where dispersion is introduced at a broader scale.

\begin{proposition}
    Suppose that $H_2$ is a mean-preserving spread of $H_1$.
    \begin{enumerate}[label=(\alph*)]
        \item If both $H_1$ and $H_2$ admit strictly quasi-convex densities with interior dips, then the maximal informativeness decreases, i.e.,
        $a_1^M\geq   a_2^M$.
        \item If both $H_1$ and $H_2$ admit strictly quasi-concave  densities with interior peaks, then the maximal informativeness increases, i.e., $a_1^M\leq a_2^M$.
    \end{enumerate}
    \label{proposition:mps}
\end{proposition}

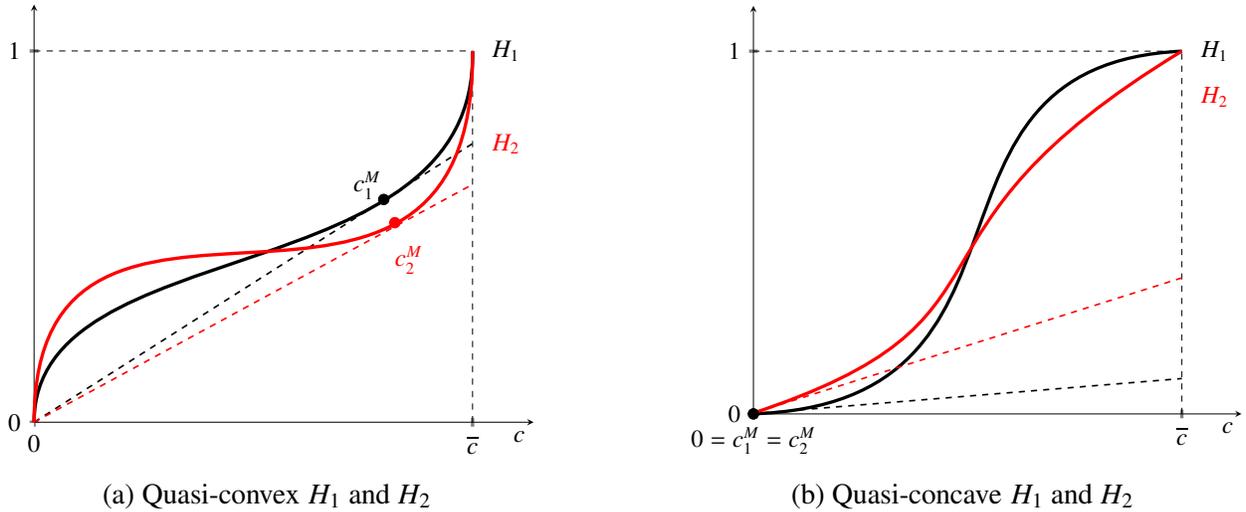
\begin{figure}[ht]
    \begin{subfigure}[b]{.43\textwidth}
        \centering
        \begin{adjustbox}{width=\linewidth}    
            \begin{tikzpicture}
                \begin{axis}[
                    xlabel style = {at={(axis description cs:1,0)},anchor=north east},
                    ylabel style = {at={(axis description cs:0,1)},anchor=south east},
                    axis lines = left,
                    xmin = 0,
                    xmax = 1.8,
                    ymin = 0,
                    ymax = 0.9,
                    clip = false,
                    legend pos = north west,
                    xtick = {0, 1.58},
                    xticklabels = {$0$,$\overline{c}$},
                    ytick = {0, 0.8},
                    yticklabels = {$0$,$1$},
                    xlabel = {$c$}
                ]            
                \draw[ultra thick]    (axis cs: -0.00,0.00) .. controls (axis cs: -0.00,0.42) and (axis cs: 1.60,0.27) .. (axis cs: 1.58,0.8) ;
                \node[circle, draw=black, fill=black, inner sep=0pt, minimum size = 5pt] at (axis cs: 1.26, 0.48) {};
                \node at (axis cs: 1.20, 0.51) {$c^M_1$};
                \node at (axis cs: 1.7, 0.8) {$H_1$};
                \addplot[
                    domain = 0:1.58,
                    color=black,
                    style = thick, dashed]
                    {0.38*x};
                \draw[dashed, thin] (axis cs: 1.58,0) -- (axis cs: 1.58,0.8);
                \draw[dashed, thin] (axis cs: 0,0.8) -- (axis cs: 1.58,0.8);
                \draw[thick, dashed ,red]    (axis cs: -0.00,0.00) -- (axis cs: 1.58,0.5124) ;
                \draw[ultra thick, red]    (axis cs: -0.00,0.00) .. controls (axis cs: -0.00,0.70) and (axis cs: 1.59,0.01) .. (axis cs: 1.58,0.8) ;
                \node[circle, draw=red, fill=red, inner sep=0pt, minimum size = 5pt] at (axis cs: 1.3, 0.43) {};
                \node at (axis cs: 1.7, 0.6) {\textcolor{red}{$H_2$}};
                \node at (axis cs: 1.35, 0.35) {\textcolor{red}{$c^M_2$}};
                \end{axis}
            \end{tikzpicture}
        \end{adjustbox}
        \caption{Quasi-convex $H_1$ and $H_2$}
        \label{subfig:quasi_convex}
    \end{subfigure}
    \hfill
    \begin{subfigure}[b]{.45\textwidth}
        \begin{adjustbox}{width=\linewidth}    
            \begin{tikzpicture}
                \begin{axis}[
                    xlabel style = {at={(axis description cs:1,0)},anchor=north east},
                    ylabel style = {at={(axis description cs:0,1)},anchor=south east},
                    axis lines = left,
                    xmin = 0,
                    xmax = 1.8,
                    ymin = 0,
                    ymax = 0.9,
                    clip = false,
                    legend pos = north west,
                    xtick = {0, 1.58},
                    xticklabels = {$0=c_1^M=c_2^M$,$\overline{c}$},
                    ytick = {0, 0.8},
                    yticklabels = {$0$, $1$},
                    xlabel = {$c$}
                ]            
            
                \draw[ultra thick, black]    (axis cs: -0.01,0.00) .. controls (axis cs: 1.16,0.02) and (axis cs: 0.51,0.77) .. (axis cs: 1.58,0.8) ;
                \draw[thick, dashed, black]    (axis cs: -0.01,0.00) -- (axis cs: 1.58,0.078) ;
                \draw[dashed, thin] (axis cs: 1.58,0) -- (axis cs: 1.58,0.8);
                \draw[dashed, thin] (axis cs: 0,0.8) -- (axis cs: 1.58,0.8);
                \node[circle, draw=black, fill=black, inner sep=0pt, minimum size = 5pt] at (axis cs: 0, 0) {};
                \draw[ultra thick, red]    (axis cs: -0.01,0.00) .. controls (axis cs: 1.06,0.21) and (axis cs: 0.45,0.39) .. (axis cs: 1.58,0.8) ;
                \draw[thick, dashed, red]    (axis cs: 1.58,0.30) -- (axis cs: -0.01,0.00) ;
                \node[circle, draw=black, fill=black, inner sep=0pt, minimum size = 5pt] at (axis cs: 0, 0) {};
                \node at (axis cs: 1.7, 0.7) {\textcolor{red}{$H_2$}};
                \node at (axis cs: 1.7, 0.8) {$H_1$};
                \end{axis}
            \end{tikzpicture}
        \end{adjustbox}
        \caption{Quasi-concave $H_1$ and $H_2$}
        \label{subfig:quasi_concave}
    \end{subfigure}
    \caption{Non-monotonicity of MPS on $h(c^M_i)$}
    \label{fig:quasi_MPS}
\end{figure}

Proposition \ref{proposition:mps} highlights how dispersion of distribution can either enhance or reduce equilibrium informativeness depending on its effect on evenness. If the initial density is quasi-convex, an MPS shift further polarizes the distribution, making it less even ($\underline{h}_2^{avg}<\underline{h}_1^{avg}$,
    see Figure \ref{subfig:quasi_convex}) and consequently decreasing informativeness. Conversely, if the initial density is quasi-concave, an MPS shift redistributes mass more evenly 
    ($\underline{h}_2^{avg}>\underline{h}_1^{avg}$
    , see Figure \ref{subfig:quasi_concave}), thereby enhancing informativeness.

\section{Extensions and Discussions}
\label{section:extensions}

\paragraph{Small number of firms} We relax the large market assumption to study equilibria in small markets (small $n$). While the large-market assumption significantly simplifies the analysis and allows us to directly compare our results to the literature, one can still characterize equilibria in small markets under additional assumptions on the cost distribution $H$ (Proposition \ref{proposition:small_market}) or the prior distribution $F$ (Proposition \ref{proposition:small_market_binary}). In equilibrium, firms still place an atom at the highest reservation value $\overline{r}$ to retain consumers, while competing below the lowest reservation value $\underline{r}$. However, weaker competitive pressure reduces the incentive for complete unraveling. Instead, equilibrium information provision for low signals displays an alternating structure: intervals of full disclosure $(G=F)$ alternate with intervals where information is garbled to ensure $G^{n-1}$-linearity. Unlike large markets, where competitive pressure forces full disclosure at lower signals, small markets permit more localized pooling, reflecting the reduced intensity of competition.

To summarize, under these variations, the central equilibrium forces we identify---the balance between pooling incentives driven by search frictions and competitive pressures for more disclosure---remain robust. The differences mainly arise in local equilibrium details, specifically in how pooling is achieved and the extent of competitive unraveling at lower signals. By concentrating on large markets, we deliver clear, prior-independent comparative statics, as equilibrium structure in smaller markets becomes more sensitive to the underlying distributions. Moreover, our continuous search cost framework clearly distinguishes between the effects of increased average costs and increased cost heterogeneity, although the key insights extend naturally to any distributions with finite discontinuities.

\paragraph{Comparison with price competition literature}
We now compare our results to the classic literature on the price competition in search markets. A natural benchmark is provided by \textcite{stahl1996oligopolistic}.\footnote{We thank an anonymous referee for highlighting this connection.} He studies an oligopoly pricing game in a search market, under the same assumption on search costs as ours: a continuum of consumer types with $\supp(H)=[0,\overline{c}]$ and no mass of shoppers $(H(0)=0)$.  A striking parallel emerges between his result and ours. Specifically, he shows that a continuum of pure-strategy symmetric equilibria exists, including one exhibiting the Diamond paradox (i.e., monopoly price). This mirrors our finding that a continuum of upper-censorship equilibria exists, one of which features the informational Diamond paradox (i.e., no information disclosure). 

However, the mechanisms that sustain “non-Diamond” equilibria differ fundamentally, shedding light on how the nature of competition diverges when firms compete on price versus information. In \textcite{stahl1996oligopolistic}, the existence of an equilibrium with a price below the monopoly level depends critically on the \textit{local} density of shoppers, i.e., the value of $h(0)$. The logic is straightforward: when a firm slightly raises its price, it trades off revenue gains against a potential loss in market share. If $h(0)=0$, there are essentially no consumers at the margin to lose, so the firm always prefers to increase price, pushing the equilibrium toward the Diamond outcome. If $h(0)>0$, however, the threat of losing marginal consumers makes such deviations unprofitable, sustaining lower equilibrium prices.

In contrast, the maximally informativeness of equilibrium outcomes in our model depends on the \textit{global} structure of $H$, particularly the minimum average density $\underline{h}^{avg}$. The key distinction stems from the constraints firms face: in price competition, firms can choose any price freely, whereas in information competition, strategies must satisfy Bayes plausibility. Therefore, when a firm designs a partial-purchase signal targeting cost type $c$, it must reallocate probability mass from both high and low signals. This tradeoff involves gaining high-cost consumers ($1-H(c)$) at the expense of losing low-cost ones ($H(c)$). 
Moreover, the Bayes plausibility constraint implies that targeting lower $c$ (higher reservation value) requires shifting more mass away from the purchase signal $x=\overline{r}$, which can increase the cost of deviation. Hence, it is not the local behavior of $h$ near zero, but the global shape of $H$---captured by $\frac{H(c)}{c}$---that determines the profitability of deviations.\footnote{In \textcite{stahl1996oligopolistic}, all consumers stop upon their initial visit, whereas in our setting, the consumer actively searches with probability $F(a^M)>0$, again due to the Bayes plausibility constraint of our model: To send a high signal, you must reveal low signals as well. Such tradeoffs are not captured when choosing prices.}

\paragraph{Distributions with atom at $c=0$}
We demonstrate that well-known discontinuities in the price search literature—particularly those related to properties of the search cost distribution $H$ \parencite{stahl1989oligopolistic,stahl1996oligopolistic}—carry over to the context of competitive information disclosure. A concurrent paper by \textcite{boleslavsky2023information} studies a similar problem as ours but under different assumptions on $H$: specifically, they assume that search costs are discrete and that $H(0) > 0$ (i.e., there exist a positive mass of “shoppers”). Their setting can be viewed as the informational analogue of \textcite{stahl1989oligopolistic}, while our setup parallels \textcite{stahl1996oligopolistic}, as discussed above.

\textcite{boleslavsky2023information} show that a unique symmetric equilibrium always exists in their setting, and when competition is sufficiently intense, firms reveal no information below the reservation value. The key intuition is that the non-zero mass of shoppers---who visit all firms---intensify competitive pressure even at high signals. This stands in contrast to our setting, where such pressure vanishes at the top due to the absence of a mass point at zero.

Another important difference lies in the support of $H$: namely, the existence of a ‘gap’ around zero. While  \textcite{boleslavsky2023information} assumes that $\supp(H)$ contains a gap near zero (as $H$ is discrete), we assume no such gap—$H$ has full support on $[0, \overline{c}]$. The next result highlights the significance of this distinction, showing that full information disclosure may arise an equilibrium when $H(0)>0$ and there is no gap in $\supp(H)$.
\begin{proposition}\label{proposition:atom_full_info}
If there is a mass of shoppers ($H(0)>0$) and $H$ is weakly concave over its support $[0,\overline{c}]$,  there exists some $N>0$ such that full information $G=F$ is an equilibrium if $n>N$.
\end{proposition}
The above proposition is similar to \textcite{stahl1996oligopolistic}, who show that the Bertrand outcome can arise in the limit when $H(0)>0$ and $h(0)\rightarrow \infty$. In our informational counterpart, full information disclosure arises if $H(0)>0$ and $H$ is weakly concave. Once again, it is the global structure of $H$—not just local behavior near zero—that determines the equilibrium outcome.\footnote{To understand this result, recall that the gain from marginally increasing the signal can be decomposed into two components: the local monopoly gain, $(1-H(c_G(x)))’J_G(x)$, and the competition gain, $\Big(G^{n-1}\Big)'H(c_G(x))$. When $H(0)=0$,  the competition gain vanishes as $x \to \overline{r}$. But if $H(0)>0$, it remains positive even at the top. Furthermore, if $H$ is weakly concave, the local monopoly gain increases with $x$: higher signals attract more marginal consumers who stop, as there are more low-cost consumers. As a result, $D(x;F)$ is globally convex. The incentive to fully reveal informatinon at the top unravels downwards, resulting in full disclosure in equilibrium.}

In summary, the upper-censorship result of Theorem~\ref{thm:upper} (no shoppers, no gap), the non-disclosure at the bottom result of \textcite{boleslavsky2023information} (shoppers, gap), and the full-disclosure result of Proposition~\ref{proposition:atom_full_info} (shoppers, no gap) collectively show that both (i) the presence of an atom at zero and (ii) a gap in $\supp(H)$ can critically affect equilibrium behavior. This pattern closely parallels the classic price competition literature: depending on the structure of $H$, one may observe pure-strategy equilibria (no shoppers, no gap), mixed-strategy equilibria (shoppers, gap), or marginal-cost pricing (shoppers, no gap).

\section{Conclusion}
\label{section:conclusion}

This paper analyzes how competition and heterogeneous search costs jointly shape firms' incentives for information disclosure in search market. Firms balance the opposing incentives of competition and the local monopoly power granted by search frictions by adopting a two-faced strategy: aggressively providing information for low match values, while pooling information for high match values to induce consumers to stop searching, regardless of their search costs. This balance results in partial disclosure rather than complete transparency or full concealment.
We highlight that the evenness of the search cost distribution $H$ is the critical determinant of equilibrium informativeness. Firms can strategically induce different purchase probability across different consumer groups by providing partial-purchase signals. Greater evenness of the distribution weakens firms' ability to exploit consumer heterogeneity, resulting in more informative equilibria. 

\begin{onehalfspace}
	
\printbibliography
\end{onehalfspace}

\newpage
\appendix
\section{Omitted Proofs}
\label{appendix:a}
In Appendix \ref{appendix:a}, all proofs of the results are provided, 
except the equilibrium characterization of Theorem \ref{thm:upper} and Proposition \ref{proposition:charcteriztion}, which are deferred to Appendix \ref{appendix:b}. 
\paragraph{Proof of Proposition \ref{proposition_nofriction}}
See Corollary 1 and Theorem 1 of \textcite{hwang2023} for existence and uniqueness, respectively.\qed

\paragraph{Proof of Proposition \ref{proposition_homogeneous} and Proposition \ref{proposition:positive_c_h}}
We first prove that if $G\in MPC(F)$ is an equilibrium, then $G(r-)=0$ where $r=r(c;G)$. Let $M=\max(\supp(G))$. Assume $G$ is an equilibrium but $G(r-)>0$, and choose intervals $[a,b]\subset[0,r)$ and $[M-\delta,M]\subset (r,M]$ with positive mass under $G$. 
We construct a profitable deviation $G_{\epsilon}$, which collapses an $\epsilon\in(0,1)$ fraction of mass of $[a,b]$ and total the mass of $[M-\delta,M]$ to its barycenter $\mu_{\epsilon}$. Formally, for $\epsilon>0$, define $G_{\epsilon}:=(1-p_{\epsilon})G + p_{\epsilon}\delta_{\mu_{\epsilon}}$ where $p_{\epsilon}:=\epsilon G([a,b])+G([M-\delta,M])$ and $\mu_{\epsilon}:=\frac{\epsilon G([a,b])}{p_{\epsilon}}\int_{[a,b]}xdG(x)+\frac{G([M-\delta,M])}{p_{\epsilon}}\int_{[M-\delta, M]}xdG(x)$. 
Since $\lim_{\epsilon\rightarrow 0}\mu_{\epsilon}=\int_{[M-\delta, M]}xdG(x)>r$ and 
$\mu_{\epsilon}$ is continuous in $\epsilon$, 
there exist some $\epsilon_1>0$ such that 
$\mu_{\epsilon_1}>r$. Such $G_{\epsilon_1}$ is indeed a profitable deviation: $$\mathbb{E}_{G_{\epsilon_1}}[D(x,G)]-\mathbb{E}_{G}[D(x,G)]
= \frac{\epsilon_1 G([a,b])}{p_{\epsilon_1}}\left( D(r,G)-\int_{a}^{b}D(x,G)dG(x)\right)>0.
$$
The inequality holds because $D(x;G)$ weakly increases and exhibits a discrete jump at $x=r$.

Now assume $G(r-)=0$.
Then, $r=\mu-c$ and $D(x;G)=\frac{1}{n}\mathbb{I}(x\geq \mu-c)$.
Any best response to $G$ must put all the probability measure above $\mu-c$. 
Hence $G$ is a best response to itself. This proves Proposition \ref{proposition_homogeneous}.
Replacing $c$ with $\underline{c}$ and $r$ with $r(\underline{c};G)$ proves Proposition \ref{proposition:positive_c_h}.
\qed

\paragraph{Proof of Theorem \ref{thm:upper} and Proposition \ref{proposition:charcteriztion}}
See Appendix \ref{appendix:b}.\qed

\paragraph{Proof of Proposition \ref{proposition:surplus}}
We first start by showing for all $c>0$, there exists an unique censorship threshold $a_c$ such that the reservation value of a $c$-consumer under $U_{a_c}$ matches the censorship threshold itself.

\begin{lemma}
    For each $c>0$, define $a_c:=c_F^{-1}(c)$. Then, $a_c=r(c;U_{a_c})$. Furthermore, $a\leq r(c;U_a)$ if and only if $a\leq a_c$. 
    \label{lemma:fixed_point}
\end{lemma}
\begin{proof}
    Note for any feasible $G$, $c_G(x)=\int_x^1(1-G(t))dt=\mu-\int_0^x(1-G(t))dt$.
    Since $U_{a_c}=F$ over $[0,a_c]$, we have $c=c_{F}(a_c)=c_{U_{a_c}}(a_c)$. Taking the inverse, this shows $a_c=r(c;U_{a_c})$. For the second part, observe that $a\leq r(c;U_a)$ if and only if $c_{U_a}(a)\geq c$. Since $U_a=F$ over $[0,a]$, $c_{U_a}(a)=c_F(a)$. Hence, $a\leq r(c;U_a)$ if and only if $c_F(a)\geq c$, which is equivalent to $a\leq c_F^{-1}(c)=a_c$.
\end{proof}

We first prove (a). For each $c$-consumer with $c>0$, 
let $a_c$ as in Lemma \ref{lemma:fixed_point}. Define the expected value of the purchased product as $B_c(a)$
and the expected accumulated search cost as $C_c(a)$:
$$
B_c(a):=
\begin{dcases}
    \int_0^a xd(F(x)^n) + k_a(1-F(a)^n) & \text{ if } a\leq a_c \\ 
    \int_0^{a_c}xd(F(x)^n)+k_a(1-F(a_c)^n) & \text{ if } a> a_c
\end{dcases}
$$
$$
C_c(a):=
\begin{dcases}
    (1-F(a))\cdot \Big(c+2cF(a)+\dots +(n-1)cF(a)^{n-2}\Big)+ncF(a)^{n-1}
    =\frac{1-F(a)^n}{1-F(a)}c & \text{ if } a\leq a_c \\ 
    (1-F(a_c))\cdot \Big(c+2cF(a_c)+\dots +(n-1)cF(a_c)^{n-2}\Big)+ncF(a_c)^{n-1}
    =\frac{1-F(a_c)^n}{1-F(a_c)}c & \text{ if } a> a_c
\end{dcases}
$$
Define the consumer surplus of the $c$-consumer $CS_c(a):=B_c(a)-C_c(a)$,
and the total consumer surplus $CS(a):=\int_0^{\overline{c}}CS_c(a)dH(c)$. We show that $CS_c(a)$ is strictly positive and strictly increasing in $a$, which implies the same for $CS(a)$. Since $CS_c(0)=\mu-c>0$, it suffices to show $CS_c'(a)>0$. Since $\frac{d}{da}k_a>0$, $CS_c(a)$ is strictly increasing for $a> a_c$. For $a\leq  a_c$, direct calculation yields:
\begin{align*}
    CS_c'(a)&=f(a)\times\left( \frac{1-F(a)^n}{1-F(a)}-nF(a)^{n-1}\right)\times \left( k_a-\frac{c}{1-F(a)}-a\right)=f(a)\times nJ_F(a)\times\Big(r(c;U_a)-a \Big)
\end{align*}
The last bracket is strictly positive from Lemma \ref{lemma:fixed_point}.

To prove (b), observe $a^M_k\leq r(\overline{c};U_{a^M_k})$ for each $k=1,2$ and $a^M_1\leq a^M_2$ both holding is equivalent to $a^M_1\leq a^M_2\leq a_{\overline{c}}$. Defining $B(a):=\int_0^{\overline{c}}B_c(a)dH(c)$, we have $CS(a;H)=B(a)-\frac{1-F(a)^n}{1-F(a)}\mathbb{E}_H[c]$ for $a\leq a_{\overline{c}}$. Since $\mathbb{E}_{H_1}[c]=\mathbb{E}_{H_2}[c]$ by assumption, $CS(a^M_1;H_1)=CS(a^M_1;H_2)\leq CS(a^M_2;H_2)$ follows.
\qed

\paragraph{Proof of Proposition \ref{proposition:comp_stat_shift}}
We prove the result maintaining Assumption \ref{assumption:diag}. (The general case without the assumption is deferred to Online Appendix \ref{appendix:general_comp_stat}.)  Fix $H_1$ satisfying Assumption \ref{assumption:diag} and $\alpha\in(1,\frac{\mu}{\overline{c}}).$ Note an $\alpha$-scale shift $H_2$ of $H_1$ also satisfies the assumption. By Corollary \ref{corollary:a^M}, it suffices to show $h_2(c^M_2)\leq h_1(c^M_1)$. Define $S_k(c):=\frac{H_k(c)}{c}$ for $c>0$ and $S_k(0):=h_k(0)$. Then, $$h_2(c_2^M)=\min_{c\in[0,\alpha \overline{c}]} S_2(c) = \frac{1}{\alpha}\min_{c\in[0,\overline{c}]} S_1(c) < \min_{c\in[0,\overline{c}]} S_1(c)=h_1(c_1^M).$$
If $h_1(c_1^M)>\frac{1}{\mu}$, they by Proposition \ref{proposition:charcteriztion}, $a^1_M>0$. By Corollary \ref{corollary:a^M}, $a^M_2<a^M_1$. If $h_1(c_1^M)\leq \frac{1}{\mu}$, we obtain $a_1^M=a_2^M=0$.
\qed

\paragraph{Proof of Proposition \ref{proposition:comp_stat_fosd}}
Assume $H_1, H_2$ both satisfy Assumption \ref{assumption:diag} and $H_2$ first order stochastically dominates $H_1$. Since $H_2(c)\leq H_1(c)$ for all $c$, we have $S_2(c)\leq S_1(c)$, implying $h_2(c^M_2)=\min_{c\in[0,\overline{c}]}S_2(c)\leq \min_{c\in[0,\overline{c}]}S_1(c)=h_1(c^M_1).$
By Corollary \ref{corollary:a^M}, $a^M_2\leq a^M_1$.
If $H_2(c)<H_1(c)$ for all $c$, then the inequality is strict, hence implying $a^M_2<a^M_1$ whenever $a^M_1>0$.
\qed

\paragraph{Proof of Corollary \ref{corollary:convergence}}
Assume each $H_k$ satisfies Assumption \ref{assumption:diag}.
(The general case without the assumption is deferred to Online Appendix \ref{appendix:general_comp_stat}.) Suppose $H_k\rightarrow \delta_0$ weakly and $\supp(H_k)\searrow \{0\}$. We claim that $h_k(c^M_k)\rightarrow \infty$, which, by continuity of $c_F(\cdot)$, implies $a^M_k=c_F^{-1}\left(\frac{1}{h_k(c^M_k)}\right)\rightarrow c_F^{-1}(0)=1$.

Assume to the contrary that $\{h_k(c^M_k)\}_{k=1}^{\infty}$ is bounded above 
by some constant $M$. Then, $H_k(c^M_k)\leq Mc^M_k$ for all $k$.
Since $\supp(H_k)\searrow \{0\}$, we have $c^M_k\rightarrow 0$, and thus $H_k(c_k^M)\rightarrow 0$.
Define the indicator functions $f_k=\mathbb{I}[c\leq c_k^M]$, so that $\int f_kdH_k=H_k(c^M_k)\rightarrow 0$.
Also, since $\supp(H_k)\searrow \{0\}$, we have $\liminf_{k\rightarrow\infty, c\rightarrow c'}f_k(c')=\mathbb{I}[c\leq 0]$.
Applying Fatou's lemma for weakly converging measures (Theorem 1.1, \textcite{feinberg2014fatou}) yields:
$$1=\int\mathbb{I}[c\leq 0]d\delta_0=
\int \liminf_{k\rightarrow\infty, c\rightarrow c'}
f_k(c')d\delta_0
\leq 
\liminf_{k\rightarrow \infty}\int f_kdH_k 
=\liminf_{k\rightarrow\infty}H(c_k^M)=0,$$
which is a contradiction.
Hence $h_k(c_k^M)\rightarrow \infty$.
\qed

\paragraph{Proof of Corollary \ref{corollary:convergence_homo}}
Assume the sequence $\{H_k\}_{k=1}^{\infty}$ converges to $\delta_{\overline{c}}$ weakly and $\supp(H_k)=[0,\overline{c}]$ for all $k$. We show that for large enough $k$, each $H_k$ satisfies Assumption \ref{assumption:diag} and $h_k(c^M_k)\leq \frac{1}{\mu}$, which by Proposition \ref{proposition:charcteriztion} implies $a^M_k=0$. Fix $\e>0$. Since $H_k\rightarrow \delta_{\overline{c}}$ weakly and the point of discontinuity of $\mathbb{I}[c\leq \overline{c}-\e]$ is measure zero under $\delta_{\overline{c}}$,
$$\int_0^1\mathbb{I}[c\leq \overline{c}-\e]dH_k(c)=H_k(\overline{c}-\e)\rightarrow \int_0^1\mathbb{I}[c\leq \overline{c}-\e]d\delta_{\overline{c}}(c)=0.$$
Hence, there exists some $K>0$ such that for all $k>K$, $H_k(\overline{c}-\e)<(\overline{c}-\e)\frac{1}{\mu}$ which implies $\frac{H_k(\overline{c}-\e)}{\overline{c}-\e}<\frac{1}{\mu}<\frac{1}{\overline{c}}$. This inequality guarantees that Assumption \ref{assumption:diag} holds for $H_k$, and moreover, $h_k(c^M_k)<\frac{1}{\mu}$.

\qed

\paragraph{Proof of Proposition \ref{proposition:even}}
We show part (a). Part (b) is deferred to Appendix \ref{appendix:general_comp_stat}. Assume $H_0$ satisfies Assumption \ref{assumption:diag}. By Corollary \ref{corollary:a^M}, it suffices to show $h_{\lambda}(c^M_{\lambda})\geq h_0(c^M_{0})$ and that $h_{\lambda}(c^M_{\lambda})$ is strictly increasing in $\lambda$.
Note the average density associated with $H_{\lambda}$ is given by $S_{\lambda}(c)=\frac{\lambda}{\overline{c}}+(1-\lambda)S_0(c)$, so  $H_{\lambda}(c)$ also satisfies the assumption. Because $H_0\neq U[0,\overline{c}]$, Assumption \ref{assumption:diag} implies $\min_c S_0(c)=h_0(c_0^M)< \frac{1}{\overline{c}}$.
Hence, for all $\lambda\in(0,1]$, 
$$h_0(c_0^M)< \lambda \frac{1}{\overline{c}}+(1-\lambda)h_0(c_0^M)
=\min_c S_{\lambda}(c)=h_{\lambda}(c_{\lambda}^M).$$
Finally, differentiating with respect to $\lambda$ gives $\frac{d}{d\lambda}h_{\lambda}(c^M_{\lambda})=\frac{1}{\overline{c}}-h_0(c_0^M)>0$.

\paragraph{Proof of Proposition \ref{proposition:mps}}
Assume $H_1$ and $H_2$ both satisfy Assumption \ref{assumption:diag} and $H_2$ is an MPS of $H_1$. The general case without the assumption is proved in Online Appendix \ref{appendix:general_comp_stat}. Under Assumption~\ref{assumption:diag}, the global minimum $c^M_k\in[0,\overline{c})$ of $S_k(c)$ must satisfy the first- and second-order conditions: $S_k(c_k^M)=h_k(c^M_k)$ and $h_k'(c^M_k)\geq 0$, respectively.

To prove (a), assume $H_1$ and $H_2$ both admit strictly quasi-convex densities. Note $H_2$ single crosses $H_1$ from above at some $s\in(0,\overline{c})$, i.e. $H_2(c)>H_1(c)$ for $c\in(0,s)$ and $H_2(c)<H_1(c)$ for $c\in(s,\overline{c})$. Note each $c^M_k$ must both lie above $s$. Hence, we have $S_2(c^M_2)\leq S_2(c^M_1)<S_1(c^M_1)$, i.e. $h_2(c^M_2)<h_1(c^M_1)$. By Corollary~\ref{corollary:a^M}, $a^M_2\leq a^M_1$.

To prove (b), assume both $H_1$ and $H_2$ have strictly quasi-concave densities. We show $c^M_k=0$ for $k=1,2$, which proves the desired result: since $H_2$ is a MPS of $H_1$, it follows that $h_2(0)\geq h_1(0)$, and thus $a^M_2\geq a^M_1$ by Corollary \ref{corollary:a^M}. Because $H_k$ admits a strictly quasi-concave density with some interior peak $p_k\in(0,\overline{c})$, the density $h_k$ is strictly increasing over $[0,p_k]$ and strictly decreasing over $[p_k,\overline{c}]$. By the mean value theorem, for every $c\in(0,p_k)$, there exists some $m_c\in (0,c)$ such that $\frac{H_k(c)}{c}=S_k(c)=h_k(m_c)$. It follows that $S_k'(c)=\frac{1}{\overline{c}}\left(h_k(c)-S_k(c)\right)=\frac{1}{\overline{c}}\left(h_k(c)-h_k(m_c)\right)>0$, where the last inequality holds because $h_k$ is increasing over $[0,p_k]$. Hence, $S_k(c)$ is strictly increasing over $(0,p_k)$ and $m_c<c$. Since $h_k'(c^M_k)\geq 0$ by the second order condition, we must have $c^M_k\in[0,p_k]$. But $S_k$ is strictly increasing over this interval, so the global minimum is attained only at the left endpoint. Therefore, $c^M_k=0$. \qed

\paragraph{Proof of Proposition \ref{proposition:atom_full_info}}
Let $m_f$ be the minimum of $f$ over $[0,1]$ and $M_{f'}$, $M_f$ and $M_h$ be the maximum of $f'$, $f'$, $h$ over the same interval, which exists by assumption of continuity.
Observe that
\begin{align*}
    D''(x;F)&=(n-1)F^{n-3}\Big[
    (n-1)f^2 H(c_F)+f'FH-2(1-F)F fh(c_F)
    \Big]-(1-F)^2 J_F h'(c_F) \\ 
    & \geq  (n-1)F^{n-3}\Big[
        (n-1)m_f^2 H(0) -M_{f'}H(0)-2M_fM_h
    \Big]
\end{align*}
The inequality follows from the weak concavity of $H$, $h'\leq 0$.
The term in the square bracket is strictly positive for large enough $n$ since $m_f>0$. Hence, $D(x;F)$ is strictly convex over $x\in[0,1]$, implying $F$ is an equilibrium.

\section{Existence and Uniqueness of the Upper Censorship Equilibria}
\label{appendix:b}
This appendix aims to prove the existence and uniqueness of upper censorship equilibria. Section~\ref{appendix:b_1} provides preliminary analysis essential for the subsequent results. Section~\ref{appendix:b_2} proves Theorem~\ref{thm:upper}-(b), showing that if an upper censorship equilibrium exists, there must be a continuum of them. Section~\ref{appendix:b_3} proves the existence (Theorem~\ref{thm:upper}-(a)) by explicitly constructing the maximally informative upper censorship equilibrium, without relying on Assumption~\ref{assumption:diag} and thus extending the results to arbitrary cost distributions $H$ (Theorem~\ref{thm:appendix_maximal} generalizes Proposition~\ref{proposition:charcteriztion}). Finally, Section~\ref{appendix:b_4} establishes that upper censorship equilibria are the only limit equilibria.

First, we begin with Lemma \ref{appendix_lemma:dworczak_martini}, which restates relevant results in \textcite{dworczak2019simple}.

\begin{lemma}[\textcite{dworczak2019simple}]
    Assume $u:[0,1]\rightarrow \mathbb{R}$ is 
    continuous and has bounded slope. Then,
    \begin{equation*}
        G^*\in \argmax_{G\in MPC(F)}\int_0^1 u(x)dG(x)
        \label{appendix_eq:original_DM}
    \end{equation*}
    if and only if there exists some convex function 
    $\phi:[0,1]\rightarrow \mathbb{R}$ such that (a) $\int_0^1 \phi dF=\int_0^1 \phi dG^*$,  (b) $\phi(x)\geq u(x)$, and (c) $\supp(G^*)\subseteq \{x|\phi(x)=u(x)\}$.
    Moreover, for any $[a,b]\subseteq [0,1]$:
    \begin{enumerate}
        \item if $\phi$ is strictly convex in $[a,b]$, then $\phi(x)=u(x)$ and $G^*(x)=F(x)$ for all $x\in[a,b]$, or
        \item if $\phi$ is affine in $[a,b]$ and $[a,b]$ is the maximal interval on which $\phi$ is affine,
        then 
        \begin{equation}
            G(a)=F(a),\;\; G(b)=F(b),\;\; \int_a^btG(t)=\int_a^btF(t),\;\; \phi(c)=u(c) \text{ for some } c\in[a,b]    
            \label{eq:strict_MPC}
        \end{equation}
    \end{enumerate}
    \label{appendix_lemma:dworczak_martini}
\end{lemma}

We introduce multiple definitions:
\begin{itemize}
    \item When the distribution $G$ satisfies the conditions in \eqref{eq:strict_MPC}, we say
    $G$ is a \textit{strict MPC of $F$ over $[a,b]$}, or just \textit{strict MPC} when the context is clear.
    \item Define the average slope of $H(\cdot)$ in $[0,c]$ as 
    $S(c):=
    \begin{dcases}
        h(0) & \text{ if } c=0 
        \\ \frac{H(c)}{c} & \text{ if } c>0
    \end{dcases}
    $
    \item Define $c_{cav}\in [0,\overline{c}]$ as
    \begin{equation}
        c_{cav}:=\inf_c\{ c: h'(t)\leq 0 \text{ for all } t\in (c,\overline{c}]\}\label{eq:c_{cav}}.
    \end{equation}
    If the set in \eqref{eq:c_{cav}} is empty, we define $c_{cav}:=\overline{c}$. If the set in \eqref{eq:c_{cav}} is nonempty, $[c_{cav},\overline{c}]$ characterizes the maximal interval including $\overline{c}$ where $H$ is concave. 
    \item Define $\overline{a}:=c_F^{-1}(\overline{c})$. By Lemma \ref{lemma:fixed_point}, $a\leq \underline{r}$ if and only if $a\leq \overline{a}$.
\end{itemize}

\subsection{Preliminary Analysis}
\label{appendix:b_1}

In this section, we prove the following proposition, which serves a key role in characterizing and verifying whether an upper censorship distribution $U_a$ is an equilibrium. It states that when the competition is sufficiently intense, it suffices to check conditions related to the cost distribution $H$.

\begin{proposition}
    There exists some $N>0$ such that $U_a$ is an equilibrium for $n>N$ 
    if and only if
    \begin{enumerate}[label=(\alph*)]
        \item $a\leq \overline{a}$ and $S(c)\geq\frac{1}{c_F(a)}$ for all $c\in[0,\overline{c}]$.
        \item $a>\overline{a}$ and the following conditions all hold:
        \begin{enumerate}[label=(\roman*)]
            \item $S(c)\geq S(c_F(a))>h(c_F(a))$ for all $c\in[0,c_F(a)]$
            \item $c_F(a)\geq c_{cav}$
        \end{enumerate}
    \end{enumerate}
    \label{proposition:full_iff_cost}
\end{proposition}





\paragraph{Proof of Proposition \ref{proposition:full_iff_cost}} We prove the proposition in multiple steps. First, we show that verifying whether $U_a$ is an equilibrium reduces to analyzing the virtual interim demand $\phi_a$. Then, we establish lemmas connecting the properties of $\phi_a$ to those of the cost distribution $H$. 

\begin{lemma}
    For any $a>0$, define the auxiliary function $\phi_a$ as follows:
    $$\phi_a(x):=
    \begin{cases}
        D(x;U_a) & x\leq a \\ 
        \dfrac{D(k_a;U_a)-D(a;U_a)}{k_a-a}(x-a)+D(a;U_a) & x>a 
    \end{cases}
    $$
    then, $U_a$ is an equilibrium if and only if (i) $\phi_a$ is convex and (ii) $\phi_a(x)\geq D(x;U_a)$.
    \label{lemma:iff}
\end{lemma}
\begin{proof}
    The if direction is trivial by construction and Lemma \ref{appendix_lemma:dworczak_martini}. To prove the only if part, assume $U_a$ is an equilibrium. Then, there exists some function $\phi$ satisfying conditions of Lemma \ref{appendix_lemma:dworczak_martini}. We prove $\phi=\phi_a$. Clearly, $\phi(x)=D(x;U_a)=\phi_a(x)$ over $[0,a]$. Since $U_a$ is a strict MPC of $F$ over $[a,1]$, $\phi$ must be linear over $[a,1]$. Furthermore, $\{a,k_a\}\subseteq \{x|\phi(x)=D(x;U_a)\}$ implies $\phi=\phi_a$ at $x=a$ and $x=k_a$. Since both $\phi, \phi_a$ are both linear over $[a,1]$ and coincide in two points of the interval, they must be equal over $[a,1]$. Hence, $\phi=\phi_a$.
\end{proof}

We can restate the conditions of $\phi_a$ from Lemma \ref{lemma:iff} in terms of $H$. Since $\phi_a$ is linear over $[a,1]$, establishing the convexity of $\phi_a$ over $[0,1]$ reduces to verifying whether $\phi_a$ is (i) convex over $[0,a]$, and (ii) has an increasing kink at $x=a$. Each of these conditions are established in Lemma \ref{lemma:phi_a_convex_F} and Lemma \ref{lemma:phi_a_convex_kink}. Additionally, Lemma \ref{lemma:phi_a_geq_D} establishes conditions ensuring $\phi_a\geq D$.

If $a\geq \underline{r}$, then the convexity of $\phi_a(x)=D(x;U_a)$ over $x\in[\underline{r},a]$ is directly tied to the curvature of the search cost distribution $H$, specifically the sign of the second derivative $h'(c)$. This relationship is formalized in the lemma below.



\begin{lemma}
    Let $G\in MPC(F)$ and fix some $x\geq \underline{r}=r(\overline{c};G)$ such that $(x-\e, x+\e)\subset \supp(G)$ for some $\e>0$. Assume $D(x;G)$ is twice differentiable at $x$ for some $n=m$. Then, there exists some $N_{x,G}>0$ such that if $n>N_{x,G}$,
    \begin{enumerate}[label=(\alph*)]
        \item if $h'(c_G(x))>0$, then $D''(x;G)<0$.
        \item if $h'(c_G(x))\leq 0$, then $D''(x;G)>0$
    \end{enumerate}
    Furthermore, if $h'(c)\leq 0$ for all $c\in[c_F(x),\overline{c}]$,
    then there exists some $N>0$ such that $D''(t;F)>0$ for all $t\in[\underline{r},x]$.
    \label{lemma:partial_stop_convex}
\end{lemma}

\begin{proof}
    See Online Appendix \ref{online_appendix:proofs}.
\end{proof}

\begin{lemma}[Convexity of $\phi_a(x)$ over {$[0,a]$}]
    Fix $a$ and let $\underline{r}=r(\overline{c};U_a)$.
    \begin{enumerate}[label=(\alph*)]
        \item There exists some $N>0$ such that if $n>N$, 
        $\phi_a(x)$ is convex over $[0,\min(a,\underline{r})]$.
        \item Assume $a\in(\overline{a},1)$.
        There exists some $N>0$ such that if $n>N$, 
        $\phi_a$ is convex over $[0,a]$ if and only if $h'(c_F(x))\leq 0$ for all $x\in[\underline{r},a]$.
    \end{enumerate}
    \label{lemma:phi_a_convex_F}
\end{lemma}
\begin{proof}
    Since $\phi_a(x)=F^{n-1}(x)$ over the interval $[0,\min(a,\underline{r})]$, we prove the convexity of $F^{n-1}$ for large enough $n$. Direct calculation gives 
    $$
    (F^{n-1})'' = 
    (n-1)F^{n-1}\Big[(n-2)f^2+f'F\Big]\geq (n-1)F^{n-1}\Big[(n-2)m_f^2-M_{f'}\Big],
    $$
    where $m_f:=\min_{x\in[0,1]}f(x)$ and $M_{f'}:=\max_{x\in[0,1]}f'(x)$, which are guaranteed to exist because $f,f'$ is continuous over a compact interval $[0,1]$. Strictly positive density $f$ implies $m_f>0$, ensuring the term inside the square bracket is strictly positive for large enough $n$.
    This proves part (a) of the lemma. To prove (b), it suffices to prove the convexity of $\phi_a$ over $[\underline{r},a]$. Since $D(x;U_a)=D(x;F)$ over $[\underline{r},a]$, the result follows from Lemma \ref{lemma:partial_stop_convex}.
\end{proof}
Part (a) and (b) of the following lemma links the kink of $\phi_a$ at $x=a$ to the cost distribution $H$. Part (c) proves that if $\phi_a$ is convex over $[0,1]$, so is $\phi_b$ for all $b<a$.
\begin{lemma}[Convexity of $\phi_a(x)$ at $x=a$]
    There exists some $N>0$ such that
    $\lim_{x\rightarrow a+}\phi'_a(x)\geq \lim_{x\rightarrow a-}\phi'_a(x)$ for all $n>N$
    if and only if either
    \begin{enumerate}[label=(\alph*)]
        \item $a\leq\overline{a}$ or
        \item $a> \overline{a}$ and $S(c_F(a))>h(c_F(a))$.
    \end{enumerate}
    Moreover,
    \begin{enumerate}[label=(\alph*), start=3]
        \item Assume $\lim_{x\rightarrow a+}\phi_a'(x)\geq \lim_{x\rightarrow a-}\phi_a'(x)$ and $D(x;F)$ is convex over $[0,a]$, then
        \begin{enumerate}[label=(\roman*)]
            \item $\lim_{x\rightarrow b+}\phi_b'(x)\geq \lim_{x\rightarrow b-}\phi_b'(x)$ holds for all $b<a$.
            \item The right derivative $\lim_{x\rightarrow b+}\phi_b'(x)$ is a increasing function in $b\in[0,a]$ and $\frac{H(c_F(b))}{c_F(b)}$ is a increasing function in $b\in[\overline{a},a]$.
        \end{enumerate}  
    \end{enumerate}
    \label{lemma:phi_a_convex_kink}
\end{lemma}

\begin{proof}
    We first prove the case of $a<\overline{a}$.
    Observe that for any $a>0$, 
    \begin{equation}
        \lim_{n\rightarrow \infty}\frac{J_F(a)}{(n-1)F(a)^{n-2}f(a)}=\infty
        \label{J(x)_diverge}
    \end{equation}
    since the denominator vanishes exponentially faster than the numerator $J_F(a)=\frac{1-F(a)^n}{n(1-F(a))}-F(a)^{n-1}$. Calculation gives $\lim_{x\rightarrow a+}\phi_a'(x)=\frac{J_F(a)}{k_a-a}$ and $\lim_{x\rightarrow a-}\phi_a'(x)=(n-1)F(a)^{n-2}f(a)$. Hence, $\lim_{x\rightarrow a+}\phi_a'(x)\geq \lim_{x\rightarrow a-}\phi_a'(x)$ if and only if $\frac{J_F(a)}{(n-1)F(a)^{n-2}f(a)}\geq k_a-a$, which holds for large enough $n$ from \eqref{J(x)_diverge}.

    Now assume $a\geq \overline{a}$ and $S(c_F(a))=\frac{H(c_F(a))}{c_F(a)}>h(c_F(a))$.
    We need to show
    \begin{equation*}
        \lim_{x\rightarrow a+}\phi_a'(x) = \frac{J_F(a)}{k_a-a}H(c_F(a))\geq  \lim_{x\rightarrow a-}\phi_a'(x)=
        (1-F(a))h(c_F(a))J_F(a)+ (n-1)F(a)^{n-2}f(a)H(c_F(a))
    \end{equation*}
    The inequality holds if and only if 
    $
    \frac{J_F(a)}{(n-1)F(a)^{n-2}f(a)}(1-F(a))\left(\frac{H(c_F(a))}{c_F(a)}-h(c_F(a)) \right)\geq H(c_F(a)).
    $
    The term inside the bracket is strictly positive by assumption. 
    Hence from \eqref{J(x)_diverge}, the inequality holds for large enough $n$.

    The proof of (c) is deferred to the Online Appendix~\ref{online_appendix:proofs}; we only sketch the argument here.
    Define $R(a)$ and $L(a)$ as the right and left derivatives of $\phi_a(x)$ at $x=a$, respectively. The kink size $\Delta(a):=R(a)-L(a)$ can be explicitly characterized by solving a first-order ordinary differential equation $\Delta'(a)=\frac{1}{k_a-a}\Delta(a)-L'(a)$. Using the boundary condition $\Delta(a)\geq 0$ and the monotonicity of $L(t)$, we show $\Delta(b)\geq 0$ for all $b<a$, which proves part (i). A straightforward calculation then establishes part (ii) as a corollary.
\end{proof}

\begin{lemma}[Conditions for $\phi_a\geq D$]
    \begin{enumerate}[label=(\alph*)]
        \item If $a<\overline{a}$, then $\phi_a(x)\geq D(x;U_a)$ for all $x$ 
        if and only if $S(c)\geq\frac{1}{c_F(a)}$ for all $c\in[0,\overline{c}]$.
        \item If $a\geq \overline{a}$, then $\phi_a(x)\geq D(x;U_a)$ for all $x$ 
        if and only if $S(c)\geq S(c_F(a))$ for all $c\in[0,c_F(a)]$.
    \end{enumerate}
    \label{lemma:phi_a_geq_D}
\end{lemma}

\begin{proof}
    We only need to check over the interval $[a,1]$, where
    $\phi_a(x)=\frac{J_F(a)H(c_{U_a}(a))}{k_a-a}(x-a)+D(a;U_a).$
    First, consider $a<\overline{a}$. This implies $H(c_{U_a}(a))=1$. 
    Thus, $\phi_a(x)\geq D(x;U_a)$ for all $x\in[a,1]$ if and only if 
    \begin{equation}
        \frac{J_F(a)}{k_a-a}(x-a)\geq J_F(a)(1-H(c_{U_a}(x))) \quad \forall x\in[a,k_a].
        \label{lemma:a_3_1}
    \end{equation}
    Using $c_{U_a}(x)=(F(a)-1)(x-k_a)$ for $x\in[a,k_a]$, we rewrite the left-hand side as  $\frac{1}{k_a-a}(x-a)=1+\frac{c_{U_a}(x)}{(F(a)-1)(k_a-a)}$. Therefore, \eqref{lemma:a_3_1} holds if and only if $\frac{H(c_{U_a}(x))}{c_{U_a}(x)}\geq \frac{1}{(F(a)-1)(a-k_a)}=\frac{1}{c_F(a)}$ for all $x$. Therefore, the inequality holds if and only if $S(c)=\frac{H(c)}{c}\geq \frac{1}{c_F(a)}$ for all $c$.

    If $a\geq \overline{a}$, then $H(c_{U_a}(a))<1$. The condition $\phi_a(x)\geq D(x,U_a)$ becomes
    \begin{equation}
        \frac{J_F(a)}{k_a-a}H(c_{U_a}(a))(x-a)\geq J_F(a)(H(c_{U_a}(a))-H(c_{U_a}(x))).
        \label{lemma:a_3_2}
    \end{equation}
    Using the same substitution as in (a), this simplifies to $\frac{H(c_{U_a}(x))}{c_{U_a}(x)}\geq \frac{H(c_{U_a}(a))}{c_{U_a}(a)}=\frac{H(c_F(a))}{c_F(a)}$ for all $x\in[a,\overline{r}]$. Under $c_{U_a}(\cdot)$, this range corresponds to $c\in[0,c_F(a)]$. The lemma is proved.
\end{proof}

Now, Proposition \ref{proposition:full_iff_cost} follows from Lemmas \ref{lemma:iff}, \ref{lemma:phi_a_convex_F}, \ref{lemma:phi_a_convex_kink} and \ref{lemma:phi_a_geq_D}. \qed

\subsection{Proof of Theorem \ref{thm:upper}-(b)}
\label{appendix:b_2}

Assume $U_a$ is an equilibrium. Let $b<a$. Since $\phi_a(x)$ is convex, $\phi_b(x)=\phi_a(x)$ is convex over $[0,b]\subseteq [0,a]$. Furthermore, by Lemma~\ref{lemma:phi_a_convex_kink}-(c), $\phi_b$ exhibits an increasing kink at $x=b$. By Lemma~\ref{lemma:iff}, it remains to show $\phi_b(x)\geq D(x;U_b)$ for all $x$ to prove $U_b$ is an equilibrium. If $a\leq \overline{a}$, then by Lemma~\ref{lemma:phi_a_geq_D}-(a), we have $S(c)\geq \frac{1}{c_F(a)}$ for all $c\in[0,\overline{c}]$. Since $c_F(\cdot)$ is a strictly decreasing function, for any $b<a$, we have $S(c)\geq \frac{1}{c_F(a)}>\frac{1}{c_F(b)}$ for all $c$. Hence, $\phi_b(x)\geq D(x;U_b)$ by Lemma~\ref{lemma:phi_a_geq_D}-(a). 
If $a>\overline{a}$, then by Lemma~\ref{lemma:phi_a_geq_D}-(b), $S(c)\geq S(c_F(a))$ for all $c\in[0,c_F(a)]$. Lemma~\ref{lemma:phi_a_convex_kink}-(c) implies $S(c_F(a))\geq S(c_F(b))$ for $\overline{a}\leq b\leq a$. Plugging $b=\overline{a}$, this implies $S(c)\geq S(\overline{c})=\frac{1}{\overline{c}}$ for all $c$. Hence, $S(c)\geq \frac{1}{\overline{c}}>\frac{1}{c_F(b)}$ if $b<\overline{a}$. This prove $\phi_b(x) \geq D(x;U_b)$ for all $x$ by Lemma~\ref{lemma:phi_a_geq_D}-(b).\qed

\subsection{Existence and Characterization of Upper Censorship Equilibria
    (Theorem \ref{thm:upper}-(a) and Proposition \ref{proposition:charcteriztion})}
\label{appendix:b_3}

We now construct the maximally informative upper-censorship equilibrium $U_{a^M}$, establishing existence as stated in Theorem~\ref{thm:upper}-(a). Theorem~\ref{thm:appendix_maximal} generalizes Proposition~\ref{proposition:charcteriztion} by characterizing the maximally informative upper-censorship equilibria without relying on Assumption~\ref{assumption:diag}. According to Proposition~\ref{proposition:full_iff_cost}, identifying an equilibrium $U_a$ reduces to characterizing the set of feasible values $c_F(a)$ that satisfy the conditions of Proposition~\ref{proposition:full_iff_cost}.

Define $m(H)$ as the set of critical points $c_0\in[0,\overline{c}]$ where $S'(c_0)=0$ and $S(c_0)$ is global minimum of $S$ over the interval $[0,c_0]$, and define $c_{loc}$ as follows:
\begin{eqnarray}
    m(H)=\left\{c_0: S'(c_0)=0 \text{ and } S(c)\geq S(c_0) \text{ for all } c\leq c_0\right\}
\end{eqnarray}
\begin{equation}
    c_{loc} =
\begin{dcases}
    c_{cav} & \text{ if $m(H)=\emptyset$ or $m(H)=\{\overline{c}\}$} \\ 
    \max\Big(\argmin_{c\in m(H)} S(c) \Big)& \text{ if $m(H)\neq \emptyset$ and $m(H)\neq \{\overline{c}\}$}
\end{dcases}
\label{eq:c_{loc}}
\end{equation}
Since $S'$ is continuous, $m(H)$ is closed and $c_{loc}$ is well-defined.
If $m(H)$ is nonempty, then $c_{loc}$ is the point at which $S$ attains the smallest local minimum.
If there are multiple points where the smallest local minimum is attained, we define $c_{loc}$ be the largest minimizer.
If $m(H)=\emptyset$ or if it is the boundary case of $m(H)=\{\overline{c}\}$, we set $c_{loc}:=c_{cav}$. 
Note $m(H)=\emptyset$ or $m(H)=\{\overline{c}\}$ if and only if $H$ is convex over $[0,\overline{c}]$ and strictly convex over some subinterval.
We first establish the properties of $c_{loc}$.

\begin{lemma}
    Assume $m(H)\neq \emptyset$ and $c_{loc}<\overline{c}$.
    \begin{enumerate}[label=(\alph*)]
        \item 
            Then, the equation 
            \begin{equation}
                S(c)=S(c_{loc})\label{eq:cross}
            \end{equation} 
            admits either one or no solution in $c\in(c_{loc},\overline{c}]$.
        \item The solution $c_{sol}\in (c_{loc}, \overline{c})$ to \eqref{eq:cross} exists if and only if $S(c_{loc})>\frac{1}{\overline{c}}$. Moreover, $c_{sol}=\overline{c}$ if and only if $S(c_{loc})=\frac{1}{\overline{c}}$.
        \item 
            $S'(c_{sol})<0$ (equivalently, $S(c_{sol})> h(c_{sol})$).
            Moreover, $S'(c)<0$ for all $c\in(c_{sol},\overline{c}]$.
        \item 
            $S(c)\geq S(c_{sol})$ for all $c\in [0,c_{sol}]$ with the equality holding only at $c\in\{c_{loc},c_{sol}\}$.
            Moreover, if $S(c_{loc})\leq \frac{1}{\overline{c}}$, then $S(c)\geq S(c_{loc})$ for all $c\in[0,\overline{c}]$,
            with $S(c)>S(c_{loc})$ for all $c\in(c_{loc},\overline{c}]$.
    \end{enumerate}
    \label{lemma:c_{loc}_1}
\end{lemma}

\begin{proof}
    See Appendix \ref{online_appendix:proofs}.
\end{proof}

Note if $S(c_{loc})>\frac{1}{\overline{c}}$, choosing $a$ such that $c_F(a)=\max(c_{sol}, c_{cav})$ satisfies the conditions in Proposition \eqref{proposition:full_iff_cost}-(b). The condition $S(c)\geq S(c_F(a))>h(c_F(a))$ follows from directly Lemma~\ref{lemma:c_{loc}_1}-(c,d), while $c_F(a)\geq c_{cav}$ follows by construction. If $S(c_{loc})\leq \frac{1}{\overline{c}}$ and $m(H)\neq \emptyset,\{\overline{c}\}$, this is equivalent to Assumption~\ref{assumption:diag} in the text, implying that Proposition~\ref{proposition:charcteriztion} corresponds to Theorem~\ref{thm:appendix_maximal}-(a,b).

\begin{theorem}
    Fix $N$ large enough that Lemma \ref{lemma:phi_a_convex_F} and Lemma \ref{lemma:phi_a_convex_kink} holds. Then, the maximally informative upper censorship equilibrium threshold $a^M$ is characterized as follows.
    \begin{enumerate}[label=(\alph*)]
        \item If $m(H)\neq \emptyset$ and $S(c_{loc})\leq \frac{1}{\mu}$, 
        $$a^M=0$$
        \item If $m(H)\neq \emptyset,\{\overline{c}\}$ and $\frac{1}{\mu}<S(c_{loc})\leq \frac{1}{\overline{c}}$,
        $$a^M=c_F^{-1}\left( \frac{1}{S(c_{loc})}\right)$$
        \item If $m(H)\neq \emptyset$ and $S(c_{loc})>\frac{1}{\overline{c}}$, then 
        $$a^M=c_F^{-1}\left( \max(c_{cav}, c_{sol})\right)$$
        \item If $m(H)=\emptyset$ or $m(H)=\{\overline{c}\}$
        $$a^M=c_F^{-1}(c_{cav})$$
    \end{enumerate}
    \label{thm:appendix_maximal}
\end{theorem}

\subsubsection{(a): $m(H)\neq \emptyset$ and $S(c_{loc})\leq \frac{1}{\mu}$}
\noindent\textbf{Equilibrium Check}:
$a^M=0$ is identical to $G=\delta_{\mu}$. Let $\phi(x):=\frac{1}{n}$, then, 
$\phi(x)\geq D(x,G)$, $\supp(G)=\{\mu\}\subset \{x|\phi(x)=D(x;G)\}=[\mu,1]$,
and $\int_0^1\phi(x)dG(x)=\int_0^1\phi(x)dF(x)$.
By Lemma \ref{appendix_lemma:dworczak_martini}, $G_{a^M}$ is an equilibrium.

\noindent\textbf{Maximallity}:
By Lemma \ref{lemma:c_{loc}_1}-(d), $S(c)\geq S(c_{loc})$ holds for all $c\in[0,\overline{c}]$.
For any $a>0$, we have $c_F(a)<\mu$ by definition.
Since $S(c_{loc})\leq \frac{1}{\mu}$ whereas $\frac{1}{c_F(a)}>\frac{1}{\mu}$,
there exists some $c$ such that $S(c)<\frac{1}{c_F(a)}$.
Hence, $U_a$ with $a\in(0,\overline{a}]$ cannot be an equilibrium by Proposition \ref{proposition:full_iff_cost}. This implies that any $U_a$ with $a\in(\overline{a},1]$ also cannot be an equilibrium by Theorem~\ref{thm:upper}-(b).
Hence $a^M$ is maximal.
\qed

\subsubsection{(b): $m(H)\neq \emptyset,\{\overline{c}\}$ and $\frac{1}{\mu}<S(c_{loc})\leq \frac{1}{\overline{c}}$}
\noindent\textbf{Equilibrium Check}:
Observe that $a^M=c_F^{-1}\left(\frac{1}{S(c_{loc})} \right)\leq \overline{a}=c_F^{-1}(\overline{c})$ since $c_F(\cdot)$ is decreasing.
Hence, $S(c)\geq S(c_{loc})=\frac{1}{c_F(a^M)}$ for all $c\in[0,\overline{c}]$ by Lemma \ref{lemma:c_{loc}_1}-(d).
By Proposition \ref{proposition:full_iff_cost}, $U_{a^M}$ is an equilibrium.

\noindent\textbf{Maximallity}:
Assume first that $a^M<\overline{a}$, and fix some $a\in(a^M,\overline{a}]$.
Then, $c_F(a)<c_F(a^M)=\frac{1}{S(c_{loc})}$, which implies $S(c_{loc})<\frac{1}{c_F(a)}$. By the continuity of $S$, there exists some $c$ such that $\frac{H(c)}{c}<\frac{1}{c_F(a)}$. By Proposition~\ref{proposition:full_iff_cost}, such $U_a$ cannot be an equilibrium. Furthermore, by Theorem~\ref{thm:upper}-(b), no $a>\overline{a}$ can be an equilibrium.

Now suppose $a^M=\overline{a}$, which is equivalent to $S(c_{loc})=\frac{1}{\overline{c}}$. We consider 2 cases depending on whether $c_{loc}=\overline{c}$ or $c_{loc}<\overline{c}$.
If $c_{loc}=\overline{c}$, let $A:=\argmin_{c\in m(H)}S(c)$. Since $m(H)\neq \{\overline{c}\}$, we have $|A|>1$. Consider two subcases depending on whether $\overline{c}$ is an isolated point of $A$ or not. Consider the first case where $\overline{c}$ is isolated point in $A$. For small enough $\e>0$, $A\cap(\overline{c}-\e, \overline{c})=\emptyset$. Then there exists some $c_1<\overline{c}-\e$ such that $c_1\in A$. Moreover, $S(c)>S(\overline{c})$ for $c\in (\overline{c}-\e, \overline{c})$, implying $S(c_1)<S(c)$ on this interval. Hence, for such $c$, any $U_a$ with $a=c_F^{-1}(c)>a^M$ cannot be an equilibrium by Lemma~\ref{lemma:phi_a_geq_D}-(b). Thus, $U_{a^M}$ is maximal. The case of $c_{loc}<\overline{c}$ follows exactly the same reasoning.

Now consider the second case where $\overline{c}$ is not an isolated point of $A$.
Then there exists some $\delta>0$ such that $(\overline{c}-\delta, \overline{c})\subseteq A$, implying that $S(c)$ is constant over $(\overline{c}-\delta, \overline{c})$. Thus, $S(c)=h(c)$ over $(\overline{c}-\delta, \overline{c})$. By Lemma~\ref{lemma:phi_a_convex_kink}-(b), for all such $c$, $a=c_F^{-1}(c)>a^M$ cannot be an equilibrium. Again, $U_{a^M}$ is maximal.
\qed

\subsubsection{(c): $m(H)\neq \emptyset$ and $S(c_{loc})> \frac{1}{\overline{c}}$ and (d): $m(H)=\emptyset$ or $m(H)=\{\overline{c}\}$.}
\noindent\textbf{Equilibrium Check}:
Let $\hat{c}=\max(c_{cav}, c_{sol})$. Note $\hat{c}\leq \overline{c}$ implies $a^M\geq \overline{a}$. 

We first show (c).
Assume first that $m(H)\neq \emptyset$ and $S(c_{loc})>\frac{1}{\overline{c}}$.
First consider the case $\hat{c}=c_{sol}\geq c_{cav}$. We have $S(c)\geq S(c_{sol})>h(c_{sol})$ for all $c\in [0,c_{sol}]$, where the first and second inequalities follow from Lemma~\ref{lemma:c_{loc}_1}-(d) and (c), respectively. Thus, $S(c)\geq S(\hat{c})>h(\hat{c})$ for all $c\in[0,\hat{c}]$, and $\hat{c}\geq c_{cav}$. It follows from Proposition~\ref{proposition:full_iff_cost} that $U_{a^M}$ is an equilibrium. 
Now consider the case $\hat{c}=c_{cav}>c_{sol}$.
Since $c_{cav}>c_{sol}$, Lemma~\ref{lemma:c_{loc}_1}-(c) implies $S'(c_{cav})<0$, so $S(c)\geq S(c_{cav})>h(c_{cav})$ for all $c\in[0,c_{cav}]$. Again, $U_{a^M}$ is an equilibrium by Proposition~\ref{proposition:full_iff_cost}-(b).

Now we show (d). Assume $m(H)=\emptyset$. Then, $S'(c_{cav})<0$, since otherwise $c_{cav}\in m(H)$. The function $S(\cdot)$ is monotonically decreasing on $[0,\overline{c}]$, so in particular $S(c)\geq S(c_{cav})>h(c_{cav})$ for all $c\in[0,c_{cav}]$. Hence, $U_{a^M}$ is an equilibrium by Proposition \ref{proposition:full_iff_cost}.
If $m(H)=\{c\}$, then again $S'(c_{cav})<0$, as otherwise $c_{cav}\in m(H)$. The same argument shows $U_{a^M}$ is an equilibrium.

\noindent\textbf{Maximality}: 
Fix $a>a^M$.
If $\hat{c}=c_{cav}$, then $c_F(a)<c_{cav}$, violating the requirement $c_F(a)\geq c_{cav}$ from Proposition~\ref{proposition:full_iff_cost}-(b). Hence, $U_a$ is not an equilibrium. If $\hat{c}=c_{sol}$, then $c_F(a)<c_{sol}$, but $S(c_F(a))>S(c_{sol})=S(c_{loc})$, which violates the equilibrium condition from Proposition~\ref{proposition:full_iff_cost} that $S(c)\geq S(c_F(a))$ for all $c\in[0,c_F(a)]$. Again, $U_a$ is not an equilibrium.
\qed

\subsection{Upper censorship equilibria as the unique class of equilibrium (Theorem \ref{thm:upper}-(c))}
\label{appendix:b_4}

In this section, we prove the following Proposition, which encompasses Theorem \ref{thm:upper}-(c) as a corollary. The result shows that only upper censorship distributions $U_a$ can be limit equilibrium when competition becomes sufficiently intense.

\begin{proposition}
If $G\notin \{U_a:a\leq a^M\}$, then there exist some $N_G>0$ such that $G\notin E_n$ for all $n>N_G$. Moreover, if there exists some $N>0$ such that $G$ is an equilibrium for all $n>N$, then $G\in \{U_a:a\in[0,1]\}$.
    \label{proposition:uniqueness}
\end{proposition}

\paragraph{Proof of Theorem \ref{thm:upper}-(c)}
By Theorem \ref{thm:appendix_maximal}, there exists some $N>0$ such that if $n>N$, $\{U_a:a\leq a^M\}\subseteq E_n$. This implies $\{U_a:a\leq a^M\}\subseteq \liminf_{n\rightarrow\infty}E_n, \limsup_{n\righatrrow\infty}E_n$. Assume $G\notin \{U_a:a\in[0,1]\}$. Proposition~\ref{proposition:uniqueness} implies $G\notin \limsup_{n\rightarrow \infty}E_n$ as well as $G\notin \liminf_{n\rightarrow \infty}E_n$. Furthermore, since $U_{a^M}$ is maximal, $\{U_{a}:a>a^M\}\cap E_n=\emptyset$ for all $n$. Hence, $\liminf_{n\righatrrow\infty}E_n=\limsup_{n\righatrrow\infty}E_n=\{U_a:a\leq a^M\}$.
\qed

\paragraph{Proof of Proposition \ref{proposition:uniqueness}}
We prove necessary lemmas. Lemma \ref{lemma:continuity_atom} establishes that equilibrium $G$ must have an atom at the maximum of its support, and continuous elsewhere. Lemma \ref{lemma:full_info_bottom} and Lemma \ref{lemma:partial_stop} addresses characterizes equilibrium properties over $[0,\underline{r}]$ and $[\underline{r},1]$, respectively. Combining these lemmas proves Proposition \ref{proposition:uniqueness}, which is presented at the end of this section.

\begin{lemma}[Necessity of continuity and atom]
    Assume $G$ is an equilibrium and $k_G:=\max(\supp(G))$.
    \begin{enumerate}[label=(\alph*)]
        \item $G$ is continuous at $[0,k_G)$.
        \item $G$ has an atom at $k_G$.
    \end{enumerate}
    \label{lemma:continuity_atom}
\end{lemma}
\begin{proof}
    The proof of (a) is deferred to Online Appendix~\ref{online_appendix:proofs}; we sketch the proof here. Suppose to the contrary that $G$ has an atom at some $x_0\in (0,k_G)$ with size $\alpha=G(\{x_0\})$. For small enough $\epsilon>0$, consider a deviation $G_{\epsilon}$ that splits the mass $\alpha$ at $x_0$ into mass $\frac{1}{n+1}\alpha$ at $x_0-n\e$ and mass $\frac{n}{n+1}\alpha$ at $x_0+\e$. We show this constitutes a profitable deviation:
    $$
    \lim_{\epsilon\rightarrow 0}\left(
        \int_0^1 D(x;G)dG_{\epsilon}(x)-\int_0^1 D(x;G)dG(x)
    \right)\geq \frac{\alpha H(c_G(x_0))}{n(n+1)}(G(x_0)^{n-1}-G(x_0-)^{n-1})>0.
    $$

    To prove (b), assume, to the contrary, that $G$ is an equilibrium but is continuous in $k_G$. We show this implies $D(x;G)$ to be strictly concave locally around $k_G$, which contradicts the convexity of $D(x;G)$ (Lemma~\ref{appendix_lemma:dworczak_martini}).
    Let $\e>0$ be small enough such that $(k_G-\e,k_G)\subset \supp(G)$. Since $G$ is an equilibrium, Lemma~\ref{appendix_lemma:dworczak_martini} implies that $\phi(x)=D(x;G)$ is convex over $(k_G-\e, k_G)$. Therefore, $D(x;G)$ is almost everywhere differentiable over this interval, with the derivative
\begin{align*}
    D'(x;G)=(1-G(x))h(c_G(x))J_G(x)+(n-1)G(x)^{n-2}g(x)H(c_G(x))\geq 0.     
\end{align*}
The first term is strictly positive for $x<k_G$ and the second term is weakly positive. This implies $D'(x;G)>0$ for all $x<K_G$. 
Moreover, $\lim_{x\rightarrow k_G-}G(x)=1$ implies $\lim_{x\rightarrow k_G-}D'(x;G)=0$. 

Now, fix some $y\in (k_G-\e, k)$, and choose $\delta \in (0, D'(y;G))$.
By the assumption of continuity, $\lim_{x\rightarrow k-}D'(x;G)=0$. Therefore,
there exists some $\epsilon_{\delta}$ such that $D'(x;G)<\delta$ for all $x\in(k-\epsilon_{\delta}, k)$. But this implies that $D'(y;G)>D'(x;G)$ for all $\max(y, k-\e_{\delta})<x<k$, which contradicts the weak convexity of $D(x;G)$ over $(k_G-\e, k_G)$. Therefore, $G$ must have an atom at $k_G$. 
\end{proof}
 
\begin{lemma}[Equilibrium properties of $G$ over {$[0,\underline{r}]$}]
    There exists some $N>0$ such that for any $n>N$, any equilibrium $G$ must either 
    \begin{enumerate}[label=(\alph*)]
        \item $G=0$ over $[0,\underline{r}]$, or
        \item $G=F$ over $[0,a]$ and constant over $[a,\underline{r}]$ for some $a\leq \underline{r}$.
    \end{enumerate}
    Furthermore, if $G$ does not satisfy either (a) or (b), then there exists some $N_G$ such that $G$ is not an equilibrium if $n>N_G$.
    \label{lemma:full_info_bottom}
\end{lemma}
\begin{proof}
    The detailed proof of the steps are deferred to Online Appendix \ref{online_appendix:proofs}. Here, we outline the main steps of the proof. Assume $G$ is an equilibrium. We first show that $m:=\min(\supp(G))$ must be either $0$ or $\underline{r}$. If $m=\underline{r}$, then the statement (a) of the lemma holds. So, we focus on the case $m=0$. We then establish that $\supp(G)\cap[0,\underline{r}]$ is connected, meaning that for some $a\leq \underline{r}$, $G$ is increasing over $[0,a]$ and flat over $[a,\underline{r}]$.
    Next, we prove that $G=F$ over $[0,\e]$ for some $\e>0$. If this were not true, then the linearity of $G^{n-1}$ over some interval of $0$ (Lemma \ref{appendix_lemma:dworczak_martini}) implies $\lim_{x\rightarrow 0+}g(x)=\infty$, contradicting the assumption that $G$ is an MPC of $F$. 
    We then show the full information region $\{x|G(x)=F(x)\}\cap[0,\underline{r}]$ must be connected. To see why, suppose $G=F$ over some $[0,x_1]$ and again over $[x_2,x_3]$ with $G^{n-1}$ being linear over $[x_1, x_2]$. If $N$ is sufficiently large, $F^{n-1}$ is convex, implying $G^{n-1}>F^{n-1}$ over $[x_1,x_2]$. Since $G=F$ over $[0,x_1]$, this implies $\int_0^x G(t)dt>\int_0^xF(t)dt$ over $x\in[x_1,x_2]$, contradicting that $G$ is an MPC of $F$. 
    Finally, we rule out the possibility that $G=F$ over some $[0,a]$ and $G^{n-1}$ is linear over $[a,\underline{r}]$. Let $[a,b]$ be the maximal interval where $G$ is a strict MPC of $F$. If $b=\underline{r}$, then by the same reasoning above, $G$ cannot be an MPC of $F$. If $b>\underline{r}$, then $D(x;G)$ always exhibits a strictly increasing kink at $x=\underline{r}$, contradicting the linearity of $D(x;G)$ over $[a,b]$.
\end{proof}

\begin{lemma}[Equilibrium properties of $G$ over {$[\underline{r},1]$}]
    If there exists some $N$ such that $G$ is an equilibrium for all $n>N$, then either
    \begin{enumerate}[label=(\alph*)]
        \item $\supp(G)\cap(\underline{r},k_G)=\emptyset$, or
        \item $\supp(G)\cap(\underline{r},k_G)=(\underline{r},a]\cup\{k_G\}$ where 
        $G=F$ over $[\underline{r},a]$.
    \end{enumerate}
    Furthermore, if $G$ does not satisfy either (a) or (b), then there exists some $N_G$ such that $G$ is not an equilibrium if $n>N_G$.
    \label{lemma:partial_stop}
\end{lemma}
\begin{proof}
    Let $A:=\supp(G)\cap(\underline{r},k_G)$. We focus on the case where $A\neq \emptyset$. By Lemma \ref{appendix_lemma:dworczak_martini}, $D(x;G)$ must be a convex function over $\supp(G)$ and a necessary condition for the convexity (Lemma \ref{lemma:partial_stop_convex}) is that $h'(c_G(x))\leq0$ for all $x\in A$. Since $D(x;G)$ is strictly convex over $A$, Lemma \ref{appendix_lemma:dworczak_martini} implies that $A\subseteq\{x|G(x)=F(x)\}$. 
    Since $G$ has an atom at $k_G$, $G$ is a strict MPC of $F$ over $(k_G-\e, k_G+\e)$. Therefore, $\max A\neq k_G$. This implies that $G$ is constant over $[\max A, k_G)$ and jumps to $1$ at $k_G$. Let $\beta:=\inf A$. Suppose that $\beta>\underline{r}$. From Lemma \ref{lemma:full_info_bottom}, $G=F$ over some $[0,a]$ and is constant over $[a,\underline{r}]$. Since $\beta>\underline{r}$, $G$ is also constant over $[\underline{r},\beta]$ and hence over $[a,\beta]$. This is a contradiction because $F(a)=G(a)=G(\beta)=F(\beta)>F(a)$. Hence, we must have $\beta=\underline{r}$, implying that $G=F$ over $[\underline{r},\max A)$.
    
    The last part follows from the Lemma \ref{lemma:partial_stop_convex}: if $G$ is a strict MPC of $F$ over some $(x-\e, x+\e)\subseteq \supp(G)\cap (\underline{r},k_G)$, either $D(\cdot;G)$ is strictly concave or strictly convex over $(x-\e, x+\e)$ for large enough $n>N_G$, contradicting Lemma~\ref{appendix_lemma:dworczak_martini}.
\end{proof}

Proposition \ref{proposition:uniqueness} directly follows from Lemmas \ref{lemma:continuity_atom}, \ref{lemma:full_info_bottom} and \ref{lemma:partial_stop}. 
Assume there exists some $N>0$ such that $G$ is an equilibrium for all $n>N$.
By Lemma \ref{lemma:continuity_atom}, $G$ must have an atom at $k_G$ and continuous over $[0,k_G]$. Lemma \ref{lemma:full_info_bottom} shows that either $G=0$ over $[0,\underline{r}]$ or $G=F$ over some $[0,a_1]$ and remains constant over $[a_1,\underline{r}]$. We consider three cases:

First, if $G=0$ over $[0,\underline{r}]$, then by Lemma \ref{lemma:partial_stop}, $\supp(G)\cap(\underline{r},k_G)=\emptyset$. Therefore, $G$ must be a point mass $G=\delta_{k_G}=\delta_{\mu}$.
Second, if $G=F$ over $[0,a_1]$ with $a_1<\underline{r}$, then Lemma \ref{lemma:partial_stop} again implies that $\supp(G)\cap(\underline{r},k_G)=\emptyset$. Therefore, $G=U_{a_1}$.
Finally, assume $G=F$ over $[0,\underline{r}]$. If Lemma \ref{lemma:partial_stop} applies, we have $G=U_{a_1}$. Otherwise, if $\supp(G)\cap(\underline{r},k_G)\neq \emptyset$, then $G=F$ over some interval $[\underline{r},a_1]$. Therefore, $G=U_{a_2}$. For all three cases, $G\in \{U_a:a\in[0,1]\}$.
This completes the proof of Proposition \ref{proposition:uniqueness}.

Assume $G\notin \{U_a:a\in[0,1]\}$. If the conditions of Lemma~\ref{lemma:continuity_atom}-(a,b) fail, then $G$ is not an equilibrium for any $n$. If they hold, then $G$ must violate at least one of the necessary conditions stated in Lemma~\ref{lemma:full_info_bottom}-(a,b) or Lemma~\ref{lemma:partial_stop}-(a,b). Hence, there exists some $N>0$ such that $G$ is not an equilibrium for any $n>N$.
\qed

\clearpage
\setcounter{page}{1}
\begin{center}
    \Huge{Supplemental Materials for Online Publication}
\end{center}
\section{Small $n$ equilibria}
\label{section:online_appendix_c}

Throughout the paper, we have focused on large markets (large but finite $n$) and its limit ($n\rightarrow \infty$).
In real-world settings, the large market assumption---together with the random search process---is justified in context such as online marketplaces, digital platforms, and large-scale service industries. In these environments, consumers often face a vast number of options and engage in search without a predetermined order, making random search a reasonable approximation.

We adopt the large market assumption for two key reasons. First, from a theoretical standpoint, our objective is to explore the extent competition can promote (or deter) information disclosure under search cost heterogeneity. Since established results in the literature rely on sufficiently competitive settings, adopting this assumption directly compare our findings with those in the existing literature.

Second, the large market assumption helps mitigate the technical complexities associated with small markets.
As noted in \textcite{dworczak2019simple}, the objective function $D(x;G)$ must be convex in the region of full-disclosure $(G=F)$
and linear in every subinterval of $\supp(G)$
where $G$ is a strict MPC of $F$.\footnote{
    We say $G\in MPC(F)$ is a strict MPC of $F$ over $[a,b]$ if 
    $G(a)=F(a)$, $G(b)=F(b)$, $\int_a^bxdF=\int_a^bxdG$ and 
    there doesn't exist any $\e>0$ where $G=F$ over $(x-\e, x+\e)$ for any $x\in(a,b)$.
}
The inherent non-linearity and complexity
of the objective function (see Equation \eqref{D(x_i;G)}),
makes the small market characterization exceedingly difficult. \footnote{For instance, while \textcite{au2023attraction} mainly focus on discrete priors, they also resolve the continuous prior case only under the assumption of sufficiently large markets.}

Nevertheless, if we relax some generality of our framework by imposing additional structure on either the cost distribution $H$ (Proposition \ref{proposition:small_market}) or the prior distribution $F$  (Proposition \ref{proposition:small_market_binary}), we can still characterize equilibria in small markets.

\begin{proposition}
    Assume $G$ is an equilibrium. For any $n>0$,
    \begin{enumerate}[label=(\alph*)]
        \item (Atom at the top)
        $G$ must have an atom at $\max(\supp(G))=\underline{r}$.
        \item (Disclosure at the bottom)
        If $G(\underline{r})>0$, there exist some partition 
        $0=x_0<\dots<x_{k-1}\leq x_k=\underline{r}$ that for each $i=0,\dots, k-2$,
        on the interval
        $[x_i,x_{i+1}]$, either $G=F$ or $G^{n-1}$ is linear. 
        These intervals alternate, and $G$ is constant for the final interval $[x_{k-1},\underline{r}]$.
    \end{enumerate}
    Furthermore, if $H$ is strictly convex, strictly concave or linear,
    \begin{enumerate}[label=(\alph*),resume]
        \item (No partial-purchase signals) 
        if $x_{k-1}<\underline{r}$, then
        $\supp(G)\cap(\underline{r},\overline{r})=\emptyset$
    \end{enumerate}
    \label{proposition:small_market}
\end{proposition}

The general intuition behind the upper-censorship distribution remains unchanged. 
Firms place an atom at $\overline{r}$ to retain all consumer types, and they compete below $\underline{r}$ to maximize their chances of being revisited. The only difference lies in how firms compete in low signals below $\underline{r}$. For small $n$, the competition is not intense enough to promote full disclosure. The alternating equilibrium structure below $\underline{r}$ is exactly that of \textcite{hwang2023}. Locally, firms either fully disclose information $(G=F)$, or garble information in such a way that $G^{n-1}$ is linear, ensuring indifference. 

If instead, we relax the generality of the prior distribution by considering binary match values, we obtain a similar pooling structure:

\begin{proposition}
    Assume the match values are binary $\{0,1\}$
    , i.e. $F\in \Delta(\{0,1\})$ and that $H$ is strictly convex, strictly concave, or linear.
    For any $n$, if $G$ is an informative equilibrium, 
    there exist some $a\leq \underline{r}$ such that 
    \begin{enumerate}[label=(\alph*)]
        \item $G^{n-1}$ is linear over $[0,a]$
        \item The complementary mass is placed at $\max(\supp(G))$.
    \end{enumerate}
    \label{proposition:small_market_binary}
\end{proposition}


\section{Other Omitted Proofs}
\label{online_appendix:proofs}

\subsection{Omitted Proofs of Appendix \ref{appendix:a} without Assumption \ref{assumption:diag}}
\label{appendix:general_comp_stat}
In Appendix \ref{appendix:a}, all distributions 
$H_k$ were assumed to satisfy Assumption \ref{assumption:diag}, corresponding to cases (a) and (b) of Theorem \ref{thm:appendix_maximal}. In this section, we address the remaining cases (c) and (d), where Assumption \ref{assumption:diag} fail.

\paragraph{Proof of Proposition \ref{proposition:comp_stat_shift}}
Since $H_2(c)$ is an $\alpha$-scale shift of $H_1$, all relevant quantities scale accordingly: $c_{2,cav}=\alpha c_{1,cav}$, $c_{2,loc}=\alpha c_{1,loc}$, as well as $c_{2,sol}=\alpha c_{1,sol}$ if $c_{1,sol}$ exists. Note $H_1$ and $H_2$ both must fall into the same case of Theorem~\ref{thm:appendix_maximal}.
Consider case (c). By definition, $m(H_1)=m(H_2)$ and 
$S_1(c_{1,loc})>\frac{1}{\overline{c}}$
implies $S_2(c_{2,loc})=\frac{1}{\alpha}S_1(c_{1,loc})>\frac{1}{\alpha \overline{c}}$.
Moreover, since $\max(c_{2,cav},c_{2,cav})=\alpha\max(c_{1,cav},c_{1,cav})$, we conclude:
$$
a^M_2=c_F^{-1}(\max(c_{2,cav},c_{2,sol}))<c_F^{-1}(\max(c_{1,cav},c_{1,sol}))=a^M_1.
$$
Case (d) also follows from this inequality.

\paragraph{Proof of Corollary \ref{corollary:convergence}}

Assume all $H_k$ satisfies Assumption (d).
As $\supp(H_K)\searrow \{0\}$, $c_{k,cav}\searrow 0$.
Hence, $a^M_k=c_F^{-1}(c_{k,cav})\rightarrow c_F^{-1}(0)=1$.
Now assume all $H_k$ satisfies Assumption (c).
Similarly, $c_{k,cav}$, $c_{k,loc}$ and $c_{k,sol}$ all vanishes to zero.
Hence, $\max(c_{k,cav},c_{k,sol})\searrow 0$, implying $a^M_k\rightarrow 1$.

Now consider an arbitrary sequence of $\{H_k\}$.
For any subsequence of $\{H_{k_i}\}_i$, 
we can find sub-subsequence $\{H_{k_{i_j}}\}_j$
where all the elements fall into the same case in Theorem \ref{thm:appendix_maximal}.
We already proved that $a^M_{k_{i_j}}\rightarrow 1$ if all elements fall into the same case.
Hence, $a^M_k\rightarrow 1$.
\qed

\paragraph{Proof of Proposition \ref{proposition:even}}
Note that $S_{\lambda}(c)=\lambda\frac{1}{\overline{c}}+(1-\lambda)S_0(c)$, so its derivatives satisfy $S_{\lambda}'=(1-\lambda)S_0'$ and $S_{\lambda}''=(1-\lambda)S_0''$. It follows that $S(c)\geq S(c_0)$ if and only if $S_{\lambda}(c)\geq S_{\lambda}(c_0)$, implying that the sets $m(H_{\lambda})=m(H_0)$ coincide for all $\lambda\in(0,1)$, and in particular, $c_{0,cav}=c_{\lambda,cav}$ for all $\lambda\in(0,1)$. Additionally, if $m(H_0)\neq \emptyset$, then $c_{0,loc}=c_{\lambda, loc}$ for all $\lambda\in(0,1)$. Hence, we suppress the subscript $\lambda$ and write simply $c_{cav}, c_{loc}$ throughout.

\textit{\textbf{Case (c)}}: Assume $m(H_0)\neq \emptyset$ and $S_0(c_{loc})>\frac{1}{\overline{c}}$. Fix $\lambda\in(0,1)$. Then, $S_{\lambda}(c_{loc})=\lambda\frac{1}{\overline{c}}+(1-\lambda)S_0(c_{loc})>\frac{1}{\overline{c}}$. By Lemma \ref{lemma:c_{loc}_1}-(b), the equation $S_{0}(c)=h_0(c_{loc})$ admits a unique solution $c_{0,sol}\in(c_{loc},\overline{c})$ and $S_{\lambda}(c)=h_{\lambda}(c_{loc})$ admits a unique solution in $c_{\lambda, sol}\in(c_{loc},\overline{c})$. Observe $$S_{\lambda}(c_{\lambda, sol})=h_{\lambda}(c_{loc})=\lambda \frac{1}{\overline{c}}+(1-\lambda)h_{0}(c_{loc})=\lambda \frac{1}{\overline{c}}+(1-\lambda)S_0(c_{0,sol})=S_{\lambda}(c_{0,sol}).$$ By the uniqueness of solutions, we must have $c_{\lambda,sol}=c_{0,sol}$. Therefore, $$a^M_{\lambda}=c_F^{-1}(\max(c_{cav},c_{\lambda,sol}))=c_F^{-1}(\max(c_{cav},c_{0,sol}))=a^M_0$$

\textit{\textbf{Case (d)}}: Follows directly since $c_{cav}$ remains constant for all $\lambda\in[0,1)$, implying $a^M_{\lambda}=c_F^{-1}(c_{cav})=a^M_0$.
\qed

\paragraph{Proof of Proposition \ref{proposition:mps}}

Assume $H_1$ and $H_2$ both admits strictly quasi-convex densities.
Note $H_1$ single crosses $H_2$ from below.
Note $c_{k,cav}=\overline{c}$ for both $k=1,2$.
Either only $H_1$ fails Assumption \ref{assumption:diag} or both does.
First consider the case only $H_1$ fails Assumption \ref{assumption:diag}.
Since $h_2(c^M_1)\leq \frac{1}{\overline{c}}$, we have 
$a^M_2=c_F^{-1}(\frac{1}{h(c_{2,loc})})\leq c_F^{-1}(\overline{c})=\overline{a}$ 
by Theorem \ref{thm:appendix_maximal}.
$a^M_1=c_F^{-1}(\overline{c})=\overline{a}$.
Hence, $a^M_1\geq a^M_2$.
If both $H_1$ and $H_2$ fails Assumption \ref{assumption:diag},
$a^M_k=c_F^{-1}(c_{k,cav})=\overline{a}$ by Theorem \ref{thm:appendix_maximal}.
Hence, $a^M_1=a^M_2$,

For (b), assume $H_1$ and $H_2$ both admit strictly quasi-concave densities 
with interior peaks. Denote the peaks $p_k$ as in Appendix \ref{appendix:a}.
Since $H_1$ single crosses $H_2$ from below,
either only $H_1$ satisfies Assumption \ref{assumption:diag} or both doesn't.
Note that $c_{k,loc}=0$ for both $k=1,2$
and $c_{k,cav}=p_k$, where $p_k$ is the peak of $h_k$, for both $k=1,2$.

Assume only $H_1$ satisfies Assumption \ref{assumption:diag}.
By Theorem \ref{thm:appendix_maximal}, 
$a^M_1=c_F^{-1}(\frac{1}{h(c_{1,loc})})\leq c_F^{-1}(\overline{c})=\overline{a}$.
On the other hand, $a^M_2=c_F^{-1}(c_{k,sol})\geq c_F^{-1}(\overline{c})=\overline{a}$.
Hence, $a^M_1\leq a^M_2$.

If both $H_1$ and $H_2$ fails Assumption \ref{assumption:diag},
$m(H_k)\neq \emptyset$ as $c_{k,loc}=0$ for $k=1,2$.
Since the solution $c_{k,sol}$ must lie in the concave region of $H$,
$a^M_k=c_F^{-1}(\max(c_{k,sol},c_{k,cav}))=c_F^{-1}(c_{k,sol})$.
It is enough we prove $c_{1,sol}\geq c_{2,sol}$.
Note both solutions must lie in the region where $H_1(c)\geq H_2(c)$
since $H_1$ single crosses $H_2$ from below.
Let $\hat{c}$ be the solution to $S_1(\hat{c})=h_2(0)$.
The solution $\hat{c}$ exists from Lemma \ref{lemma:c_{loc}_1}-(b) since $h_2(0)>\frac{1}{\overline{c}}$.
By replicating the proof of Lemma \ref{lemma:c_{loc}_1}-(c),
we can show $S_1'(\hat{c})<0$, with $S_1'(c)<0$ for all $c>\hat{c}$.

We first show $c_{2,sol}\leq \hat{c}$, and $\hat{c}\leq c_{1,sol}$. Observe 
$$
S_1(c_{2,sol})>S_2(c_{2,sol})=h_2(0)
$$
where the strict inequality follows from single-crossing.
Since $S_1(\overline{c})=\frac{1}{\overline{c}}\leq h_2(0)$,
a solution to $S_1(c)=h_2(0)$ must exist between $(c_{2,sol},\overline{c})$. Since this solution is unique, it must be $c_{2,sol}\leq \hat{c}$.
Now, $\hat{c}\leq c_{1,sol}$ follows from 
$$S_1(\hat{c})=h_2(0)\geq h_1(0)=S_1(c_{1,sol}).$$
\qed

\subsection{Omitted Proofs of Appendix \ref{appendix:b}}
\paragraph{Proof of Lemma \ref{lemma:partial_stop_convex}}
Since $D(x;G)$ is twice differentiable, we have:
\begin{align}
    D''(x;G)&=(n-1)G^{n-3}\Big[
        (n-1)g^2H(c_G)+g'GH-2(1-G)Ggh(c_G)\Big]-(1-G)^2J_Gh'(c_G)
        \label{eq:D''}
        \\
        &=(n-1)G^{n-3}\Big[
        (n-1)g^2H(c_G)+g'GH-2(1-G)Ggh(c_G)-\frac{(1-G)^2J_G}{(n-1)G^{n-3}}h'(c_G)\Big].
        \label{eq:D''_2}
\end{align}
Let $M_h$ and $m_h$ be the global minimum and maximum of $h$ over $[0,\overline{c}]$.
Likewise, define $M_f, M_{|f'|}$ and $m_f, m_{|f'|}$ 
analogously for $f$ and $|f'|$.

First consider the case $h'(c_G(x))\leq 0$.
Since $x<k_G$, there exist some $\e>0$ such that 
$H(c_G(x))>\epsilon$.
From \eqref{eq:D''},
$$
D''(x;G) \geq (n-1)G^{n-3}\Big[
    (n-1)g^2\epsilon-|g'|
    -2gM_h
\Big]
$$
The term inside the bracket is strictly positive for large enough $n$.
This proves (a).

Now assume $h'(c_F(x))>0$.
We show the term inside the square bracket in \eqref{eq:D''_2} is negative for large enough $n$.
The term inside the square bracket in \eqref{eq:D''_2} is strictly less than
$$
a_n:=(n-2)g^2+g'-\frac{(1-G(x))^2J_G(x)}{(n-1)G(x)^{n-3}}h'(c_G(x)).
$$
As $n$ increases,
$\frac{(1-G)^2J_G}{(n-1)G^{n-3}}$ diverges exponentially faster
than $(n-2)g^2$, implying $\lim_{n\rightarrow \infty}a_n=-\infty$.
Therefore, $D''(x;G)\leq (n-1)G^{n-3}a_n<0$ for large enough $n$.

Now consider $D(x;F)$ and that $h'(c)\leq 0$ for all $c\in[c_F(x),\overline{c}]$.
For any $t\in[\underline{r},x]$,
\begin{align*}
    D''(t;F)&\geq (n-1)F(\underline{r})^{n-3}[
        (n-1)m_f^2H(c_G(x))-M_{|f'|}-2M_fM_h
    ]
\end{align*}
Hence there exists some $N>0$ such that if $n>N$,
then $D''(t;F)>0$ for all $t\in[\underline{r},x]$.
\qed

\paragraph{Proof of Lemma \ref{lemma:phi_a_convex_kink}}

To prove (c), we distinguish cases depending on whether $a\leq \overline{a}$ or $a>\overline{a}$. First, consider $a\leq \overline{a}$. For $t\leq \overline{a}$, define $R(t):=\frac{J_F(t)}{k_t-t}$ and $L(t):=(n-1)F(t)^{n-2}f(t)$ the right and left derivative of $\phi_t(x)$ at $x=t$, respectively, and $\Delta(t):=R(t)-L(t)$ the kink size. By assumption, $\Delta(a)\geq 0$ and $L'(t)\geq 0$ for $t\in[0,a]$. We wish to show $\Delta(b)\geq 0$ for all $b\leq a$. Note that $R(0)=\frac{1}{n\mu}>0=L(0)$. Direct differentiation yields:
$$\frac{d}{dt}R(t)=\frac{1}{k_t-t}\left(\frac{J_F(t)}{k_t-t}-(n-1)F(t)^{n-2}f(t)\right)=\frac{1}{k_t-t}(R(t)-L(t))$$
so that $\Delta'(t)=C(t)\Delta(t)-L'(t)$ where $C(t):=\frac{1}{k_t-t}$. This is a first-order ordinary differential equation. Let $E(t):=\exp\left(-\int_0^t C(s)ds\right)>0$, which allows us to express the solution:
$$
E(t)\Delta(t)-E(0)\Delta(0)=-\int_0^t L'(s)E(s)ds \quad \text{or equivalently} \quad
\Delta(t)=\frac{1}{E(t)}\left(\Delta(0)-\int_0^tL'(s)E(s)ds\right)
$$
Since $L'(t)\geq 0$ and $E(t)>0$, the integrand is non-negative and increasing in $t$. Hence, the expression in brackets is decreasing. Hence $\Delta(a)\geq 0$ implies $\Delta(t)\geq 0$ for all $t\leq a$.

If $a>\overline{a}$, define similarly for $t\in[\overline{a},a]$,
$\hat{R}(t):=\frac{J_F(t)H(c_F(t))}{k_t-t}$,$\hat{L}(t)=(n-1)F(t)^{n-2}f(t)H(c_F(t))+(1-F(t))h(c_F(t))J_F(t)$ and $\hat{\Delta}(t):=\hat{R}(t)-\hat{L}(t)$. Calculation again yields $\hat{R}'(t)=C(t)(\hat{R}(t)-\hat{L}(t))$. Defining $\hat{E}(t):=\int_{\overline{a}}^t C(s)ds$, we obtain the solution 
$
\hat{\Delta}(t)=\frac{1}{E(t)}\left(\hat{E}(\overline{a})\hat{\Delta}(\overline{a})-\int_{\overline{a}}^t\hat{L}'(s)E(s)\right).
$
Again, since the term in the bracket is decreasing in $t$, $\hat{\Delta}(a)\geq 0$ implies $\hat{\Delta}(b)\geq 0$ for all $b\in[\overline{a},a]$. The result over $[0,\overline{a}]$ follows from the preceding paragraph. 

For the second part of (c), note $\hat{R}'(t)=C(t)\hat{\Delta}(t)\geq 0$ for $t\leq a$ implies that $\hat{R}(t)$ and $R(t)$ is increasing in $t$. Furthermore, $\frac{H(c_F(a))}{c_F(a)}=\frac{\hat{R}(a)}{(1-F(a))J_F(a)}$. Calculation shows 
$$
\frac{\partial}{\partial b}\frac{H(c_F(b))}{c_F(b)}=\frac{\hat{R}'(a)(1-F(a))J_F(a)+R(b)(n-1)F(b)^{n-2}f(b)(1-F(b))}{(1-F(b))^2J_F(b)^2}\geq 0.
$$
The inequality follows since $\hat{R}'(a)\geq 0$.Hence, $\frac{H(c_F(b))}{c_F(b)}$ is increasing in $b$ for $b\in[\overline{a},a]$.

\paragraph{Proof of Lemma \ref{lemma:c_{loc}_1}}
We first prove (a). Let $A=\{c\;|\;S(c)=S(c_{loc}) \text{ and } c\in(c_{loc},\overline{c}]\;\}$ be the set of solutions to \eqref{eq:cross}. Assume to the contrary $|A|\geq 2$.Let $c_1<c_2$ be distinct elements of $A$. We must have $S(c)\geq S(c_{loc})$ for all $c\in(c_{loc},c_1)$. Otherwise, there exists some local minima $\hat{c}\in (c_{loc},c_1)$ of $S(\cdot)$. However, this implies $\hat{c}\in m(H)$ while $S(\hat{c})<S(c_{loc})$, contradicting minimality of $c_{loc}$. Applying the same logic to the interval $(c_1, c_2)$, we have $S(c)\geq S(c_{loc})$ for all $c\in(c_1, c_2)$. However, since $S(c)\geq S(c_{loc})$ for all $[0,c_1]$ and $c_1$ is a local minimum, this contradicts the maximality of $c_{loc}$. Thus, $A$ must be either empty of singleton.

We now prove (b). Since $c_{loc}$ is a local minimum, there exists some $\e>0$ such that $S(c)>S(c_{loc})>\frac{1}{\overline{c}}$ over $c\in(c_{loc},c_{loc}+\e)$. Since $S(\overline{c})=\frac{1}{\overline{c}}$, a solution to \eqref{eq:cross} exists in $(c_{loc}+\e, \overline{c}]$ by the intermediate value theorem (IVT).
This proves the "if" part.

The proof of the "only if" part is a direct and repeated application of the IVT. 
Suppose the solution $c_{sol}$ exists and $S(c_{loc})\leq \frac{1}{\overline{c}}$. We consider 3 cases, where $S'(c_{sol})$ is strictly greater, strictly smaller, and equal to $0$. If $S'(c_{sol})>0$, then $S(c)<S(c_{sol})=S(c_{loc})$ locally over $c\in(c_{sol}-\e, c_{sol})$, implying the existence of another solution of \eqref{eq:cross} over $c\in(c_{loc},c_{sol})$. Similarly, if $S'(c)<0$, then $S(c)<S(c_{sol})=S(c_{loc})\leq \frac{1}{\overline{c}}=S(\overline{c})$ over $c\in(c_{sol},c_{sol}+\e)$, implying the existence of another solution of \eqref{eq:cross} over $c\in(c_{sol}+\e,\overline{c})$. If $S'(c_{sol})=0$, we consider three subcases based on the sign of $S''(c_{sol})$.
If $S''(c_{sol})>0$ then $c_{sol}$ is a local maximum, implying another solution of \eqref{eq:cross} over $c\in(c_{loc},c_{sol})$. If $S''(c_{sol})>0$, then $c_{loc}<c_{sol}\in m(H)$ with $S(c_{sol})=S(c_{loc})$, contradicting the maximality of $c_{loc}$. If $S''(c_{sol})=0$, define $\hat{c}:=\inf_c \{S''(c)\leq 0, c\geq c_{sol}\}$. If $S''(\hat{c})>0$, 
$\hat{c}$ is again a local minimum of $S(\cdot)$ and $S(\hat{c})=S(c_{loc})$, contradicting the maximality of $c_{loc}$.     If $S''(\hat{c})<0$, then IVT implies another solution exists of \eqref{eq:cross} in $(\hat{c},\overline{c})$.

To prove (c), assume $S'(c_{sol})>0$. Then $S(c)>S(c_{loc})>\frac{1}{\overline{c}}$ locally at $(c_{sol}, c_{sol}+\e)$. By IVT, there exists another solution to \eqref{eq:cross}, contradiction. If $S'(c_{sol})=0$, this implies $c_{sol}$ is a local minimum with value equal to $S(c_{loc})$, and global minimum of $S$ over $[0,c_{sol}]$. This contradicts the maximality of $c_{loc}$. Hence, $S'(c_{sol})< 0$ or equivalently, $S(c_{sol})> h(c_{sol})$. To prove the second part, assume to the contrary that $S'(c)\geq 0$ for some $c\in(c_{sol}, \overline{c}]$. Let $\hat{c}=\inf_{c>c_{sol}}\{c|S'(c)\geq 0\}$. If $\hat{c}>c_{sol}$, continuity of $S'(\cdot)$ implies $S'(\hat{c})=0$. Since $S'(c)<0$ for all $c\in (c_{sol}, \hat{c})$, $\hat{c}$ is a local minimum of $S$ with $S(\hat{c})<S(c_{sol})=S(c_{loc})$, contradicting the minimality of $c_{loc}$. If $\hat{c}=c_{sol}$, this again implies that $c_{sol}$ is a local minimum of $S$ with $S(c_{sol})=S(c_{hat})$, contradicting the maximality of $c_{loc}$. Hence, $S'<0$ for all $c>c_{sol}$.

We now prove (d). $S(c)\geq S(c_{loc})$ for all $c\in[0,c_{loc}]$ holds by definition, and over $c\in[c_{loc},c_{sol}]$ was proved in part (a) of this lemma. This proves the first part. To prove the second part, assume $S(c_{loc})< \frac{1}{\overline{c}}$. If $S(c)<S(c_{loc})$ for some $c>c_{loc}$ IVT implies the existence of solution to \eqref{eq:cross}, contradicting Lemma \ref{lemma:c_{loc}_1}-(a). Hence, $S(c)\geq S(c_{loc})$ for all $c\in[0,\overline{c}]$.

\paragraph{Proof of Lemma \ref{lemma:continuity_atom}-(a)}

Fix $c>0$.
We first derive the payoff $D_c(x;G)$ in Section \ref{section:model} 
including the case of atoms.
If $x<r(c)$, firm $i$ is chosen if and only if $i\in \argmax_{j}x_j$, with fair tie-breaking in case of multiplicity.
The probability $i$ is chosen is $\frac{1}{m}G(x-)^{n-1-k}G(\{x\})^m$
where $m$ is the number of firms in tie with $i$.
Summing over all $\binom{n-1}{m}$ combinations for each $m$,
$$
D_c(x;G)=\sum_{m=0}^{n-1}\binom{n-1}{m}\frac{1}{m}G(x-)^{n-1-m}G(\{x\})^k=
\begin{dcases}
    \frac{G(x)^{n}-G(x-)^n}{nG(\{x\})} & \text{ if } x<r(c) \text{ and } G(\{x\})\neq 0 \\
    G(x)^{n-1} & \text{ if } x<r(c) \text{ and } G(\{x\})=0
\end{dcases} 
$$

Assume $G$ has an atom at some $x_0\in (0,k_G)$.
Denote the size of the atom $\alpha:=G(\{x_0\})$.
We construct a profitable deviation $G_{\epsilon}$
that spreads the atom at $x_0$ to $x_0-n\epsilon$ and $x_0+\epsilon$ in a mean-preserving way.
Formally, $G_{\epsilon}:=G_1+G_2$, where $G_1(A) = (1-\alpha)G(A)$ for all $A\subseteq [0,1]\backslash \{x_0\}$ and $G_2=\frac{1}{n+1}\alpha\delta_{x_0-n\epsilon}+\frac{n}{n+1}\alpha\delta_{x_0+\epsilon}$. We consider 3 subcases depending on the position of the atom $x_0$.

\noindent\textbf{Case 1} ($x_0<\underline{r}$):
Fix some $\epsilon$ small enough such that $x_0+\epsilon<\underline{r}$.
Then, 
\begin{align*}
    \mathbb{E}_{G_{\epsilon}}[D(x;G)]-\mathbb{E}_G[D(x;G)]=
    \alpha H(c_G(x_0))\Bigg(
    \frac{1}{n+1}G(x_0-n\epsilon)^{n-1} + \frac{n}{n+1}G(x_0+\epsilon)^{n-1}
    -\frac{G(x_0)^n-G(x_0-)^n}{n\alpha}
\Bigg)
\end{align*}
Taking $\epsilon\rightarrow 0$, we have
\begin{equation}
    \alpha H(c_G(x_0))\left(
        \frac{1}{n+1}G(x_0-)^{n-1} + \frac{n}{n+1}G(x_0)^{n-1}
        -\frac{G(x_0)^n-G(x_0-)^n}{n\alpha}
    \right)
    \geq \frac{\alpha H(c_G(x_0))}{n(n+1)}(G(x_0)^{n-1}-G(x_0-)^{n-1})>0.
    \label{eq:no_atom}
\end{equation}

\noindent\textbf{Case 2} ($x_0> \underline{r}$):
Fix some $\epsilon$ small enough such that $x_0-n\epsilon>\underline{r}$.
Note $\int_{c_G(x)}^{\overline{c}}\frac{1-G(r(c-))^n}{n(1-G(r(c)-))}dH(c)$ is still continuous in $x$ while $G$ is discontinuous.
Since $H(c_G(x_0-n\epsilon))$ and $H(c_G(x_0+\epsilon))$ both converges to $H(c_G(x_0))$ as $\epsilon\rightarrow 0$.
the deviation payoff $  \lim_{\epsilon\rightarrow 0+}\mathbb{E}_{G_{\epsilon}}[D(x;G)]-\mathbb{E}_G[D(x;G)]$ is exactly \eqref{eq:no_atom}.

\noindent\textbf{Case 3} ($x_0=\underline{r}$):
As $\epsilon\rightarrow 0+$, $c_G(x_0+\epsilon)\rightarrow \overline{c}$ and 
$\lim_{\epsilon\rightarrow +}\int_{c_G(x_0+\epsilon)}^{\overline{c}}\frac{1-G(r(c)-)^n}{n(1-G(r(c)-))}dH(c)=0$. 
Hence, calculation gives exactly \eqref{eq:no_atom}. \qed

\paragraph{Proof of Lemma \ref{lemma:full_info_bottom}}
Assume $G$ is an equilibrium and $\phi$ as in Lemma \ref{appendix_lemma:dworczak_martini}.
We prove in multiple steps. 
\vspace{0.2cm}
\\ \noindent\textbf{Step 1}: $m:=\min(\text{supp}(G))\notin (0,\underline{r})$.
\vspace{0.2cm}\\
Assume to the contrary that $m\in (0,\underline{r})$.
For some $\e>0$, $G$ is a strict MPC of $F$ over $[0,m+\e]$.
$D(x;G)$ is zero over $[0,m]$ and is positive for $(m,m+\e)$.
However, Lemma \ref{appendix_lemma:dworczak_martini} implies that 
$\phi$ is linear over $(0,m+\e)$ while $\phi(0)=\phi(m)=0$ and 
$\phi(m+\e)>0$, which is a contradiction.
Therefore, it must be either $m=0$ or $m\geq\underline{r}$.
The latter proves (a).
To prove (b), we assume $m=0$ for the remainder of proof.
\vspace{0.2cm}\\
\noindent\textbf{Step 2}: $\text{supp}(G)\cap [0,\underline{r}]$ is connected.
\vspace{0.2cm}\\
Assume not. Then, there exist $x_1, x_2 \in \supp(G) \cap [0, \underline{r}]$ with $x_1 < x_2$ such that $(x_1, x_2) \cap \supp(G) = \emptyset$ and $G(x_1) > 0$. For some $\epsilon > 0$, $D(x;G) = \phi(x)$ is strictly increasing over $(x_1 - \epsilon, x_1)$ and $(x_2, x_2 + \epsilon)$, while constant over $(x_1, x_2)$. But no convex function can be strictly increasing, then flat, and then strictly increasing again. Contradiction.
\vspace{0.2cm}\\
\noindent\textbf{Step 3}: There exists some $\epislon>0$ such that $G(x)=F(x)$ for $x\in[0,\epsilon]$.
\vspace{0.2cm}\\
Assume not. 
Then, $G$ is a strict MPC of $F$ over some $[0,\epsilon]$
and $\phi=D=G^{n-1}$ is linear over $[0,\epsilon]$.
There exists some $a_n>0$ such that $G(x)^{n-1}=a_nx$ over $[0,\epsilon]$, or equivalently, $G(x)=(a_nx)^{\frac{1}{n-1}}$ over $[0,\epsilon]$.
Then $G$ admits a density $g(x)=\frac{a_n}{n-1}(a_nx)^{\frac{1}{n-1}-1}$ in $(0,\epsilon)$
which satisfies $\lim_{x\rightarrow 0+}g(x)=\infty$. Since $f$ is bounded, there exists some $\widetilde{\epsilon}<\epsilon$ where
$g(x)>f(x)$ for all $x\in(0,\widetilde{\epsilon})$, implying $G(x)>F(x)$ over the interval. This leads to $\int_0^x G(t)dt>\int_0^x F(t)dt$ over the interval,
contradicting that $G\in MPC(F)$.
\vspace{0.2cm}\\
\noindent\textbf{Step 4}: $\{x\;|\;G(x)=F(x)\}\cap [0,\underline{r}]$ is connected for large enough $n>N$.
\vspace{0.2cm}\\
Assume not. Then we can find some $0\leq x_1<x_2<x_3\leq \underline{r}$ where $G=F$ over $[0,x_1]$, $G$ is a strict MPC of $F$ over $[x_1,x_2]$, and again $G=F$ over $[x_2,x_3]$. Fix $n$ large enough so that $F^{n-1}$ is strictly convex in $[0,1]$. By definition, $G(x_1)^{n-1}=F(x_1)^{n-1}$ and $G(x_2)^{n-1}=F(x_2)^{n-1}$. Moreover, $G(x)^{n-1}$ is linear over $[x_1,x_2]$ by Lemma \ref{appendix_lemma:dworczak_martini}, while $F(x)^{n-1}$ is strictly convex over the same interval. Hence, $G^{n-1}>F^{n-1}$ and $G>F$ over $(x_1,x_2)$. This contradicts that $G$ is an MPC of $F$ because $\int_0^x G(t)dt>\int_0^x F(t)dt$ for $x\in (x_1,x_2)$.
\vspace{0.2cm}\\
\noindent\textbf{Step 5}: $G(x)=F(x)$ for all $x\in\text{supp}(G)\cap [0,\underline{r})$
\vspace{0.2cm}\\
From Step 4, we must have $G(x)=F(x)$ for some $x\in[0,a]$.
If $a=\underline{r}$, there is nothing be proven.
Assume $a<\underline{r}$ so that $G$ is a strict MPC of $F$ over $[a,b]$
for some $b$. 
By the same logic of Step 4, it must be the case $b>\underline{r}$.
Since $\supp(G)\cap[0,\underline{r}]$ is connected, either 
$G$ is constant over $[a,\underline{r}]$ or
$(a,a+\e)\subseteq \supp(G)$ for some $\e>0$.
For the first case, from Lemma \ref{lemma:partial_stop_convex}, for $N$ large enough, 
$D(x;G)$ is either strictly convex or strictly concave over $(b-\e,b)$,
contradicting the linearity in Lemma \ref{appendix_lemma:dworczak_martini}.

For the latter, let 
$\alpha=\max(\supp(G)\cap[0,\underline{r}])$ and 
$\beta=\min(\supp(G)\cap[\underline{r},k_G])$. By Lemma \ref{appendix_lemma:dworczak_martini},
there exists some $\phi$ such that $\phi$ is linear over $[a,b]$. Note $[\alpha,\beta]\subseteq [a,b]$.
Over $[a,\alpha]$, $\phi(x)=G(x)^{n-1}$.
Since $\phi$ is convex, the kink must be increasing at $x=a$, i.e. $\phi'(a+)=(n-1)F(a)^{n-2}g(a+)\geq \phi'(a-)=(n-1)F(a)^{n-2}f(a)$. By the MPC constraint, we also obtain $g(a+)\leq f(a)$. Hence, $g(a+)=f(a)$.
Moreover, since $\phi$ is linear, the slope over $[\alpha, \beta]$ must coincide with $\phi'(a+)$, i.e. 

$$
\phi'(a+)=(n-1)F(a)^{n-2}f(a)=\frac{J_F(\alpha)(1-H(c_G(\beta)))}{\beta-\alpha}=\frac{D(\beta;G)-D(\alpha;G)}{\beta-\alpha}
$$
Rearranging, 
$$
\frac{f(a)}{1-H(c_G(\beta))}=\frac{J_F(\alpha)}{(n-1)F(a)^{n-2}}.
$$
The right hand side diverges as $n\rightarrow \infty$.
Such $G$ cannot be a limit equilibrium.

\subsection{Omitted Proofs of Appendix \ref{section:online_appendix_c}}

\paragraph{Proof of Proposition \ref{proposition:small_market}}
By Lemma \ref{lemma:continuity_atom}-(b), $G$ must have an atom at $\max(\supp(G))$
and be continuous over $[0,\max(supp(G)))$, proving (a).
Assume $G(\underline{r})>0$. Then, by Steps 1 and 2 of the proof of Lemma~\ref{lemma:full_info_bottom}, we have
$\min(\supp(G))=0$ and $\supp(G)\cap[0,\underline{r}]$ is connected.
Over this convex interval $\supp(G)\cap[0,\underline{r}]$,
$D(x;G)=\phi(x)$ is convex by Lemma \ref{appendix_lemma:dworczak_martini}.
We can partition $[0,\underline{r}]$ into intervals 
$0=x_0<\dots<x_{k-1}=\max(\supp(G))\cap[0,\underline{r}]$,
such that on each interval $[x_i, x_{i+1}]$ ($i=1,\dots, k-2$), either $D(x;G)$ is strictly convex, 
or linear. By Lemma~\ref{appendix_lemma:dworczak_martini}, in the former case where $D(x;G)$ is strictly convex, it must be $G=F$, whereas the latter case implies $D(x;G)=G^{n-1}$ is linear. This establishes part (b).
Now we prove (c).
\\ \noindent\textbf{Case 1}: Strictly Convex $H$.

\noindent
Assume $\supp(G)\cap[\underline{r},\max(\supp(G))\neq \emptyset$ and let $M:=\sup\Big(\supp(G)\cap[\underline{r},\max(\supp(G)))\Big)$. We first show that $M<k_G$. If this is the case, $G$ is flat over $[M,k_G)$, implying $D(x;G)=D(M;G)+J_G(M)(H(c_G(M))-H(c_G(x)))$ is strictly concave over this interval. However, it is impossible to have any convex $\phi(x)$ such that $\phi(x)=D(x;G)$ for all $x\in \supp(G)$, especially at $x=M$ and $x=k_G$. Hence, $\supp(G)\cap[\underline{r},\max(\supp(G))\neq \emptyset$.

Suppose instead $M=k_G$, i.e. $G(x)$ is continuous at $x=k_G$. Note $G$ must be a strict MPC of $F$ over some $(k_G-\e, k_G)$. But under the strict convexity of $H$ (i.e. $h'>0$), the second derivative satisfies:
\begin{align*}
    D''(x;G)&=(n-1)G^{n-3}\Big[
    (n-1)g^2H(c_G)+g'GH-2(1-G)Ggh(c_G)\Big]-(1-G)^2J_Gh'(c_G) \\ 
    &< (n-1)G^{n-3}\Big[(n-1)g^2H(c_G)+g'GH-2(1-G)Ggh(c_G)\Big]
\end{align*}
As $x\rightarrow k_G-$, $H(c_G(x))\rightarrow 0$, so the first two terms in the bracket vanishes to zero, whereas the last term remains weakly positive. Thus, $\lim_{x\rightarrow k_G-}D''(x;G)<0$, implying $D(x;G)$ is strictly concave over some interval $(k_G-\e,k_G)$. This again contradicts that $\phi=D(x;G)$ must be convex over $(k_G-\e, k_G)$ by Lemma~\ref{appendix_lemma:dworczak_martini}.

\noindent \textbf{Case 2, 3}: Strictly Concave $H$ or Linear $H$

\noindent Assume $\supp(G)\cap[\underline{r},\max(\supp(G))\neq \emptyset$ and define $m:=\min(\supp(G)\cap[\underline{r},\max(\supp(G))))$.
Since $x_{k-1}<\underline{r}$, we must have $m>\underline{r}$. Otherwise, $m=\underline{r}$ implies $D(x;G)$ is flat over $[x_{k-1},\underline{r}]$, which contradicts the existence of convex $\phi$ in Lemma~\ref{appendix_lemma:dworczak_martini}.

We consider 2 cases, $\partial_+G(m)>0$ and $\partial_+G(m)=0$. First, assume $\partial_+G(m)>0$. In this case, $D(x;G)$ has a strictly increasing kink at $x=m$.
$$
D'(m+;G)=
\underbrace{(n-1)G^{n-2}(m)\partial_+G(m)H(c_G(m))}_{D'(m-;G)}
+\underbrace{J_G(m)(1-G(m))h(c_G(m))}_{>0}
$$
Since $G$ is a strict MPC of $F$ over some interval $(m-\e, m+\e)$, Lemma~\ref{appendix_lemma:dworczak_martini} implies that the associated $\phi$ must be linear and satisfy $\phi(x)\geq D(x;G)$. But this is impossible, as $\phi$ would lie strictly below $D(x;G)$ just to the left of $m$ due to the kink. Contradiction.

Now suppose $\partial_+G(m)=0$. Then, $D'(m;G)=J_G(m)(1-G(m))h(c_G(m))$. To be consistent with $\phi(x)\geq D(x;G)$ and $\phi$ linear near $m$, we must have $D''(m-;G)\leq 0$. However, if $H$ is strictly concave, then $h'(c)<0$ implies 
$D''(m-)=-(1-G(m))^2J_G(m)h'(c_G(m))>0,$ a contradiction. 

Assume $H$ is linear.
From the linearity of $\phi$ over $[x_{k-1},m]$, we must have:
$$
\frac{D(m;G)-D(x_{k-1};G)}{m-x_{k-1}}
=D'(m;G)
$$
Substituting expressions and simplifying:
$$
\frac{J_G(m)(1-H(c_G(m)))}{m-x_{k-1}}
=J_G(m)(1-G(m))h(c_G(m))
$$
Using $H(c)=\frac{c}{\overline{c}}$ and $h(c)=\frac{1}{\overline{c}}$:
\begin{align*}
    \frac{1}{\overline{c}}(\overline{c}-c_G(m))&= \frac{1}{\overline{c}}(m-x_{k-1})(1-G(x_{k-1})) \\ 
    &=\frac{1}{\overline{c}}(c_G(x_{k-1})-c_G(m)))
\end{align*}
The last equality follows from the fact that $c_G(x)$ is decreasing and linear with slope $(1-G(x_{k-1}))$ over $[x_{k-1},m]$. This implies $c_G(x_{k-1}) = \overline{c}$, which contradicts $x_{k-1} < \underline{r}$ (since $c_G(x)$ is decreasing). Contradiction.

Thus, in all cases, we reach a contradiction. Therefore, $\supp(G)\cap[\underline{r},\max(\supp(G))= \emptyset$.

\qed
\paragraph{Proof of Proposition \ref{proposition:small_market_binary}}

Note since $F$ is a binary distribution supported at $\{0,1\}$ with mean $\mu$, choosing a distribution of posterior means $G\in \Delta([0,1])$ is equivalent to choosing a Bayes-plausible distribution over posteriors,$\mu=\mathbb{E}_{G}[x]$, as in \textcite{kamenica2011bayesian}. By their result, any equilibrium must satisfy 
$cav(D(x;G))=D(x;G)$ for all $x\in \supp(G)$, where $cav(D)$ denotes the concave closure of $D$.
Furthermore, by Lemma~\ref{appendix_lemma:dworczak_martini}, $\phi(x)=D(x;G)$ must be convex for all $x\in \supp(G)$. These two conditions together imply that $\phi(x)=D(x;G)$ must be linear over all $x\in \supp(G)$. 

By Lemma~\ref{lemma:continuity_atom}, $G$ must have an atom at $\max(\supp(G))$ and continuous over $[0,\max(\supp(G)))$. Moreover, if $G(\underline{r})>0$, then by Steps 1 and 2 of the  proof in Lemma \ref{lemma:full_info_bottom}, $\min(\supp(G))=0$ and $\supp(G)\cap[0,\underline{r}]$ must be connected. Hence $G^{n-1}$ must be linear over $\supp(G)\cap[0,\underline{r}]$.

Let $a:=\max(\supp(G)\cap[0,\underline{r}])$. Suppose $a=\underline{r}$. Then,
$$
D'(a+;G)=\underbrace{(n-1)G(a)^{n-2}g(a)}_{=D'(a-;G)}+(1-G(a))h(c_G(a))J_G(a),
$$
implying $D(x;G)$ exhibits an upward kink at $x=\underline{r}$. Therefore, $\underline{r}\notin\{x:cav(D(x;G))=D(x;G)\}$, implying $a<\underline{r}$. 
This proves part (a). Proof of part (b) follows directly from the argument in Proposition \ref{proposition:small_market}.








\qed

\end{document}